\documentclass[acmsmall,screen,nonacm,natbib=false]{acmart}

\setcopyright{none}

\RequirePackage[
  datamodel=acmdatamodel,
  style=acmnumeric,
  sortcites=true,
  backref=true,
  backrefstyle=three,
  uniquename=false,
  uniquelist=false
  ]{biblatex}
\usepackage{software-biblatex}

\DeclareFieldFormat{titlecase}{#1}

\usepackage{fontawesome5}
\usepackage{mathtools}
\usepackage{multicol}
\usepackage{bcprules}
\usepackage{commands}
\usepackage{xspace}
\usepackage{enumitem}
\newcommand{\capybara}{\textsf{Capybara}}
\newcommand{\corecapybara}{\textsf{CoreCapybara}}
\newcommand{\capless}{\textsf{Capless}}

\newcommand{\ccformal}{\textsf{CC}\ensuremath{_{<:\Box}}\xspace}

\usepackage{mdframed}

\usepackage{listings}

\newcommand{\btHL}[1]{\colorbox{yellow!20}{#1}}
\lstdefinelanguage{dotty}{
  basicstyle=\footnotesize\ttfamily,
  keywords={erased, val, var, if, then, in, handle,
    return, def, match, case, new, type, trait, exists,
     package, object, given, eff, mut, iso, imm, ro, mv, shrd, rd,
     pretype, class, extends, extension, infix, else,
     box, unbox, try, catch, import, throw, throws, using,
		 use, cap, this, consume, sealed, override, private, protected,
		 transparent, inline, opaque, open, tracked, lazy, enum, abstract, final, any, fresh, update},
  keywordstyle=\bfseries\color{magenta!80!black},
  morekeywords=[2]{@use, @consume, @assumeSafe, @untrackedCaptures},
  keywordstyle=[2]{\color{teal}},
  sensitive=true,
  comment=[l]{//},
  morecomment=[s]{/*}{*/},
  commentstyle=\color{green!40!black},
  stringstyle=\color{green!60!black},
  morestring=[b]',
  morestring=[b]",
  moredelim=**[is][\btHL]{`}{`},
	columns=fullflexible,
  alsoletter={@},
}
\lstdefinelanguage{rust}{
  basicstyle=\footnotesize\ttfamily,
  keywords={fn, let, mut, move, match, if, else, for, while, loop, return,
    struct, enum, impl, trait, use, pub, mod, ref, as, where, dyn, in,
    self, Self, crate, const, static, type, unsafe, async, await, break,
    continue, box},
  keywordstyle=\bfseries\color{magenta!80!black},
  sensitive=true,
  comment=[l]{//},
  morecomment=[s]{/*}{*/},
  commentstyle=\color{green!40!black},
  morestring=[b]",
  stringstyle=\color{green!60!black},
  columns=fullflexible,
}

\lstnewenvironment{code}[1][]
  {\lstset{language=dotty,
    basicstyle=\footnotesize\ttfamily,
    backgroundcolor=\color{gray!4},
    frame=trbL, framerule=0.0pt, rulecolor=\color{blue!20!gray!40}, rulesep=1.5pt, rulesepcolor=\color{blue!20!gray!40},
    framesep=0pt,
    framexleftmargin=5pt, framexrightmargin=1pt,
    framextopmargin=3pt, framexbottommargin=3pt,
    xleftmargin=6pt,
    columns=flexible, keepspaces=true, showstringspaces=false,
    extendedchars=true, breaklines=false,
    aboveskip=3pt, belowskip=1pt, #1}}
  {}

\lstnewenvironment{rustcode}
  {\lstset{language=rust,
    basicstyle=\footnotesize\ttfamily,
    backgroundcolor=\color{gray!4},
    frame=trbL, framerule=0.0pt, rulecolor=\color{blue!20!gray!40}, rulesep=1.5pt, rulesepcolor=\color{blue!20!gray!40},
    framesep=0pt,
    framexleftmargin=5pt, framexrightmargin=1pt,
    framextopmargin=3pt, framexbottommargin=3pt,
    xleftmargin=6pt,
    columns=flexible, keepspaces=true, showstringspaces=false,
    extendedchars=true, breaklines=false,
    aboveskip=3pt, belowskip=1pt}}
  {}

\lstnewenvironment{codelinenos}
  {\lstset{language=dotty,
    basicstyle=\footnotesize\ttfamily,
    backgroundcolor=\color{gray!4},
    frame=trbL, framerule=0.0pt, rulecolor=\color{blue!20!gray!40}, rulesep=1.5pt, rulesepcolor=\color{blue!20!gray!40},
    framesep=0pt,
    framexleftmargin=18pt, framexrightmargin=1pt,
    framextopmargin=3pt, framexbottommargin=3pt,
    xleftmargin=6pt,
    numbers=left, numberstyle=\tiny\color{gray}, numbersep=6pt,
    columns=flexible, keepspaces=true, showstringspaces=false,
    extendedchars=true, breaklines=false,
    aboveskip=3pt, belowskip=1pt}}
  {}

\usepackage{cleveref}
\crefname{appendix}{appendix}{appendices}
\Crefname{appendix}{Appendix}{Appendices}
\usepackage{stmaryrd}

\usepackage{wrapfig}

\newcommand{\mknote}[3]{}

\newcommand{\CLAUDE}[1]{}

\usepackage{yfonts,lettrine}

\AtBeginDocument{\setlength{\DefaultFindent}{0.5em}}
\makeatletter%
\begin{document}

\title[System \capybara{}: Tracking Capabilities for Separation and Freshness]{System \capybara{}: Tracking Capabilities for Separation and Freshness (Extended Version)}

\author{Yichen Xu}
\orcid{0000-0003-2089-6767}
\email{yichen.xu@epfl.ch}
\affiliation{%
  \institution{EPFL}
  \country{Switzerland}
}
\author{Oliver Bra\v{c}evac}
\orcid{0000-0003-3569-4869}
\email{oliver.bracevac@epfl.ch}
\affiliation{%
  \institution{EPFL}
  \country{Switzerland}
}
\author{Cao Nguyen Pham}
\orcid{0009-0005-2543-3309}
\email{nguyen.pham@epfl.ch}
\affiliation{%
  \institution{EPFL}
  \country{Switzerland}
}
\author{Yaoyu Zhao}
\orcid{0000-0003-2257-1413}
\email{yaoyu.zhao@epfl.ch}
\affiliation{%
  \institution{EPFL}
  \country{Switzerland}
}
\author{Martin Odersky}
\orcid{0009-0005-3923-8993}
\email{martin.odersky@epfl.ch}
\affiliation{%
  \institution{EPFL}
  \country{Switzerland}
}

\renewcommand{\shortauthors}{Xu et al.}

\begin{abstract}
Substructural type systems give strong static control over aliasing.
Examples include
uniqueness,
separation,
and borrowing.
How can such control be brought to established languages
whose programming models rely on higher-order abstraction,
unrestricted aliasing,
and pervasive sharing?
We study this problem in the context of Scala.
We show how to retrofit these guarantees \emph{selectively} instead of globally:
ordinary code keeps Scala's usual aliasing discipline,
while stronger guarantees can be enforced where they matter.

Our starting point is Scala's capture checking,
whose treatment of capabilities is inspired by the object-capability tradition:
capabilities are ordinary values,
and capture sets record,
in a value's type,
which capabilities the value may use.
We develop System~\capybara{},
which adds a selective alias-control layer to this mechanism.
By tracking separation, consumption, freshness, and read-only access for capabilities,
\capybara{} recovers key reasoning principles from substructural
and ownership-based disciplines
without global invariants.

We give a type-preserving translation from the surface calculus \capybara{}
to \corecapybara{},
a core calculus extending System \capless{},
the earlier foundation for capture checking.
The translation uses quantifiers for capture polymorphism and freshness,
and constraint-indexed modal types for separation.
We prove a semantic soundness result for the core calculus in Lean~4
and derive type safety,
memory safety (no use-after-free or double-free),
immutability of read-only computations,
and data-race freedom for well-typed programs.

Finally, we implement Scala~3's new separation checker,
which brings higher-order separation reasoning about effects, capabilities,
and resources to ordinary Scala, including fearless concurrency.
\end{abstract}

\maketitle

\section{Introduction}\label{sec:intro}

Modern safety guarantees often rest on controlling aliasing.
Memory safety, deterministic resource management, and data-race freedom
all require knowing when computations may touch
the same mutable resource.
This is hard in languages with unrestricted aliasing:
objects cross abstraction boundaries,
closures capture mutable state,
and higher-order functions defer effects until invocation.
The style is expressive,
but mutability undermines local reasoning:
a write through one reference
can silently change what another reads,
the familiar ``spooky action at a distance''.
With deallocation or parallelism,
the same aliasing problem may lead to dangling pointers, double-frees, or data races.

Rust~\cite{DBLP:journals/cacm/JungJKD21,rust} shows the value of static alias control.
Its ownership and borrowing discipline supports safe resource management,
memory safety without a garbage collector,
and \emph{fearless concurrency}~\cite{rustbook-fearless-concurrency}.
But Rust obtains these guarantees by making exclusivity
part of the language's programming discipline.
That design is hard to retrofit into languages
where sharing is pervasive.

The tension appears even in a small Rust example.
A mutable borrow and a shared borrow of the same value
cannot both be live at once.
Here a closure logs the vector's length,
with a push between two calls to it:
\begin{rustcode}
let mut v = vec![1, 2, 3];
let log_length = || println!("{}", v.len());  // borrows v
log_length();
v.push(4);     // error: cannot borrow `v` as mutable
log_length();  // ...because this later read still borrows it
\end{rustcode}
The closure \lstinline|log_length|
retains a shared borrow of \lstinline|v| until its last call,
so the push in between, which mutably borrows the vector, is rejected.
The operations are sequenced:
no data race or dangling reference can arise.
The borrow checker rejects the program because of the alias the closure
retains, not because the operations actually interfere.
Closures that capture mutable state are a common setting
in which shared-xor-mutable reasoning rejects safe higher-order patterns.
Rust programmers work around this %
by restructuring the code
or by switching to reference-counted cells with dynamically checked borrows~\cite{rust},
stepping outside the static discipline.

The need for alias control is not specific to systems programming.
In higher-level languages, resources such as file handles, buffers, and network connections
still call for exclusive access and affine ownership.
The goal of this paper is to bring such guarantees to a real-world language
where sharing remains the default: Scala.
Programmers should be able to require separation or affine use
where a resource-sensitive API needs it,
without imposing an exclusivity discipline on ordinary code.
We build on
\emph{capturing types}~\cite{DBLP:journals/toplas/BoruchGruszeckiOLLB23,whatisinthebox} (or capture checking)
which attach to a value's type a \emph{capture set}
recording the capabilities the value may use.
In this setting, capabilities are ordinary values:
the effects code may perform are bounded by the capabilities it
holds~\cite{DBLP:conf/icfem/CraigPGA18},
and capture sets statically track where those capabilities flow.
A file handle, a mutable buffer, or a network connection
can be treated as a capability;
a closure's type records the capabilities it may exercise.

Capturing types have proven non-invasive and flexible.
For higher-order code,
effect polymorphism can stay implicit:
closures over mutable state can be tracked
without changing ordinary higher-order signatures.
Capture checking is implemented in the Scala~3 compiler
and has been applied to large parts of the Scala standard
library~\cite{whatisinthebox} retroactively and non-invasively.

What they do not yet provide is \emph{exclusive} use.
Capture sets record \emph{which} capabilities a value may use,
but not whether access is read-only or exclusive,
or whether a capability has been consumed.
For a parallel combinator, capture sets describe what each closure may touch
but cannot prevent both from writing to the same buffer.
Likewise, they can record use of a file handle,
but not that it has been consumed and its old aliases disabled.

We present System~\capybara{},\footnote{Capybaras are known for living peacefully alongside others.
The name is a phonetic pun for ``capabilities and borrowing''.}
a calculus built around tracked capabilities.
Drawing on linearity, uniqueness, ownership, and borrowing
\cite{DBLP:conf/ifip2/Wadler90,DBLP:journals/mscs/BarendsenS96,DBLP:conf/esop/MarshallVO22},
\capybara{} adds selective alias control on top of Scala's capture checking.
\emph{Separation} requires non-interfering memory footprints,
for example when computations are passed to a parallel combinator.
\emph{Read-only capabilities} distinguish reads from writes:
read-only uses of the same capability may overlap,
and separation rejects an overlap only when one side may write or
consume the shared capability.
\emph{Consume} models affinity and ownership transfer:
once a capability is consumed,
all aliases covered by it become unusable,
so a consuming API can treat the capability as fresh.
Aliasing remains permitted and tracked by default;
separation, read-only access, and consumption are enforced only where
an API asks for them.
\Needspace{7\baselineskip}
The Rust example above, ported to \capybara{} with a buffer in place of the vector,
type-checks in sequence and is rejected when run in parallel with a write:
\begin{code}
val v = new Buffer(1, 2, 3)
val logSize = () => println(v.size)     // : () ->{v.rd} Unit
logSize(); v.push(4); logSize()         // ok: sequenced uses
runParallel(logSize, () => v.push(4))   // error: write overlaps read
\end{code}
The technical challenge is to make these local requirements sound
in a higher-order language with subtyping
and both implicit and explicit capture polymorphism.
We give a type-preserving translation from the surface calculus
\capybara{} to \corecapybara{},
a smaller core calculus on which we develop the metatheory,
building on System \capless{}~\cite{whatisinthebox}.
In the core,
implicit and explicit capture polymorphism become universal quantification,
consumption and freshness become existential capture,
and separation obligations become \emph{modal types}
indexed by explicit constraints.
\corecapybara{} also guides our Scala 3 implementation:
a separation checker that layers separation, consume, and read-only capabilities
on the existing capture checker.

Several practical languages now expose substructural information in
the type system:
Linear Haskell through linear arrows,
Swift and Mojo through ownership features,
and OxCaml through modes
\cite{DBLP:journals/pacmpl/BernardyBNJS18,swift-noncopyable,swift-borrowing,mojo-ownership,oxidizingocaml}.
OxCaml uses modal types for memory management, data-race freedom,
and effects~\cite{DBLP:journals/pacmpl/GeorgesPEWDECPD25,tang25modal};
in \capybara{}, modal types carry separation constraints over
tracked capabilities.

Previous work on capture checking relied on syntactic type soundness
\cite{DBLP:journals/toplas/BoruchGruszeckiOLLB23,whatisinthebox}.
For \capybara{}, we instead give a semantic model of \corecapybara{},
following the logical approach of Timany et al.
\cite{DBLP:journals/jacm/TimanyKDB24}.
In the model, a capture set denotes a runtime footprint:
the locations a term may read, write, or consume.
Semantic soundness states that well-typed terms produce only traces
authorized by their footprints.
This theorem, together with confluence of the reduction semantics,
is the basis for the formal claims summarized below.

This paper makes the following contributions.

\begin{itemize}[leftmargin=1.4em,labelsep=0.45em,itemsep=0.2ex,topsep=0.4ex]
\item We introduce System~\capybara{} by example,
  showing how tracked capabilities support local reasoning
  about aliases and resources in ordinary Scala code
  (\Cref{sec:informal}).
\item We formalize \capybara{},
  a surface calculus that retrofits selective substructural reasoning
  onto capture checking
  (\Cref{sec:formal}).
\item We define \corecapybara{}, a compact core extending System~\capless{}~\cite{whatisinthebox},
  and a type-preserving translation from \capybara{} to the core
  (\Cref{sec:core}; details in \Cref{app:translation}).
\item We prove semantic soundness for \corecapybara{},
  mechanized in Lean~4~\cite{DBLP:conf/cade/Moura021},
  deriving type safety,
  memory safety (no use-after-free or double-free),
  separation,
  and read-only immutability;
  confluence gives data-race freedom for well-typed programs
  (\Cref{sec:metatheory}).
\item We describe the implementation of the new Scala~3 separation checker
  layered on top of capture checking,
  and exercise it on resource-sensitive and concurrent case studies
  (\Cref{sec:case-studies}).
\end{itemize}
\Cref{sec:discussion,sec:related-work,sec:conclusion} discuss limitations and
design choices, related work, and conclusions.
\section{\capybara{} by Example}\label{sec:informal}

\subsection{A Brief Introduction to Capturing Types}

System~\capybara{} extends 
\emph{capturing types}~\cite{DBLP:journals/toplas/BoruchGruszeckiOLLB23,whatisinthebox},
so we introduce them first.
Capturing types support safe resource and effect programming
by recording in a value's type the capabilities it may use.
Resources and effects, such as file handles, mutable state, or control effects,
are modeled as \emph{capabilities}.
A \emph{capture set} lists them:
\lstinline[mathescape]|T^$\set{x_1,\dots,x_n}$| denotes a \lstinline|T| value
that captures at most $x_1,\dots,x_n$.

For example, given a mutable string buffer \lstinline|buf|,
the following is a closure that writes to it:
\begin{code}
() => buf.write("Hello, world!")
\end{code}
The closure has the capturing type
\lstinline|(() -> Unit)^{buf}|,
often written \lstinline|() ->{buf} Unit|.
The shape type \lstinline|() -> Unit|
says that the function takes no argument and returns unit.
The capture set \lstinline|{buf}| makes it explicit that
invoking the function may use the capability \lstinline|buf|, 
and so may write to the buffer.
A type whose capture set is non-empty, like \lstinline|() ->{buf} Unit|, is called \emph{impure},
while one with an empty capture set, like \lstinline|() -> Unit|, is called \emph{pure}.

A type system that tracks effects has to support
\emph{effect polymorphism}.
A higher-order function such as list \lstinline|map|
should work whether the mapping operation it receives is pure or effectful.
A naive treatment forces every higher-order signature to take an
extra effect parameter,
prohibitively intrusive for existing code~\cite{DBLP:journals/toplas/BoruchGruszeckiOLLB23}.
Capturing types handle this with \emph{lightweight effect polymorphism}~\cite{whatisinthebox,DBLP:journals/toplas/BoruchGruszeckiOLLB23}.
List \lstinline|map| keeps its original signature,
yet is effect-polymorphic:
\begin{code}
trait List[+A]:
  def map[B](f: A => B): List[B]
\end{code}
Under the hood, lightweight effect polymorphism rests on
\emph{subcapturing} and \emph{root capabilities}.

Subcapturing is like subtyping but over capture sets,
written \lstinline|{x1,...,xn} <: {y1,...,ym}|.
At its core it is the subset relation:
it holds whenever the left-hand set is a subset of the right.
It also relates a variable to what its type captures.
Name the closure from before,
\begin{code}
val sayHi = () => buf.write("Hello, world!")
\end{code}
so that \lstinline|sayHi| has type \lstinline|() ->{buf} Unit|.
Then \lstinline|{sayHi} <: {buf}|.
Subsetting and transitivity yield further relations, for instance
\lstinline|{} <: {sayHi} <: {buf, io}|.

The second ingredient is the \emph{root capability} \lstinline|any|.
It subsumes \emph{any} capability.
Therefore, \lstinline|{any}| is the top capture set:
every capture set is a subcapture of it.
Return to \lstinline|def map[B](f: A => B): List[B]|.
Its parameter type \lstinline|A => B| is shorthand for \lstinline|A ->{any} B|,
a function capturing \lstinline|any|.
By subcapturing,
this signature accepts functions of any captures,
pure or effectful,
which makes \lstinline|map| \mbox{effect-polymorphic}.

Note that \lstinline|map|'s result type \lstinline|List[B]| stays pure:
\lstinline|map| is strict,
applying \lstinline|f| during the call and retaining no reference to
it.
A lazy operation is different.
Consider a function returning an iterator over a buffer's lines:
\begin{code}
def linesIterator(buf: Buffer^): Iterator[String]^{buf} = ...
\end{code}
Here \lstinline|Buffer^| abbreviates \lstinline|Buffer^{any}|,
so the parameter accepts any buffer.
The iterator reads lines from \lstinline|buf| on demand,
so it retains the buffer capability.
Its type \lstinline|Iterator[String]^{buf}| captures \lstinline|buf|.
Using the iterator therefore produces mutation effects,
which the capture set makes explicit.

\subsection{Separating Anys From Anys}\label{sec:informal:separation}

Capturing types track which capabilities a value uses,
not whether they are \emph{separate} or \emph{fresh}.
Every capability subcaptures \lstinline|any|,
so the types alone cannot show two capabilities disjoint.

This matters wherever non-interference does.
For instance, consider a function that runs two operations in parallel:
\begin{code}
def runParallel(op1: () => Unit, op2: () => Unit): Unit = ...
\end{code}
Capturing types record what the two argument functions may use,
but enforce no separation constraint between them.
Nothing prevents passing two closures writing to the same buffer:
\begin{code}
runParallel(() => buf.write("Hello"), () => buf.write("World"))
\end{code}
Each closure captures the mutable buffer \lstinline|buf| and writes to it,
so running the two in parallel is a data race.
Ruling out such races calls for tracking separation
in capturing types.

This is the gap \capybara{} closes.
On top of capturing types it adds \emph{separation checking}.
It understands \lstinline|any| differently:
each \lstinline|any| in a function's parameters
stands for capabilities separate
from those of the other parameters
and of the function itself.
More concretely, \capybara{} reads the two \lstinline|any|s in \lstinline|runParallel| as distinct:
\begin{code}[mathescape]
def runParallel(op1: () ->{any$_1$} Unit, op2: () ->{any$_2$} Unit): Unit = ...
\end{code}
The subscripts are for exposition:
they show how \capybara{} reads the signature internally;
the source still writes a plain \lstinline|any|.
\capybara{} checks that
the capabilities subsumed by the two \lstinline|any|s
do not interfere.
The earlier call violates this:
both closures write to \lstinline|buf|,
so \lstinline[mathescape]|any$_1$| and \lstinline[mathescape]|any$_2$| each subsume \lstinline|{buf}|
and interfere at \lstinline|buf|.
\capybara{} rejects the call with a separation error.

When two operations act on separate resources,
\capybara{} accepts running them in parallel:
\begin{code}
def alloc(): Buffer^{fresh} = new Buffer
val f1 = alloc(); val f2 = alloc()
runParallel(() => println(f1.read()), () => f2.write("World"))
\end{code}
Each call to \lstinline|alloc| returns a \emph{fresh} capability
(a notion we return to in \Cref{sec:informal:consume}),
distinct from every other,
so \lstinline|{f1}| and \lstinline|{f2}| are separate.
Here \lstinline[mathescape]|any$_1$| subsumes \lstinline|{f1}|
and \lstinline[mathescape]|any$_2$| subsumes \lstinline|{f2}|;
the two do not interfere, so the check succeeds,
and the parallel computation is indeed safe at runtime.

The obligation paid at call sites,
that distinct \lstinline|any|s stand for separate capabilities,
becomes a guarantee for functions,
which may assume distinct \lstinline|any|s separate from each other.
Consider a function that receives two buffers and updates them in parallel:
\begin{code}[mathescape]
def updateBuffers(buf1: Buffer^{any$_1$}, buf2: Buffer^{any$_2$}): Unit =
  runParallel(() => buf1.append("Hello, "), () => buf2.append("World!"))
\end{code}
Calling \lstinline|runParallel| in the body requires
\lstinline|buf1| and \lstinline|buf2| to be separate.
\capybara{} accepts the call because the two buffers capture distinct \lstinline|any|s
(\lstinline[mathescape]|any$_1$| and \lstinline[mathescape]|any$_2$|).
We call these \lstinline|any|s \emph{roots}.
They are treated as separate from every other root in scope.
The \lstinline|fresh| of \lstinline|alloc| is a root as well,
a new one at each call,
which is why \lstinline|f1| and \lstinline|f2| never alias.
Roots are the minimal units of separation and freshness reasoning:
two capture sets are separate when they trace back to non-interfering roots.
As \Cref{sec:informal:core} shows,
roots correspond to the quantified capture variables of the core
calculus.

\capybara{} provides statically checked, data-race-free parallelism
without giving up flexibility:
a capability may be shared,
provided the sharing is recorded in the types,
and exclusivity is demanded only where safety requires it,
as at a parallel call.

The following example adapts the Rust example of \Cref{sec:intro}
into \capybara{}.
The closure \lstinline|logSize| reads \lstinline|v.size|,
so it captures \lstinline|v| read-only and has type \lstinline|() ->{v.rd} Unit|:
\begin{code}
val v = new Buffer(1, 2, 3)
val logSize = () => println(v.size)  // : () ->{v.rd} Unit
logSize()
v.push(4)
logSize()  // ok
runParallel(logSize, () => v.push(4))  // error
\end{code}
Calling \lstinline|logSize|, pushing to \lstinline|v|,
and calling \lstinline|logSize| again all type-check:
the three run in sequence, and nothing interferes.
Capturing types record that \lstinline|logSize|
has a reference to \lstinline|v|,
but do not reject the first two uses.
Only the last line is rejected,
where \lstinline|logSize| and the write to \lstinline|v|
run in parallel and form a read/write conflict at \lstinline|v|:
separation is enforced when concurrency makes it matter.
Where Rust rejects the program for the alias that
\lstinline|logSize| retains,
\capybara{} rejects only the interfering use.
This permissiveness matters for adoption:
existing Scala code aliases pervasively,
and a discipline that strictly restricts aliasing
would be infeasible to adopt.

\subsection{Read-Only Capabilities}\label{sec:informal:readonly}

Consider the following example,
where two parallel operations read from the same buffer:
\begin{code}
runParallel(() => println(buf.size), () => println(buf(0)))
\end{code}
Both closures use \lstinline|buf|,
but only for reading,
so they are data-race-free.

To accept such programs,
\capybara{} introduces \emph{read-only capabilities}.
A capturing type carries
a \emph{permission}
that limits what its value may be used for.
It is written as a prefix qualifier on the type, as in
\lstinline|ro Buffer^{...}|.
The permission \lstinline|ro| restricts a reference to \lstinline|Buffer|'s
read-only operations:
\begin{code}
val f1r: ro Buffer^ = f1
f1r.size           // ok, a non-mutating operation
f1r.write("Hello") // error, write is not permitted through f1r
\end{code}
\lstinline|f1r| is a read-only view of \lstinline|f1|,
exposing its read-only methods but not its mutating ones.
By default, a capability has read-write permission:
every \lstinline|Buffer| type seen so far omitted the qualifier.
A read-write reference (\lstinline|Buffer^|)
may be used wherever a read-only one is expected (\lstinline|ro Buffer^|),
but not the reverse.
In our Scala 3 implementation, the qualifier is never written
but inferred (\Cref{sec:impl}).

The permission on a type constrains what a value may do.
Capture sets instead record how a capability is \emph{used}.
Concretely,
a capture set tags each of its elements with
an \emph{access mode},
which is either read-only or read-write.
We write a read-only use with the suffix \lstinline|.rd|.
For instance,
the closure \lstinline|() => f1.size| calls
a read operation of \lstinline|f1|,
so it has the capturing type \lstinline|() ->{f1.rd} Int|.

Separation checks treat
two read-only uses of the same capability as non-interfering,
a principle familiar from fractional
permissions~\cite{DBLP:conf/sas/Boyland03},
so the snippet type-checks.
A read-only use concurrent with a read-write one still conflicts.

\subsection{Always Consume Your Capabilities Fresh}\label{sec:informal:consume}

Recall the two calls to \lstinline|alloc| from before:
each returned a \emph{fresh} capability.
Freshness is a global property.
A fresh capability is aliased by nothing,
so it is separate from anything else in the program.
\capybara{} records this guarantee in the capture set
with the special element \lstinline|fresh|,
as in \lstinline|Buffer^{fresh}|.
The element appears only in the result type of a function,
as in \lstinline|alloc|'s signature \lstinline|() -> Buffer^{fresh}|:
it states that every call returns a fresh capability.

There are two ways to obtain a fresh capability.
The first is to create one outright:
a constructor can return its newly allocated resource as fresh,
since nothing in the program could alias it yet.

The second is to \emph{consume} an existing capability.
A parameter marked \lstinline|consume| disables every other reference to its
argument, so a reference returned by the call may be treated as fresh.
\begin{code}
// Merging two buffers is more efficient if we can reuse the allocated memory
def merge(consume a: Buffer^, consume b: Buffer^): Buffer^{fresh} =
  if a.size <= b.freeSize then b.prepend(a.toSeq)
  else a.append(b.toSeq)

val f1 = alloc()
val f2 = alloc()
val fn = () => f2.write("hello")
val f = merge(f1, f2)
f1.read()  // error, f1 was consumed
fn()       // error, f2 was consumed
\end{code}
The call \lstinline|merge(f1, f2)| consumes both buffers,
so any later use of \lstinline|f1| or \lstinline|fn| is rejected.
With \lstinline|f1| and \lstinline|f2| disabled,
nothing can reach the buffers through them,
so the capability returned as \lstinline|f| is fresh:
the implementation of \lstinline|merge| can freely assume that no other aliases
exist for its two arguments, and return one of them
exactly as if it had just been created.

Consuming a capability is allowed only when it is fresh.
For instance,
a parameter typed \lstinline|Buffer^|
is \emph{separate} but not fresh: 
the caller may still hold other references,
so consuming it is rejected.
\begin{code}
def buildStr(acc: Buffer^): Unit =
  acc.append("Hello, ")        // ok: reading and writing acc is allowed
  val res = merge(acc, alloc())  // error: acc cannot be consumed
\end{code}
The parameter grants full read-write access, so the append succeeds.
But \lstinline|merge| would consume \lstinline|acc|,
and \lstinline|buildStr| received \lstinline|acc| as an argument
that other code may still alias.
Only a capability the function owns,
one freshly created or passed under \lstinline|consume|,
may be passed to \lstinline|merge|.

The three disciplines of this section are analogous to
Rust's three kinds of references.
An owned value is a \emph{fresh} capability:
\lstinline|f| returned by \lstinline|merge| above owns its buffer.
Only owned capabilities may be consumed,
which is the counterpart of a \emph{move}.
A mutable borrow \lstinline|&mut Buffer| is a tracked read-write alias,
e.g. \lstinline|g: Buffer^{f}|:
it may read and write the buffer but not consume it,
and it conflicts with every concurrent use of \lstinline|f|.
A shared borrow \lstinline|&Buffer| is a read-only view,
e.g. \lstinline|g: ro Buffer^{f.rd}|:
it may overlap with other reads, but not with writes.
The disciplines differ in how exclusivity is enforced:
Rust enforces shared-xor-mutable throughout each borrow's lifetime;
\capybara{} tracks aliases and demands exclusivity only
at a \lstinline|consume| or a separation check.

\subsection{Capybara under the Hood}\label{sec:informal:core}

The language presented so far is a surface language
that translates to a core calculus, \corecapybara{}, which gives
\lstinline|any|, \lstinline|ro|, \lstinline|consume|, and \lstinline|fresh|
an explicit semantics.
We preview the translation informally.

In the core, the translation of \lstinline|runParallel| uses capture
quantifiers for root capabilities and a \emph{modal type} for separation,
in the style of
contextual modal type theory~\cite{DBLP:journals/tocl/NanevskiPP08}:
\begin{code}[mathescape]
def runParallel[c$_1$, c$_2$](op1: () ->{c$_1$} Unit, op2: () ->{c$_2$} Unit): [{c$_1$} $\bowtie$ {c$_2$}] Unit
\end{code}
The subscripted \lstinline[mathescape]|any$_1$| and
\lstinline[mathescape]|any$_2$| of \Cref{sec:informal:separation} are now
explicit capture parameters,
which a call instantiates with the arguments' actual capture sets.
The result is a modal type $[\Psi]\,T$, a $T$ usable only under the
separation constraint $\Psi$, where $C_1 \bowtie C_2$ constrains two
capture sets to be separate;
to run the function, the caller must prove
\lstinline[mathescape]|{c$_1$} $\bowtie$ {c$_2$}|.
Dually, a \lstinline|fresh| in a result type is existentially quantified:
the result type of \lstinline|alloc| translates to
$\exists c.\,$\lstinline|Buffer^{c}|,
a buffer using some fresh capability,
and binding the result unpacks the package,
introducing a capture variable that nothing else in scope can mention.
\Cref{app:tr:example} gives the full translation.

A \lstinline|consume| parameter receives an owned capability,
so its \lstinline|any| is quantified existentially rather than universally:
the function translates to a \emph{consumer},
whose argument arrives as a package and is unpacked on entry.
The signature of the \lstinline|merge| function in \Cref{sec:informal:consume} becomes
\begin{code}[mathescape]
def merge(a: $\exists c_1.\,$Buffer^{c$_1$}, b: $\exists c_2.\,$Buffer^{c$_2$}): $\exists c.\,$Buffer^{c}
\end{code}
The call site types each argument at \lstinline|Buffer^{fresh}| and
packs it with its capability:
\lstinline|merge(f1, f2)| translates to
\lstinline[mathescape]|merge($\langle${f1}, f1$\rangle$, $\langle${f2}, f2$\rangle$)|.
Packing \emph{consumes} the witness,
which disables \lstinline|f1| and \lstinline|f2|.
\section{Formalizing \capybara{}}\label{sec:formal}

This section formalizes \capybara{}, the surface language of
\Cref{sec:informal}. Its core calculus \corecapybara{}, used for the
metatheory, follows in \Cref{sec:core}.

\Cref{fig:syntax} gives the abstract syntax of \capybara{}.
Its static semantics is split across three figures.
\Cref{fig:all-typing} presents the main judgments:
subcapturing, subbounding, term typing, and subtyping.
\Cref{fig:capture-typing} gives the capability-specific rules:
kinding, separation, and sequential composition.
\Cref{fig:state-typing} types the memory primitives, parallel
composition, and conditionals (\Cref{sec:formal:state}).

\begin{wide-rules}\noindent
	{\small\begin{multicols}{2}\noindent
		\begin{flalign*}
			x,\,y,\,z         \tag*{\textbf{Variable}}\\
			X                 \tag*{\textbf{Type Variable}}\\
			c,\,\ANY,\,\FRESH                 \tag*{\textbf{Capture Variable}}\\
			s,\,t,\,u\coloneqq\ &           \tag*{\textbf{Term}}\\
			&x                                      \tag*{variable}\\
			&v                              \tag*{value}\\
			&x\,y                              \tag*{app.}\\
			&x[S]                              \tag*{type app.}\\
			&x[C]                              \tag*{capture app.}\\
			&\ALLOC x \,\mid\, \READ x \,\mid\, x:=y \,\mid\, \DEALLOC x  \tag*{memory primitives}\\
			&\LET x = t \IN u                  \tag*{let}\\
			&\PAR t_1, t_2                  \tag*{parallel}\\
			&\IF x \THEN t_1 \ELSE t_2         \tag*{cond.}\\
			v\coloneqq\ &           \tag*{\textbf{Value}}\\
			&\lambda(\alpha\,x: T)t \,\mid\, \lambda[X<:S]t \,\mid\, \lambda[c<:B]t   \tag*{abstractions}\\
			&()\,\mid\,\TTRUE\,\mid\,\TFALSE    \tag*{constants}\\
		\end{flalign*}
		\begin{flalign*}
			a\coloneqq\ &x\,\mid\,v           \tag*{\textbf{Answer}}\\
			m\coloneqq\ &\epsilon\,\mid\,\RO           \tag*{\textbf{Mutability}}\\
			\mu\coloneqq\ &m\,\mid\,\CONSUME           \tag*{\textbf{Access Mode}}\\
			\alpha\coloneqq\ &\epsilon\,\mid\,\CONSUME           \tag*{\textbf{Consume Mode}}\\
			B\coloneqq\ &C\,\mid\,m           \tag*{\textbf{Capture Bound}}\\
			\theta\coloneqq\ &x\,\mid\,c           \tag*{\textbf{Capture}}\\
			C\coloneqq\ &\set{\mu_1\,\theta_1,\dots,\mu_n\,\theta_n}     \tag*{\textbf{Capture Set}}\\
			T,\,U\coloneqq\ &m\,S\capt C\,\mid\,S           \tag*{\textbf{Type}}\\
			R,\,S\coloneqq\ &           \tag*{\textbf{Shape Type}}\\
			&\top                    \tag*{top}\\
			&X                    \tag*{type variable}\\
			&\forall(\alpha\,x: T)U \,\mid\, \forall[X<:S]U \,\mid\, \forall[c<:B]U             \tag*{functions}\\
			&\REF[T]                       \tag*{reference}\\
			&\TBOOL\,\mid\,\TUNIT            \tag*{constant}\\
			\G,\,\Delta\coloneqq\ &\emptyset\,\mid\,\G, x: T\,\mid\,\G, X<:S\,\mid\,\G, \alpha\,c<:B                    \tag*{\textbf{Context}}\\
		\end{flalign*}
		\end{multicols}}
	\vspace{-3em}
	\caption{Abstract syntax of System \capybara{}.}\label{fig:syntax}
\end{wide-rules}
\begin{figure*}[htbp]
\rulefigsetup

\judgheader{Subcapturing}{\subs{\G}{C_1}{C_2}}

\begin{multicols}{3}

\infrule[\ruledef{sc-trans}]
{\subs{\G}{C_1}{C_2}\andalso
  \subs{\G}{C_2}{C_3}}
{\subs{\G}{C_1}{C_3}}

\infrule[\ruledef{sc-elem}]
{C_1\subseteq C_2}
{\subs{\G}{C_1}{C_2}}

\infrule[\ruledef{sc-union}]
{\subs{\G}{C_1}{C}\andalso
  \subs{\G}{C_2}{C}}
{\subs{\G}{C_1\cup C_2}{C}}

\infrule[\ruledef{sc-var}]
{x:m\,S\capt C\in\G}
{\subs{\G}{\set{x}}{C}}

\infrule[\ruledef{sc-mode}]
{m_1\preceq m_2}
{\subs{\G}{m_1\,C}{m_2\,C}}

\infrule[\ruledef{sc-ro-mono}]
{\subs{\G}{C_1}{C_2}}
{\subs{\G}{\RO\,C_1}{\RO\,C_2}}

\end{multicols}

\judgheader[Subcapturing extended to capture bounds plus the following rules:]{Subbounding}{\subs{\G}{B_1}{B_2}}

\begin{multicols}{2}

\infrule[\ruledef{sb-mode}]
{m_1\preceq m_2}
{\subs{\G}{m_1}{m_2}}

\infrule[\ruledef{sb-kind}]
{\typs{\G}{C}{m}}
{\subs{\G}{C}{m}}

\end{multicols}

\judgheader{Typing}{\typ{C}{\G}{t}{T}}

\begin{multicols}{2}

\infrule[\ruledef{var}]
{x: m\,S\capt C\in \G}
{\typ{\set{x}}{\G}{x}{m\,S\capt \set{x}}}

\infrule[\ruledef{readonly}]
{x: S\capt C\in \G}
{\typ{\set{\RO\,x}}{\G}{x}{\RO\,S\capt \set{\RO\,x}}}

\infrule[\ruledef{fresh}]
{x: T[\FRESH\leadsto D_1,\ldots,D_n]\in \G\andalso
 D=\textstyle\bigcup_{i=1}^{n} D_i\\
 \accessonly{\G}{D_1,\ldots,D_n}\andalso
 \consumable{\G}{D_1,\ldots,D_n}\\
 \disj{\G}{D_i}{D_j}\ \text{for}\ i\neq j}
{\typ{(D\cup\CONSUME\,D)}{\G}{x}{T}}

\infrule[\ruledef{sub}]
{\typ{C}{\G}{t}{T}\andalso
 \wf{\G}{C',T'}\\
 \subs{\G}{T}{T'}\andalso
 \subs{\G}{C}{C'}}
{\typ{C'}{\G}{t}{T'}}

\infrule[\ruledef{app}]
{\typ{\set{x}}{\G}{x}{(\forall(z: T) U)\capt C}\\
 \typ{\set{y}}{\G}{y}{T[\ANY\leadsto D]}\andalso
 \accessonly{\G}{D}\\
 \sep{\G}{D}{\vert(\forall(z:T)U)\capt C\vert}}
{\typ{\set{x,y}}{\G}{x\,y}{[z:=y]U}}

\infrule[\ruledef{consume-app}]
{\typ{\set{x}}{\G}{x}{(\forall(\CONSUME\,z: T) U)\capt C}\\
 \typ{\set{y}}{\G}{y}{T[\ANY\leadsto D]}\andalso
 \accessonly{\G}{D}\\
 \seqcomp{\G}{\CONSUME\,D}{\set{x}}\andalso
 \consumable{\G}{D}}
{\typ{\set{x,y}\cup D\cup\CONSUME\,D}{\G}{x\,y}{[z:=y]U}}

\infrule[\ruledef{tabs}]
{\typ{C}{(\G, X<:S)}{t}{U}\andalso\wf{\G}{S}}
{\typ{\set{}}{\G}{\lambda[X<:S]t}{(\forall[X<:S] U)\capt C}}

\infrule[\ruledef{tapp}]
{\typ{C'}{\G}{x}{(\forall[X<:S] U)\capt C}\andalso
 \subs{\G}{S'}{S}}
{\typ{C'}{\G}{x[S']}{[X:=S']U}}

\infrule[\ruledef{cabs}]
{\typ{(C\cup\set{c})}{(\G, c<:B)}{t}{U}\andalso\wf{\G}{C}}
{\typ{\set{}}{\G}{\lambda[c<:B]t}{(\forall[c<:B] U)\capt C}}

\infrule[\ruledef{capp}]
{\typ{C'}{\G}{x}{(\forall[c<:D] U)\capt C}\andalso
 \accessonly{\G}{D}\\
 \sep{\G}{D}{\vert(\forall[c<:D]U)\capt C\vert}}
{\typ{C'\cup D}{\G}{x[D]}{[c:=D]U}}

\end{multicols}

\infrule[\ruledef{abs}]
{\typ{(C\cup\set{x}\cup D)}{(\G,\alpha\,c<:\epsilon,\,x: T[\ANY\leadsto\set{c}])}{t}{U}\andalso
 \wf{\G}{C, T}\\
 D = \set{c,\CONSUME\,c}\quad \text{if $\alpha=\CONSUME$ otherwise $\set{}$}}
{\typ{\set{}}{\G}{\lambda(\alpha\,x: T)t}{(\forall(\alpha\,x: T) U)\capt C}}

\infrule[\ruledef{let}]
{\typ{C_1}{\G}{t}{T}\andalso
 \seqcomp{\G}{C_1}{C_2}\andalso
 \wf{\G}{C_2,U}\andalso
 D = \set{c_1,\ldots,c_n}\\
 \typ{C_2\cup D\cup\CONSUME\,D}{(\G,\CONSUME\,c_1,\ldots,\CONSUME\,c_n,\, x: T[\FRESH\leadsto\set{c_1},\ldots,\set{c_n}])}{u}{U}}
{\typ{(C_1\cup C_2)}{\G}{\LET x = t\IN u}{U}}

\judgheader{Subtyping}{\subs{\G}{T_1}{T_2}}

\begin{multicols}{3}

\infax[\ruledef{top}]
{\subs{\G}{S}{\top}}

\infax[\ruledef{refl}]
{\subs{\G}{T}{T}}

\infrule[\ruledef{trans}]
{\subs{\G}{T_1}{T_2}\andalso \subs{\G}{T_2}{T_3}}
{\subs{\G}{T_1}{T_3}}

\infrule[\ruledef{tvar}]
{X<:S\in\G}
{\subs{\G}{X}{S}}

\infrule[\ruledef{capt}]
{\subs{\G}{S_1}{S_2}\andalso\subs{\G}{C_1}{C_2}}
{\subs{\G}{m\,S_1\capt C_1}{m\,S_2\capt C_2}}

\infrule[\ruledef{fun}]
{\subs{(\G, \alpha\,x: T)}{U_1}{U_2}}
{\subs{\G}{\forall(\alpha\,x: T) U_1}{\forall(\alpha\,x: T) U_2}}

\infrule[\ruledef{tfun}]
{\subs{(\G, X<:S_2)}{U_1}{U_2}\andalso\subs{\G}{S_2}{S_1}}
{\subs{\G}{\forall[X<:S_1] U_1}{\forall[X<:S_2] U_2}}

\infrule[\ruledef{cfun}]
{\subs{\G}{B_2}{B_1}\andalso\subs{(\G, c<:B_2)}{U_1}{U_2}}
{\subs{\G}{\forall[c<:B_1] U_1}{\forall[c<:B_2] U_2}}

\end{multicols}

\vspace{-1.5em}
\caption{Main typing rules of System \capybara{}.}\label{fig:all-typing}

\end{figure*}
\begin{figure*}[htbp]
\rulefigsetup

\judgheader{Capability Kinding}{\typs{\G}{C}{m}}
\begin{multicols}{3}

\infax[\ruledef{k-empty}]
{\typs{\G}{\set{}}{m}}

\infrule[\ruledef{k-union}]
{\typs{\G}{C_1}{m}\andalso
 \typs{\G}{C_2}{m}}
{\typs{\G}{C_1\cup C_2}{m}}

\infrule[\ruledef{k-sc}]
{\subs{\G}{C_1}{C_2}\andalso
 \typs{\G}{C_2}{m}}
{\typs{\G}{C_1}{m}}

\infax[\ruledef{k-rw}]
{\typs{\G}{C}{\epsilon}}

\infax[\ruledef{k-ro}]
{\typs{\G}{\RO\,C}{\RO}}

\infrule[\ruledef{k-imm}]
{c<:\RO\in\G}
{\typs{\G}{\set{c}}{\RO}}

\end{multicols}

\judgheader{Separation}{\sep{\G}{C_1}{C_2}}
{\newcommand{\seprulecell}[2]{\parbox[c]{#1\linewidth}{#2}}
\noindent
\begin{tabular}{@{}c@{\hspace{0.01\linewidth}}c@{\hspace{0.01\linewidth}}c@{}}
\seprulecell{0.34}{
\infrule[\ruledef{sep-symm}]
{\sep{\G}{C_1}{C_2}}
{\sep{\G}{C_2}{C_1}}
}
&
\seprulecell{0.30}{
\infax[\ruledef{sep-empty}]
{\sep{\G}{\set{}}{C}}
}
&
\seprulecell{0.34}{
\infrule[\ruledef{sep-sc}]
{\sep{\G}{C_1}{C_2}\andalso
 \subs{\G}{C_1'}{C_1}}
{\sep{\G}{C_1'}{C_2}}
}
\\
\seprulecell{0.34}{
\infrule[\ruledef{sep-union}]
{\sep{\G}{C_1}{C}\andalso
 \sep{\G}{C_2}{C}}
{\sep{\G}{C_1\cup C_2}{C}}
}
&
\seprulecell{0.30}{
\infrule[\ruledef{sep-ro}]
{\typs{\G}{C_1}{\RO}\andalso
 \typs{\G}{C_2}{\RO}}
{\sep{\G}{C_1}{C_2}}
}
&
\seprulecell{0.34}{
\infrule[\ruledef{sep-root}]
{c_1<:B_1\in\G\quad c_2<:B_2\in\G\\ c_1\neq c_2}
{\sep{\G}{\set{\mu_1\,c_1}}{\set{\mu_2\,c_2}}}
}
\end{tabular}}

\judgheader[The same rules as separation, dropping \ruleref{sep-ro}.]{Disjointness}{\disj{\G}{C_1}{C_2}}

\judgheader{Sequential Composition}{\seqcomp{\G}{C_1}{C_2}}
\begin{multicols}{2}

\infrule[\ruledef{seq-sc}]
{\subs{\G}{C_1}{C_2}\andalso
 \seqcomp{\G}{C_2}{C_3}}
{\seqcomp{\G}{C_1}{C_3}}

\infrule[\ruledef{seq-union}]
{\seqcomp{\G}{C_1}{C_3}\andalso
 \seqcomp{\G}{C_2}{C_3}}
{\seqcomp{\G}{(C_1\cup C_2)}{C_3}}

\infrule[\ruledef{seq-access-only}]
{\accessonly{\G}{C_1}}
{\seqcomp{\G}{C_1}{C_2}}

\infrule[\ruledef{seq-sep}]
{\sep{\G}{C_1}{C_2}}
{\seqcomp{\G}{C_1}{C_2}}
  
\end{multicols}

\judgheader{Access-Only}{\accessonly{\G}{C}}
\begin{multicols}{4}

\infrule[\ruledef{ao-sc}]
{\subs{\G}{C_1}{C_2}\\
 \accessonly{\G}{C_2}}
{\accessonly{\G}{C_1}}

\infrule[\ruledef{ao-union}]
{\accessonly{\G}{C_1}\\
 \accessonly{\G}{C_2}}
{\accessonly{\G}{C_1\cup C_2}}

\infax[\ruledef{ao-empty}]
{\\ \accessonly{\G}{\set{}}}

\infax[\ruledef{ao-access}]
{\\ \accessonly{\G}{\set{m\,c}}}

\end{multicols}

\judgheader{Consumable}{\consumable{\G}{C}}
\begin{multicols}{4}

\infrule[\ruledef{con-sc}]
{\subs{\G}{C_1}{C_2}\\
 \consumable{\G}{C_2}}
{\consumable{\G}{C_1}}

\infrule[\ruledef{con-union}]
{\consumable{\G}{C_1}\\
 \consumable{\G}{C_2}}
{\consumable{\G}{C_1\cup C_2}}

\infax[\ruledef{con-empty}]
{\\ \consumable{\G}{\set{}}}

\infrule[\ruledef{con-root}]
{\CONSUME\,c<:B\in\G}
{\consumable{\G}{\set{m\,c}}}

\end{multicols}

\vspace{-1.5em}
\caption{Typing rules for capabilities.}\label{fig:capture-typing}

\end{figure*}
 
\subsection{Syntax}\label{sec:formal:syntax}

\capybara{} is a variant of \ccformal{}, the calculus of capturing
types~\cite{DBLP:journals/toplas/BoruchGruszeckiOLLB23}, extended
with mutability permissions, consumption, and separation.
Like \ccformal{}, its terms are in \emph{monadic normal form (MNF)}:
applications operate on variables, 
and intermediate results are
named by let bindings.
The normal form matters for application, 
whose result type may mention the parameter.
When the argument is a variable $y$, 
the result type is the exact substitution $[x:=y]U$.
MNF induces no loss of expressiveness, since a general application $t_1\,t_2$ can
always be written
$\LET x_1 = t_1 \IN \LET x_2 = t_2 \IN x_1\,x_2$.

A term is a variable, a value, an application, a type or capture application,
a memory primitive, a parallel composition, a conditional,
or a let-binding.
A value is a term lambda, a type lambda, a capture lambda, or a
constant.
The memory primitives and parallel composition are typed in
\Cref{sec:formal:state}.
The parameter of a term lambda carries a \emph{consume mode} $\alpha$,
recording whether the argument may only be accessed ($\epsilon$)
or also consumed ($\CONSUME$).
\Cref{sec:formal:consume} formalizes consume.

A capturing type $m\,S\capt C$ consists of
a \emph{permission} $m$ qualifying the type,
a shape type $S$,
and a capture set $C$;
we omit the default permission $\epsilon$ and write $S\capt C$.
The permission constrains what the value may be used for:
by default a value is read-write ($\epsilon$),
and the read-only permission $\RO$ restricts it to non-mutating
operations (recall \Cref{sec:informal:readonly}'s \lstinline|ro Buffer^{...}|).
Mutabilities are ordered by a reflexive relation $\preceq$,
generated by $\RO\preceq\epsilon$:
read-only is the weaker of the two.
The capture set over-approximates
the capabilities a value of this type may use
and at what access modes:
each element $\mu\,\theta$ pairs a capture $\theta$ with an access
mode $\mu$.
A mutability marks an access to a capability: $\epsilon\,\theta$
uses it read-write, $\RO\,\theta$ read-only.
$\CONSUME\,\theta$ instead marks that the capability is \emph{consumed}.

Shape types include the top type,
type variables,
the three function types
(term, type, and capture functions),
reference types $\REF[T]$,
and the base types $\TBOOL$ and $\TUNIT$.
A capture bound $B$, 
used to bound a capture parameter, 
is either a concrete capture set or a mutability.
For instance, $\lambda[c<:\RO]t$ bounds $c$
by capabilities with at most read-only access.
Capability kinding (presented shortly) makes precise which capture
sets a mutability bound accepts.

\subsection{Typing, Subtyping and Subcapturing}\label{sec:formal:typing}

\rwpar{Subcapturing}
Subcapturing $\subs{\G}{C_1}{C_2}$ is a preorder.
It subsumes set inclusion (\ruleref{sc-union} and \ruleref{sc-elem}).
\ruleref{sc-var} resolves a variable to the capture set it was declared with:
$x:m\,S\capt C\in\G$ gives
$\subs{\G}{\set{x}}{C}$.
\ruleref{sc-mode} and \ruleref{sc-ro-mono} concern \emph{qualified} capture
sets.
Qualifying $\mu\,C$ restamps every element of $C$ with the access mode $\mu$:
\begin{align*}
\epsilon\,C &= C
& \RO\,C &= \set{\RO\,\theta \mid \mu\,\theta\in C}
& \CONSUME\,C &= \set{\CONSUME\,\theta \mid \mu\,\theta\in C}.
\end{align*}
\ruleref{sc-mode} lets a capture set be qualified down to a weaker mode,
since $\RO\preceq\epsilon$; \ruleref{sc-ro-mono} extends this to a
congruence under qualifying.

\rwpar{Subbounding}
Subbounding $\subs{\G}{B_1}{B_2}$ extends subcapturing to capture bounds $B$.
\ruleref{sb-mode} orders mutability bounds the same way as mutabilities;
\ruleref{sb-kind} admits a capture set below a mutability bound $m$ once it
carries kind $m$, a judgment defined at the end of this subsection.
This is how $c$ in $\lambda[c<:\RO]t$ can be instantiated by any
read-only capture set.

\rwpar{Typing}
The judgment $\typ{C}{\G}{t}{T}$ states that $t$ evaluates to a value of
type $T$ with \emph{use set} $C$, the set of capabilities exercised during
the evaluation of $t$.
\ruleref{var} types $x$ at $\set{x}$, 
refining its
declared capture set to the singleton $\set{x}$ and keeping its permission.
\ruleref{readonly} types the same occurrence $x$ at a read-only view instead,
charging $\set{\RO\,x}$.
This formalizes the \lstinline|.rd| suffix of \Cref{sec:informal:readonly},
which is expository, not concrete syntax.
\ruleref{sub} is ordinary subsumption, weakening the use set and type subject
to target well-formedness $\wf{\G}{\cdot}$.

\ruleref{abs} and \ruleref{app} give the $\ANY$ mechanism of
\Cref{sec:informal} its formal content.
Both rest on \emph{root capability instantiation}.
A root capability occurs as an element $\mu\,\ANY$ (or $\mu\,\FRESH$) of a capture set inside a type. 
Instantiating an occurrence with a capture set $D$
replaces the element by $D$ qualified at its mode:
\[
C[\ANY\leadsto D] = (C\setminus\set{\mu\,\ANY})\cup\mu\,D
\qquad\text{where }\mu\,\ANY\in C,
\]
descending compositionally through all other capture sets and type formers.
$T[\ANY\leadsto D_1,\ldots,D_n]$ instantiates the $n$ occurrences of $\ANY$
in $T$, read left to right, with $D_1,\ldots,D_n$ respectively.
The single-set form $T[\ANY\leadsto D]$ gives every occurrence the same $D$.
$\FRESH$-instantiation is defined the same way.
Intuitively, the two root capabilities stand for capture
quantification, $\ANY$ universal and instantiated by the caller,
$\FRESH$ existential and witnessed by the function body;
the translation to \corecapybara{} makes this reading literal
(\Cref{app:translation}).

Abstracting $\lambda(\alpha\,x:T)t$ binds a fresh capture variable $c$,
mutability-bounded ($\alpha\,c<:\epsilon$), 
and checks the body against $T$
with every $\ANY$ in it instantiated to $\set{c}$, 
written $T[\ANY\leadsto\set{c}]$.
At a call $x\,y$, \ruleref{app} requires $y$'s type to have the form
$T[\ANY\leadsto D]$ for some access-only $D$ separated from the callee's
spine capture set $|(\forall(z:T)U)\capt C|$;
\Cref{sec:formal:consume} defines both judgments.

\ruleref{fresh} gives the $\FRESH$ of \Cref{sec:informal:consume} its
formal content: it widens concrete capture sets to $\FRESH$.
A variable $x$ declared at $T[\FRESH\leadsto D_1,\ldots,D_n]$, an
instance of the $\FRESH$-carrying type $T$ with concrete witnesses
$D_1,\ldots,D_n$, can be typed at $T$ itself.
This is how \lstinline|merge| returns one of its arguments at
\lstinline|Buffer^{fresh}|.
The widening consumes the witnesses: each $D_i$ must be access-only
and consumable, and the rule charges $D\cup\CONSUME\,D$ for
$D=D_1\cup\cdots\cup D_n$.
The witnesses must moreover be pairwise \emph{disjoint}
(\Cref{sec:formal:consume}),
so that distinct $\FRESH$es never overlap.

\ruleref{let} types let bindings and opens $\FRESH$-carrying types such as the
\lstinline|Buffer^{fresh}| of \Cref{sec:informal:consume}.
A binding $\LET x=t\IN u$ binds $x$ at
$T[\FRESH\leadsto\set{c_1},\ldots,\set{c_n}]$ for new consumable capture
variables $c_1,\ldots,c_n$, one per $\FRESH$ that $t$'s declared type
$T$ carries.
Distinct $\FRESH$es thus resolve to distinct roots: each names its own
freshly created resource, which the continuation may access and consume
independently of the others.
The rule also sequences the use of $t$ before that of $u$
($\seqcomp{\G}{C_1}{C_2}$, \Cref{sec:formal:consume}), which enforces
the consume semantics: nothing $t$ consumes can be used in $u$.

\ruleref{tabs}, \ruleref{tapp}, \ruleref{cabs}, and \ruleref{capp} abstract
and apply over type and capture parameters.
\ruleref{cabs} is analogous to \ruleref{abs}, abstracting over a capture
parameter $c$ bounded by $B$; \ruleref{capp} mirrors \ruleref{app}, requiring
the instantiating capture set $D$ to be separated from the callee's own
footprint and charging $D$ alongside the callee.

Note that the abstraction rules have the empty use set:
a lambda is already a value, so evaluating it uses no capabilities,
and the body's use set becomes the capture annotation of the function
type.
This makes the types of curried functions precise.
Assume a capability $\textsf{a}$, invoked by applying
it to unit, and consider
$\lambda(z_1{:}\TUNIT)\,\lambda(z_2{:}\TUNIT)\,\textsf{a}\,z_2$.
The inner body uses $\set{\textsf{a}}$, so the term has the capturing
type
\[
\forall(z_1{:}\TUNIT)\,(\forall(z_2{:}\TUNIT)\,\TUNIT)\capt\set{\textsf{a}},
\]
whose outer function is pure: the use of $\textsf{a}$ is charged to
the arrow whose application performs it.

The same precision means the outermost capture set of a function type
understates what calling the function can do: the type above is pure
at top level, yet applying the function and then its result exercises
$\textsf{a}$.
The footprint that \ruleref{app} and \ruleref{capp} check separation
against is therefore the \emph{spine capture set} $|T|$,
the capabilities $T$ mentions along its spine:
\begin{align*}
|\top| = |X| = |\TBOOL| = |\TUNIT| &= \set{} & |\forall[X<:S']U| &= |S'|\cup|U|\\
|m\,S\capt C| &= \mfree{C}\cup|S| & |\forall[c<:B]U| &= |U|\setminus c\\
|\forall(x:m\,S\capt C)U| &= \mfree{C}\cup(|U|\setminus x) & |\REF[T]| &= |T|
\end{align*}
where $\mfree{C}$ removes the root capabilities $\ANY$ and $\FRESH$
from $C$.
The recursion walks the spine of nested arrows, type bounds, and
reference contents,
unioning the capture set at each domain and codomain,
so $|T|$ collects every capability a value of type $T$ can eventually
use, not only those its outermost capture set names.
Two kinds of elements are dropped along the way:
the variables that nested binders introduce
and the root capabilities.
\ruleref{abs} and \ruleref{cabs} type a function body treating its
parameter as separate from the rest of the function.
Checking the argument against the whole spine is what makes that
assumption good at every later application.

\rwpar{Subtyping}
Subtyping $\subs{\G}{T_1}{T_2}$ has the usual reflexive-transitive shell
\ruleref{refl}, \ruleref{trans}, a type variable rule \ruleref{tvar},
and one congruence per shape \ruleref{fun}, \ruleref{tfun}, \ruleref{cfun}.
Type and capture functions are contravariant in their bounds
and covariant in their bodies under the extended context.
Term functions are \emph{invariant} in their domains:
a function body treats its parameter's capture set as a
separation assumption (\Cref{sec:formal:consume}),
and narrowing the domain by subtyping would strengthen that assumption behind the function's back.
Call sites lose no flexibility, since \ruleref{sub} still adapts an argument
to the declared domain, and a function can always be $\eta$-expanded at a
narrower domain.
\ruleref{capt} lifts shape subtyping and subcapturing to capturing types,
$\subs{\G}{m\,S_1\capt C_1}{m\,S_2\capt C_2}$ from $\subs{\G}{S_1}{S_2}$ and
$\subs{\G}{C_1}{C_2}$, but keeps the permission $m$ fixed on both sides.
Consequently, subtyping alone can never turn a read-write value into a
read-only one; only \ruleref{readonly}, applied to a variable occurrence,
weakens a permission, and only \ruleref{sc-mode} weakens an access mode
inside a capture set.

\rwpar{Capability kinding}
The judgment $\typs{\G}{C}{m}$ certifies that 
the capabilities in $C$ are used at most at mutability $m$: kind $\RO$ means $C$ contains only read-only uses, 
while kind $\epsilon$ also admits read-write uses 
and thus holds of every capture set \ruleref{k-rw}.
A set qualified read-only has kind $\RO$ \ruleref{k-ro},
and a capture variable declared with a mutability bound inherits that bound
as its kind \ruleref{k-imm}, so $c<:\RO\in\G$ gives $\set{c}:\RO$.
Kinding is closed under union \ruleref{k-union} and inherited along
subcapturing \ruleref{k-sc}: a set below one of kind $m$ itself has kind
$m$.

\subsection{Separation and Consume}\label{sec:formal:consume}

\rwpar{Separation}
$\sep{\G}{C_1}{C_2}$ certifies that two capture sets do not interfere.
The check traces both sets back to their roots
and compares the roots pairwise.
The tracing is done by \ruleref{sep-sc}: to separate a set,
it suffices to separate a set above it in subcapturing,
in particular the roots it resolves to.
\ruleref{sep-symm} makes the rule available on both sides.
\ruleref{sep-union} and \ruleref{sep-empty} then decompose
the root sets into pairs of individual roots.
Two leaves close the derivation.
\ruleref{sep-root} separates any two distinct roots:
roots are the basic carriers of separation guarantees.
\ruleref{sep-ro} separates any two access-only sets of kind $\RO$,
so capabilities may overlap freely at read-only mode.

\rwpar{Disjointness}
$\disj{\G}{C_1}{C_2}$ is the fragment of separation without
\ruleref{sep-ro}.
It admits no read-only sharing:
two read-only views of the same capability are separate but not
disjoint.
Disjoint sets thus hold genuinely distinct capabilities,
and \ruleref{fresh} uses it to keep its witnesses apart.

\rwpar{Sequential composition}
$\seqcomp{\G}{C_1}{C_2}$ asks whether a computation using $C_1$ can safely
run before one using $C_2$.
This judgment is where the static semantics enforces consume:
reads and writes end with the first computation,
but a consumed capability stays gone,
and nothing that runs later may use it.
\ruleref{let} checks the judgment between a bound term and its
continuation, \ruleref{consume-app} between consuming the argument and
the call itself.
The check follows the same scheme as separation:
\ruleref{seq-sc} traces $C_1$ back to its roots,
\ruleref{seq-union} splits them, and two leaves decide.
A part that consumes nothing is harmless \ruleref{seq-access-only};
a part that consumes must be separated from $C_2$,
keeping what it consumes unusable in the computation that follows
\ruleref{seq-sep}.

\rwpar{Access-only and consumable}
$\accessonly{\G}{C}$ holds when every capture variable $C$ resolves to is
tagged with a mutability rather than $\CONSUME$ \ruleref{ao-access}, and is
closed under union and subcapturing.
$\consumable{\G}{C}$ holds when every capture variable $C$ resolves to is
declared with consume mode $\CONSUME$ \ruleref{con-root}.
The typing rules require access-only sets wherever they instantiate
$\ANY$, $\FRESH$, or a capture binder: the witness $D$ of \ruleref{app},
\ruleref{consume-app}, and \ruleref{capp}, and the witnesses of
\ruleref{fresh}.
A $\CONSUME$-qualified capability is more than a permission to access:
using it kills the capability.
Hiding one inside the instance of a capture variable would make the kill
invisible, since an occurrence of the capture variable reads as a plain access,
and nothing would record that using a set mentioning it may consume
capabilities.

\begin{figure}[tbp]
\small\centering
\setlength{\tabcolsep}{3.5pt}
\begin{tabular}{@{}lcl@{}}
\toprule
\textbf{Surface} (\Cref{sec:formal}) & &
\textbf{Core} (\Cref{sec:core}), with $\Psi = (\set{d_1},\set{d_2})$ \\
\midrule
$\textsf{rp} : \forall(op_1: X\capt\set{\ANY})$
 & $\longmapsto$ &
$\textsf{rp} : \forall[d_1<:\top]\,\forall(op_1: X\capt\set{d_1})$\\
$\phantom{\textsf{rp} : }\big(\forall(op_2: X\capt\set{\ANY})\,\TUNIT\big)\capt\set{op_1}$
 & &
$\phantom{\textsf{rp} : }\big(\forall[d_2<:\top]\,\forall(op_2: X\capt\set{d_2})\,([\Psi,\emptyset]\,\TUNIT)\capt\set{op_1,op_2}\big)\capt\set{op_1}$\\
\addlinespace[3pt]
$\textsf{alloc} : \forall(y:\TUNIT)\,(X\capt\set{\FRESH})$
 & $\longmapsto$ &
$\textsf{alloc} : \forall(y:\TUNIT)\,\EXCAP{c}\,X\capt\set{c}$\\
\midrule
$\LET f_1 = \textsf{alloc}\,()\IN$
 & $\longmapsto$ &
$\LET \<c_1,f_1\> = \textsf{alloc}\,()\IN$\\
$\LET f_2 = \textsf{alloc}\,()\IN$
 & $\longmapsto$ &
$\LET \<c_2,f_2\> = \textsf{alloc}\,()\IN$\\
$\LET g = \textsf{rp}\,f_1\IN$
 & $\longmapsto$ &
$\LET g_1 = \textsf{rp}[\set{f_1}]\IN\ \LET g_2 = g_1\,f_1\IN$\\
$g\,f_2$
 & $\longmapsto$ &
$\LET g_3 = g_2[\set{f_2}]\IN\ \LET g_4 = g_3\,f_2\IN\ \UNLOCK\,g_4$\\
\bottomrule
\end{tabular}
\caption{The running example, surface (left) against core (right),
aligned construct by construct.
$\textsf{rp}$ is the \lstinline|runParallel| pattern of
\Cref{sec:informal:separation} over an abstract resource shape $X<:\top$,
applied to two freshly allocated resources;
$\textsf{alloc}\,()$ abbreviates $\LET y = ()\IN \textsf{alloc}\,y$.
Each $\ANY$ becomes a capture parameter $d_i$, instantiated at the call;
each $\FRESH$ allocation returns an existential package, unpacked into a
distinct consumable root $c_i$;
the modal result $[\Psi,\emptyset]\,\TUNIT$ locks the result behind the
separation assumption $\Psi$, discharged by the final $\UNLOCK$.
\Cref{sec:formal:consume,sec:core:consume} give the surface and core
derivations; \Cref{app:tr:example} gives the translation.}
\label{fig:running-example}
\end{figure}
 
\rwpar{A worked example}
In the program of \Cref{fig:running-example} (left), each
$\textsf{alloc}$ result carries one $\FRESH$, so \ruleref{let} binds
$f_i : X\capt\set{c_i}$ with distinct consumable roots $c_1\neq c_2$.
At the first application $\textsf{rp}\,f_1$, \ruleref{app} instantiates
$\ANY$ with $D = \set{f_1}$, and the separation premise is trivial: the
spine capture set of $\textsf{rp}$'s type is empty.
The binding types $g$ at
$(\forall(op_2:X\capt\set{\ANY})\,\TUNIT)\capt\set{f_1}$:
the partial application has captured $f_1$.
The second application $g\,f_2$ is where separation is enforced:
$D = \set{f_2}$ meets the spine $\set{f_1}$, and
$\sep{\G}{\set{f_2}}{\set{f_1}}$ follows because \ruleref{sc-var}
resolves each $f_i$ to its root, \ruleref{sep-root} separates the
distinct roots, and \ruleref{sep-sc} carries the conclusion back.
Passing the same resource twice would instead demand
$\sep{\G}{\set{f}}{\set{f}}$, which no rule derives.
\Cref{sec:core} returns to this program in its core form
(\Cref{fig:running-example}, right); \Cref{app:tr:example} traces its
translation.

\subsection{Mutable State and Parallelism}\label{sec:formal:state}

\begin{figure*}[htbp]
\rulefigsetup

\begin{multicols}{2}

\infrule[\ruledef{alloc}]
{\typ{\set{x}}{\G}{x}{T}}
{\typ{\set{x}}{\G}{\ALLOC x}{\REF[T]\capt \set{\FRESH}}}

\infrule[\ruledef{read}]
{x: m\,\REF[T]\capt C\in \G}
{\typ{\set{\RO\,x}}{\G}{\READ x}{T}}

\infrule[\ruledef{write}]
{x: \REF[T]\capt C\in \G\andalso
 \typ{\set{y}}{\G}{y}{T}}
{\typ{\set{x}}{\G}{x:=y}{\TUNIT}}

\infrule[\ruledef{dealloc}]
{x: \REF[T]\capt C\in \G\andalso
 \consumable{\G}{\set{x}}}
{\typ{\set{\CONSUME\,x}}{\G}{\DEALLOC x}{\TUNIT}}

\infrule[\ruledef{par}]
{\sep{\G}{C_1}{C_2}\andalso
 \typ{C_1}{\G}{t_1}{T_1}\andalso
 \typ{C_2}{\G}{t_2}{T_2}}
{\typ{C_1\cup C_2}{\G}{\PAR t_1, t_2}{\TUNIT}}

\infrule[\ruledef{if}]
{\typ{\set{x}}{\G}{x}{\TBOOL}\andalso
 \typ{C_1}{\G}{t_1}{T}\andalso
 \typ{C_2}{\G}{t_2}{T}}
{\typ{\set{x}\cup C_1\cup C_2}{\G}{\IF x\THEN t_1\ELSE t_2}{T}}

\end{multicols}

\vspace{-1em}
\caption{Typing rules for mutable state, parallelism, and conditionals of System \capybara{}.}\label{fig:state-typing}

\end{figure*}
 
The calculus so far manages capabilities but has no concrete effects.
Mutable references, parallel composition, and conditionals give its
guarantees operational content (\Cref{fig:state-typing}).

A reference of type $\REF[T]\capt C$ is a mutable cell holding a $T$.
Allocation creates a fresh one:
\ruleref{alloc} types $\ALLOC x$ at $\REF[T]\capt\set{\FRESH}$, so
the receiving let binding names the new cell by a consumable capture
variable of its own \ruleref{let}.
Deallocation is the primitive consumer:
\ruleref{dealloc} requires $\set{x}$ consumable and charges
$\set{\CONSUME\,x}$, after which sequential composition keeps the
dead reference out of reach.
Reading is a read-only effect:
\ruleref{read} accepts a reference of either permission and charges
the read-only view $\set{\RO\,x}$.
Writing requires the read-write permission and charges $\set{x}$
\ruleref{write}.

The term $\PAR t_1, t_2$ runs its two branches concurrently.
\ruleref{par} requires the use sets of the two branches to be
separated: two computations may run in parallel exactly when
separation certifies that their capabilities do not interfere.
Conditionals are standard \ruleref{if}.
All these constructs carry over to the core calculus
(\Cref{sec:core}), whose dynamic semantics gives them operational
meaning.
\section{The Core Calculus \corecapybara{}}\label{sec:core}

\corecapybara{} is the target of the type-preserving translation
(\Cref{app:translation}) and the calculus for our metatheory
(\Cref{sec:metatheory}).
It is a variant of System \capless{}~\cite{whatisinthebox}, inheriting
existential capture types, and shares the surface calculus's monadic
normal form.
It makes surface contracts explicit: separation becomes modal types,
freshness existential capture types, and read-only views reader values.
We present only the constructs and judgments without surface counterparts;
\Cref{app:corecapybara} gives the complete definitions, which follow the
Lean mechanization.

\begin{wide-rules}\noindent
	\vspace{-2em}
	{\small\begin{multicols}{2}\noindent
		\begin{flalign*}
			s,\,t,\,u\coloneqq\ &\ldots           \tag*{\textbf{Term}}\\
			&x\,t                              \tag*{consumer app.}\\
			&\UNLOCK\,x                  \tag*{unlock}\\
			&\LET \<c_1,\ldots,c_n,x\> = t \IN u               \tag*{unpack}\\
			v\coloneqq\ &\ldots           \tag*{\textbf{Value}}\\
			&\lambda(\<c,x\>: \EXCAP{c}T)t                      \tag*{consumer lambda}\\
			&\RO\,x                      \tag*{reader}\\
			&\LOCK[\Psi,\Phi]t                      \tag*{modal}\\
			&\<C_1,\ldots,C_n,x\>                      \tag*{pack}\\
		\end{flalign*}
		\begin{flalign*}
			\alpha\coloneqq\ &\ldots\,\mid\,\KILLED           \tag*{\textbf{Consume Mode}}\\
			B\coloneqq\ &C\,\mid\,\top                      \tag*{\textbf{Capture Bound}}\\
			R,\,S\coloneqq\ &\ldots           \tag*{\textbf{Shape Type}}\\
			&\forall(\<c,x\>: \EXCAP{c}T)E             \tag*{consumer}\\
			&[\Psi,\Phi]E             \tag*{modal}\\
			E,\,F\coloneqq\ &T\,\mid\,\EXCAP{c_1,\ldots,c_n}T \tag*{\textbf{Existential Type}}\\
			\sepctx\coloneqq\ &C_1,\ldots,C_n        \tag*{\textbf{Separation Context}}\\
			\modectx\coloneqq\ &C_1:m_1,\ldots,C_n:m_n        \tag*{\textbf{Mode Context}}\\
			\G,\,\Delta\coloneqq\ &\ldots\,\mid\,\G, \LOCK[\sepctx,\modectx]                    \tag*{\textbf{Context}}\\
		\end{flalign*}
		\end{multicols}}
	\vspace{-3em}
	\caption{Syntax of System \corecapybara{}: the delta over \capybara{}
	(\Cref{fig:syntax}), with \ldots{} marking the carried-over
	productions. \Cref{app:corecapybara} presents the complete syntax.}\label{fig:core-syntax-delta}
	\vspace{-2em}
\end{wide-rules}
 
\subsection{Syntax and Judgments}
\Cref{fig:core-syntax-delta} gives the delta over \capybara{};
\Cref{fig:core-syntax} gives the complete syntax.
\corecapybara{} carries over the surface abstractions, memory primitives,
parallel composition, and conditionals, but replaces root capabilities with
existential packs and unpacks, modal locks and unlocks, and reader values.
A capture binding $\alpha\,c<:B$ is access-only ($\epsilon$), consumable
($\CONSUME$), or killed ($\KILLED$); killed bindings arise only during
typing.

The judgments of \Cref{sec:formal} carry over, with result types extended to
existentials.
Subcapturing also resolves access-only capture variables to their bounds
\rruleref{sc-cvar} and forgets read-only qualifiers
\rruleref{sc-ro}.
Below we give the changes to separation and the new satisfaction judgment
$\sat{\G}{[\Psi,\Phi]}$.

\subsection{Modals and Separation}

Separation $\sep{\G}{C_1}{C_2}$ concludes at three leaves:

{\inlinerulesetup
\threerulecolumns
{\infrule[\rruleuse{sep-ro}]
 {\typs{\G}{C_1}{\RO}\andalso
  \typs{\G}{C_2}{\RO}}
 {\sep{\G}{C_1}{C_2}}}
{\infrule[\rruleuse{sep-consumable}]
 {\CONSUME\,c_1<:B_1\in\G\\
  \CONSUME\,c_2<:B_2\in\G\\
  c_1\neq c_2}
 {\sep{\G}{\set{\mu_1\,c_1}}{\set{\mu_2\,c_2}}}}
{\infrule[\rruleuse{sep-lock}]
 {\LOCK[\Psi,\Phi]\in\G\\
  C_1,\,C_2\ \text{distinct entries of}\ \Psi}
 {\sep{\G}{C_1}{C_2}}}}

\noindent
\rruleref{sep-ro} separates two read-only sets, while
\rruleref{sep-consumable} separates distinct consumable capture variables
and replaces the surface \ruleref{sep-root}.
\rruleref{sep-lock} reads a modal lock's assumptions:
the entries of $\Psi$ are pairwise separate, while $(C:m)\in\Phi$ bounds
$C$'s kind by $m$.

{\inlinerulesetup
\tworulecolumns
{\infrule[\rruleuse{lock}]
 {\typ{C}{(\G, \LOCK[\Psi,\Phi])}{t}{E}}
 {\typ{\set{}}{\G}{\LOCK[\Psi,\Phi]\,t}{([\Psi,\Phi]E)\capt C}}}
{\infrule[\rruleuse{unlock}]
 {\typ{\set{}}{\G}{x}{([\Psi,\Phi]E)\capt \set{x}}\\
  \sat{\G}{[\Psi,\Phi]}}
 {\typ{\set{x}}{\G}{\UNLOCK\,x}{E}}}}

\noindent
\rruleref{lock} types its body under the recorded assumptions;
\rruleref{unlock} eliminates the modality when those assumptions are
satisfied, as defined by \rruleref{sat}:

{\inlinerulesetup
\infrule[\rruleuse{sat}]
{\typs{\G}{C}{m}\quad\text{for each}\ (C:m)\in\Phi\\
 \sep{\G}{C_1}{C_2}\quad\text{for each pair of distinct entries}\ C_1,\,C_2\ \text{of}\ \Psi}
{\sat{\G}{[\Psi,\Phi]}}}

\noindent
Locks make surface separation checks compositional.
A translated function assumes that its parameter is separate from the
function's footprint; each call discharges that assumption at an unlock,
which implements the separation premise of \ruleref{app}
(\Cref{app:translation}).

\subsection{Consume, Freshness, and Existentials}\label{sec:core:consume}

A consumable capture variable may be consumed once.
The \emph{kill} operation marks every consumable root of $C$ as killed:
\[
\G\ominus C \coloneqq
\G\big[\;\alpha\,c<:B \;\mapsto\; \KILLED\,c<:B
  \;\big|\; \CONSUME\,c\in\roots{\G}{C}\;\big],
\]
where $\roots{\G}{C}$ resolves the term variables of $C$ through their
capture sets until only capture variables remain (\Cref{app:corecapybara}).
Killed variables may be neither accessed nor consumed.

{\inlinerulesetup
\infrule[\rruleuse{let}]
{\seqcomp{\G}{C_1}{C_2}\andalso
 \typ{C_1}{\G}{t}{T}\andalso
 \typ{C_2}{((\G\ominus C_1), x: T)}{u}{E}}
{\typ{C_1\cup C_2}{\G}{\LET x = t\IN u}{E}}}

\noindent
\rruleref{let} therefore makes anything consumed by the bound term
unusable in its continuation.

Fresh results are existential packages,
introduced by packs and eliminated by unpacks:

{\inlinerulesetup
\infrule[\rruleuse{pack}]
{\textstyle D=\bigcup_{i=1}^{n} C_i\andalso
 \accessonly{\G}{D}\andalso
 \consumable{\G}{D}\andalso
 \disj{\G}{C_i}{C_j}\ \text{for}\ i\neq j\\
 \typ{\set{}}{\G}{x}{[c_1:=C_1,\ldots,c_n:=C_n]T}}
{\typ{D\cup\CONSUME\,D}{\G}{\<C_1,\ldots,C_n,x\>}{\EXCAP{c_1,\ldots,c_n}T}}

\infrule[\rruleuse{unpack}]
{\typ{C_1}{\G}{t}{\EXCAP{c_1,\ldots,c_n}T}\andalso
 \seqcomp{\G}{C_1}{C_2}\andalso
 \Psi_w = \set{c_1,\ldots,c_n}\,,\ C_2\andalso
 \Phi_w=\emptyset\\
 \G' = (\G\ominus C_1), \CONSUME\,c_1<:\top, \ldots, \CONSUME\,c_n<:\top,
   \LOCK[\Psi_w,\Phi_w], x: T\\
 \typ{C_2\cup\set{c_1,\ldots,c_n,\CONSUME\,c_1,\ldots,\CONSUME\,c_n}}{\G'}{u}{E}}
{\typ{C_1\cup C_2}{\G}{\LET \<c_1,\ldots,c_n,x\> = t\IN u}{E}}}

\noindent
\rruleref{pack} consumes pairwise-disjoint witnesses and hides them behind
existential capture variables,
the premise that \ruleref{fresh} already imposes on the surface.
The core disjointness judgment drops the lock leaf \rruleref{sep-lock},
so disjointness never rests on modal assumptions:
each witness is a resource of its own, not a shared view.
\rruleref{unpack} binds the witnesses as fresh consumable variables and
records the \emph{freshness lock} $[\Psi_w,\emptyset]$,
making the witnesses separate from the continuation's charge $C_2$
(\Cref{rem:tr:ownership}).
Accordingly, surface \ruleref{fresh} occurrences translate to packs and fresh
roots introduced by \ruleref{let} translate to unpacks
(\Cref{app:translation}).
A consumer lambda is a first-class unpacking function and the translated form
of a $\CONSUME$-mode function:

{\inlinerulesetup
\tworulecolumns
{\infrule[\rruleuse{consumer}]
 {\typ{C\cup\set{c,\CONSUME\,c}}{(\G, \CONSUME\,c<:\top, x: T)}{t}{E}}
 {\typ{\set{}}{\G}{\lambda(\<c,x\>: \EXCAP{c}T)t}{(\forall(\<c,x\>:\EXCAP{c}T) E)\capt C}}}
{\infrule[\rruleuse{consumer-app}]
 {\seqcomp{\G}{C_1}{\set{x}}\andalso
  \accessible{\G}{\set{x}}\\
  \typ{\set{}}{\G}{x}{(\forall(\<c,x\>: \EXCAP{c}T) E)\capt \set{x}}\andalso
  \typ{C_1}{\G}{t}{\EXCAP{c}T}}
 {\typ{C_1\cup\set{x}}{\G}{x\,t}{E}}}}

\noindent
\rruleref{consumer} binds the witness as a consumable root;
\rruleref{consumer-app} sequences argument consumption before access to the
consumer, preventing the former from disabling the latter.

\rwpar{Example}
In the core form of the running example
(\Cref{fig:running-example}, right),
the body of $\textsf{rp}$ is typed under $\LOCK[\Psi,\emptyset]$,
where \rruleref{sep-lock} supplies the separation of $\set{d_1}$ from
$\set{d_2}$ that the surface program assumes.
The two unpacks bind $f_i : X\capt\set{c_i}$ with distinct consumable
roots, and the applications instantiate $d_1:=\set{f_1}$ and
$d_2:=\set{f_2}$, leaving
$g_4 : ([\Psi',\emptyset]\,\TUNIT)\capt\set{f_1,f_2}$
with the concrete lock $\Psi' = (\set{f_1},\set{f_2})$.
The final \rruleref{unlock} demands its satisfaction,
$\sep{\G}{\set{f_1}}{\set{f_2}}$:
\ruleref{sep-sc} resolves each $f_i$ to its root,
and \rruleref{sep-consumable} separates the two.
Had both arguments been the same resource, the entries of $\Psi'$ would
share a root, and unlocking would be underivable.

\subsection{References and Readers}

Memory primitives keep their surface typing except where explicit core
constructs take over.
Allocation returns the fresh reference as a one-variable package:

{\inlinerulesetup
\infrule[\rruleuse{alloc}]
{\typ{\set{}}{\G}{x}{T}}
{\typ{\set{}}{\G}{\ALLOC x}{\EXCAP{c}\REF[T]\capt \set{c}}}}

\noindent
Unpacking introduces a consumable capture variable; read-only views become
reader values:

{\inlinerulesetup
\tworulecolumns
{\infrule[\rruleuse{reader}]
 {x: \REF[T]\capt C\in \G}
 {\typ{\set{}}{\G}{\RO\,x}{\RO\,\REF[T]\capt\set{\RO\,x}}}}
{\infrule[\rruleuse{read}]
 {\accessible{\G}{\set{x}}\\
  \typ{\set{}}{\G}{x}{\RO\,\REF[T]\capt C}}
 {\typ{\set{x}}{\G}{\READ x}{T}}}}

\noindent
\rruleref{reader} translates the surface \ruleref{readonly} view.
The other constructs carry over, but memory operations require accessible,
hence non-killed, references (\Cref{app:corecapybara}).

\subsection{Dynamic Semantics}

Evaluation runs a term against a \emph{heap},
a finite map from locations to cells:
\[
H(l) \;::=\; v \;\mid\; \refcell{l'} \;\mid\; \deadcell{l'},
\]
A cell stores a value, a live reference to $l'$, or a dead reference left by
deallocation.
At runtime, term variables additionally range over heap locations:
monadic normal form lifts every intermediate value into a heap cell of its
own (rule \rruleref{s-lift} in \Cref{app:corecapybara}),
so a location acts as the runtime name of a value
and reference cells store locations.
A step from $\cfg{H}{t}$ emits a trace of its accesses:
\[
\tau \;::=\; \cdot \;\mid\; e;\,\tau,
\qquad\qquad
e \;::=\; \RO\,l \;\mid\; \epsilon\,l \;\mid\; \ALLOC l \;\mid\; \DEALLOC l.
\]
The events $\RO\,l$, $\epsilon\,l$, $\ALLOC l$, and $\DEALLOC l$ record reads,
writes, allocation, and deallocation, respectively; only memory primitives
emit them:

{\inlinerulesetup
\begin{multicols}{2}

\infrule[\rruleuse{s-read}]
{H(l) = \RO\,l_c\andalso H(l_c) = \refcell{n}}
{\cfg{H}{\READ l} \stepsto{\RO\,l_c} \cfg{H}{n}}

\infrule[\rruleuse{s-alloc}]
{l_x\in\dom H\andalso l\notin\dom H}
{\cfg{H}{\ALLOC l_x} \stepsto{\ALLOC l} \cfg{H[l\mapsto\refcell{l_x}]}{\<\set{l},l\>}}

\infrule[\rruleuse{s-write}]
{H(l_x) = \refcell{n_0}\andalso l_y\in\dom H}
{\cfg{H}{l_x := l_y} \stepsto{\epsilon\,l_x} \cfg{H[l_x\mapsto\refcell{l_y}]}{()}}

\infrule[\rruleuse{s-drop}]
{H(l_x) = \refcell{n}}
{\cfg{H}{\DEALLOC l_x} \stepsto{\DEALLOC l_x} \cfg{H[l_x\mapsto\deadcell{n}]}{()}}

\end{multicols}}

\noindent
Access and deallocation require a live cell.
\rruleref{s-drop} leaves a dead cell that matches no such premise, making
use-after-free and double-free stuck.

With the location-only branch \emph{footprints} $C_1$ and $C_2$ from
\rruleref{par}, a left branch steps as follows:

{\inlinerulesetup
\infrule[\rruleuse{s-par-l}]
{\cfg{H}{t_1} \stepsto{\tau} \cfg{H'}{t_1'}\andalso
 \tauth{\tau}{C_1}\andalso
 \nintf{C_1}{C_2}}
{\cfg{H}{\PAR[C_1, C_2]\, t_1, t_2} \stepsto{\tau} \cfg{H'}{\PAR[C_1', C_2]\, t_1', t_2}}}

\noindent
Its trace must be \emph{authorized} by the branch's footprint:
\[
\begin{array}{lcl}
\tauth{\cdot}{C} &\text{iff}& \text{true}\\
\tauth{(\mu\,l)\,\tau}{C} &\text{iff}& \mu\preceq\mu'\ \text{for some}\ \mu'\,l\in C,\ \text{and}\ \tauth{\tau}{C}\\
\tauth{(\ALLOC l)\,\tau}{C} &\text{iff}& \tauth{\tau}{C\cup\set{\epsilon\,l,\CONSUME\,l}}\\
\tauth{(\DEALLOC l)\,\tau}{C} &\text{iff}& \CONSUME\,l\in C\ \text{and}\ \tauth{\tau}{C}
\end{array}
\]
Thus accesses require a sufficient mode, deallocation requires
$\CONSUME$, and allocation grants full rights to the new cell.
The branch footprints must also be \emph{non-interfering}:
\[
\nintf{C_1}{C_2}
\quad\text{iff}\quad
l_1\neq l_2\ \text{or}\ \mu_1=\mu_2=\RO
\quad\text{for all}\ \mu_1\,l_1\in C_1\ \text{and}\ \mu_2\,l_2\in C_2,
\]
so a shared location is read-only in both.
After a step, the footprint grows with newly allocated cells:
$C_1' = C_1\cup\set{\epsilon\,l,\CONSUME\,l\mid \ALLOC l\ \text{occurs in}\ \tau}$.
\rruleref{s-par-r} is symmetric,
and \rruleref{s-par-join} returns unit once both branches are answers.
Well-typed programs satisfy authorization and non-interference, the
conditions underlying confluence (\Cref{sec:metatheory}).

$\redsto{\tau}$ is the reflexive-transitive closure of small-step reduction.
Its sequential restriction $\seqredsto{\tau}$ schedules a right branch only
after the left reaches an answer.
\Cref{app:corecapybara} gives both the complete relation and its big-step
counterpart $\evalsto{\tau}$, which follows the same schedule.
Confluence (\Cref{thm:confluence}) transfers results from sequential to
arbitrary runs.
\section{Metatheory}\label{sec:metatheory}

We prove \corecapybara{} sound semantically
\cite{DBLP:journals/jacm/TimanyKDB24}: the fundamental theorem gives
type safety, memory safety, and immutability.
We then prove confluence of the reduction semantics, which gives
data-race freedom.
The core development is mechanized in Lean~4
\cite{DBLP:conf/cade/Moura021}.\footnote{The Lean~4 mechanization was
developed with AI assistance. It is used to do proof engineering work.}
We sketch the model and results, then give the type-preserving translation
that transfers them to \capybara{}.
\Cref{app:metatheory} contains the full model and its Lean
correspondence; \Cref{app:translation} proves the translation on paper.

\subsection{A Logical Model of \corecapybara{}}

\begin{figure*}[tbp]
\small

\judgheader{Footprint denotation}{\fdenot{C}^{H}}
\vspace{-1.4em}
\[
\begin{array}{c}
\fdenot{\set{}}^H = \set{},
\qquad
\fdenot{C_1\cup C_2}^H = \fdenot{C_1}^H\cup\fdenot{C_2}^H,
\qquad
\fdenot{\set{\mu\,l}}^H = \mu\,\fploc{H}{l},\\[4pt]
\text{where}\quad
\fploc{H}{l} =
\begin{cases}
\set{\epsilon\,l} & \text{if}\ H(l)\ \text{is a reference cell,}\\
F_v & \text{if}\ H(l) = v\ \text{with footprint}\ F_v
\end{cases}
\end{array}
\]

\judgheader{Worlds and world extension}{\wext{\Sigma'}{H'}{\Sigma}{H}}
\vspace{-1.4em}
\[
\begin{array}{c}
\mathit{World}_k = \mathit{Loc}\rightharpoonup \mathit{Rel}_k,
\qquad
\mathit{Rel}_k = \textstyle\prod_{j<k}
  \left(\mathit{World}_j\to\mathit{Heap}\to\mathit{Term}\to\mathrm{Prop}\right)\\[4pt]
\wext{\Sigma'}{H'}{\Sigma}{H}
\ \coloneqq\
H'\sqsupseteq H
\ \wedge\
\forall l\in\dom{\Sigma}.\ \Sigma'(l)=\Sigma(l)
\end{array}
\]

\judgheader{Value denotation (selected clauses)}{a\in\vdenot{T}_\rho(k,\Sigma,H)}
\vspace{-1.4em}
\[
\begin{array}{c}
\resolve(H,v) = v,
\quad
\resolve(H,l) = v\ \ \text{if}\ H(l)=v,
\quad
\Sigma(l)\approx_k\vdenot{T}_\rho
\ \coloneqq\
\forall j<k.\ \
\Sigma(l)_j = \vdenot{T}_\rho(j,\cdot,\cdot)
\end{array}
\]
\vspace{-1.2em}
\[
\begin{array}{r@{\ }c@{\ }l}
\vdenot{\TBOOL}_\rho(k,\Sigma,H) &=&
  \set{a\mid \resolve(H,a)\in\set{\TTRUE,\TFALSE}}\\[3pt]
\vdenot{\REF[T]\capt C}_\rho(k,\Sigma,H) &=&
  \set{l\mid H(l)\ \text{is a cell} \wedge l\in\fdenot{C} \wedge \Sigma(l)\approx_k\vdenot{T}_\rho}\\[3pt]
\vdenot{(\forall(x:T_1)E)\capt C}_\rho(k,\Sigma,H) &=&
  \{\,a\mid \resolve(H,a)=\lambda(x:T_1')\,t\ \text{with}\
  R_0\subseteq\fdenot{C},\ \text{and}\\
&&\quad [x:=l]\,t\in\edenot{E}^{R_0}_{\rho[x\mapsto l]}(j,\Sigma',H')\
  \text{for every}\ j\leq k\\
&&\quad \text{and}\
  \wext{\Sigma'}{H'}{\lfloor\Sigma\rfloor_j}{H}\ \text{with}\
  \memtyped{j}{\Sigma'}{H'}\ \text{and}\ R_0\ \text{live in}\ H',\\
&&\quad \text{and every argument}\
  l\in\vdenot{T_1}_\rho(j,\Sigma',H')\,\}\\[3pt]
\vdenot{\EXCAP{c_1,\ldots,c_n}T}_\rho(k,\Sigma,H) &=&
  \{\,a\mid \resolve(H,a)=\<C_1,\ldots,C_n,x\>\ \text{with}\\
&&\quad \fdenot{C_1},\ldots,\fdenot{C_n}\ \text{pairwise disjoint, and}\\
&&\quad x\in\vdenot{T}_{\rho[c_1\mapsto C_1,\ldots,c_n\mapsto C_n]}(k,\Sigma,H)\,\}
\end{array}
\]

\judgheader{Expression denotation}{t\in\edenot{E}^{R}_\rho(k,\Sigma,H)}
\vspace{-1.4em}
\[
\begin{array}{l}
t\in\edenot{E}^{R}_\rho(k,\Sigma,H)
\ \coloneqq\
\safe{k}{H}{t}
\ \wedge\
\forall\,\tau,H',a.\ \
\cfg{H}{t}\evalsto{\tau}\cfg{H'}{a}
\ \wedge\
\readsof{\tau}<k
\Longrightarrow\\
\qquad
\tauth{\tau}{R}
\ \wedge\
\exists\,\Sigma'.\
\wext{\Sigma'}{H'}{\lfloor\Sigma\rfloor_{k'}}{H}
\ \wedge\
\memtyped{k'}{\Sigma'}{H'}
\ \wedge\
a\in\vdenot{E}_\rho(k',\Sigma',H')
\qquad\text{where}\ k' = k-\readsof{\tau}
\end{array}
\]

\judgheader{Semantic typing}{\semtyp{C}{\G}{t}{E}}
\vspace{-1.4em}
\[
\semtyp{C}{\G}{t}{E}
\ \coloneqq\
\forall k,\ \text{world}\ (\Sigma,H),\ \text{and}\
  \rho\ \text{realizing}\ \G\
  \text{with}\ \fdenot{C}\ \text{live in}\ H.\ \
t\rho\in\edenot{E}^{\fdenot{C}}_\rho(k,\Sigma,H)
\]

\vspace{-0.5em}
\caption{The logical model of \corecapybara{}: denotations of capture
sets, types, and expressions, and the semantic typing judgment.
$\Sigma(l)_j = \vdenot{T}_\rho(j,\cdot,\cdot)$ abbreviates the
pointwise agreement
$\Sigma(l)_j\,(\Sigma',H',a) \iff a\in\vdenot{T}_\rho(j,\Sigma',H')$
for all $\Sigma'$, $H'$, $a$.
The value denotation shows a selection of its clauses
(\Cref{app:metatheory}).}\label{fig:model}
\end{figure*}
 
\Cref{fig:model} collects the definitions of the model.
First, a capture set $C$ denotes a \emph{footprint} $\fdenot{C}^H$,
a capture set whose elements mention only heap locations
(\Cref{sec:core});
we omit the heap annotation when it is clear.
A reference cell contributes itself but not its content,
and a stored value contributes the footprint recorded when it was
lifted, so the interpretation closes transitively over the heap;
applying the mode $\mu$ restamps every element a location contributes
(\Cref{sec:formal}).
The use set's footprint thus lists every cell a term may touch
and at what mode.

Next, a type $T$ denotes a predicate $\vdenot{T}$ on answers,
relative to a \emph{world} $(\Sigma,H)$:
the runtime heap $H$ and a \emph{store typing} $\Sigma$ assigning each
reference cell a predicate for its content.
The step index $k$ breaks the circularity between store typing and
type interpretation~\cite{DBLP:journals/toplas/AppelM01,ahmed2004semantics}
($\mathit{Rel}_k$ in \Cref{fig:model}).
Evaluation moves to \emph{future} worlds, related by the extension
$\wext{\Sigma'}{H'}{\Sigma}{H}$:
it preserves the store typing and stored values,
while a reference cell may change its content and may die but not
revive; denotations are stable under extension.
$\vdenot{T}$ is also relative to a \emph{semantic environment} $\rho$,
mapping term variables to locations, type variables to predicates,
and capture variables to footprints.
Altogether, $a\in\vdenot{T}_\rho(k,\Sigma,H)$ reads: the answer $a$
behaves as a $T$-value at $(\Sigma,H)$, for runs performing fewer
than $k$ reads;
the expression denotation below makes this reading precise.

An answer names a value either directly or through a location, and
the partial function $\resolve$ recovers it.
\Cref{fig:model} shows a selection of the clauses;
the remaining ones follow the same patterns (\Cref{app:metatheory}).
A reference is a cell location in the footprint $\fdenot{C}$ of its
capture set, and the denotation of its content type agrees with the
predicate the store typing assigns,
$\Sigma(l)\approx_k\vdenot{T}_\rho$, at every lower index:
the forward direction of the agreement types the value a read
returns, the backward direction lets a write re-establish a
well-typed store, and exposing contents only below the current index
is what makes a read consume one unit of $k$.

The clause for functions shows the Kripke structure:
for every argument, the body is judged in every future world,
so a stored function remains usable however memory evolves.
The \emph{truncation} $\lfloor\Sigma\rfloor_j$ turns a world at index
$k$ into one at index $j$;
$\memtyped{j}{\Sigma'}{H'}$ states that the world is
\emph{well-typed}, chiefly that every live cell's content satisfies
its assigned predicate;
and a footprint is \emph{live} when no cell it mentions is dead,
so a call requires the captured cells to be still allocated
(\Cref{app:metatheory}).
$\edenot{E}$ is the \emph{expression denotation}, defined below;
its budget, here the closure footprint $R_0$ of the abstraction's
capture annotation, upper-bounds the cells the body may touch once
applied.

An existential package $\EXCAP{c_1,\ldots,c_n}T$ requires pairwise disjoint
witness footprints.

The expression denotation $\edenot{E}^{R}$ lifts $\vdenot{E}$ from
answers to expressions:
an expression inhabits it when it evaluates to an answer in
$\vdenot{E}$ while accessing only the locations in $R$.
$\readsof{\tau}$ counts the read events of $\tau$, and
$\safe{k}{H}{t}$ states that evaluation from $\cfg{H}{t}$ does not
get stuck within $k$ reads (\Cref{app:metatheory}).
The trace clause $\tauth{\tau}{R}$ (\Cref{sec:core}) makes the use
set operational: every access and every deallocation a run performs
is charged to the budget at its declared mode.

The evaluation relation here is the sequentially scheduled big-step
semantics (\Cref{sec:core}), so the model constrains sequential runs
only.
Its guarantees extend to arbitrary interleavings through confluence
(\Cref{thm:confluence}): an interleaved run joins against a
sequential run that the model provides, and either steps further or
coincides with its answer up to a renaming of locations
(\Cref{app:meta:adequacy}).

Finally, semantic typing $\semtyp{C}{\G}{t}{E}$ lifts the
interpretation to open terms, quantifying over every world and every
environment $\rho$ that \emph{realizes} $\G$:
each binding inhabits its denotation, and distinct consumable capture
variables denote disjoint footprints (\Cref{app:metatheory}).

\subsection{Soundness and Memory Safety}

\begin{theorem}[Fundamental Theorem]\label{thm:fundamental}
If $\typ{C}{\G}{t}{E}$, then $\semtyp{C}{\G}{t}{E}$.
\end{theorem}
The proof is by induction on typing, with one compatibility lemma per rule.

\begin{theorem}[Type Soundness]\label{thm:soundness}
Let $t$ be well typed and $H$ an initial heap providing the capabilities
declared by its context (\Cref{app:meta:adequacy}).
If $\cfg{H}{t}\redsto{\tau}\cfg{H'}{t'}$,
then $t'$ is an answer or $\cfg{H'}{t'}$ can step.
\end{theorem}

These live-cell requirements rule out use-after-free and double-free.

\begin{corollary}[Memory Safety]\label{cor:memsafe}
A well-typed program exhibits no use-after-free, and it frees each
allocation at most once.
\end{corollary}

The model also yields \emph{immutability}:
a read-only run cannot change the state it started with.
\begin{theorem}[Immutability]\label{thm:purity}
If a well-typed closed program has a use set of kind $\RO$,
then every run of it to an answer leaves every initial cell
unchanged, in both content and liveness.
\end{theorem}

Allocation remains allowed: the program may create and use new cells freely;
immutability concerns only the initial cells.

\subsection{Separation and Concurrency}

The footprint reading makes the separation judgment precise.
\begin{theorem}[Separation]\label{thm:separation}
If $\sep{\G}{C_1}{C_2}$, then in every world realizing $\G$ the two
footprints are non-interfering:
$\nintf{\fdenot{C_1}}{\fdenot{C_2}}$.
\end{theorem}

Non-interference of footprints is exactly the side condition of the
parallel reduction rules (\Cref{sec:core}), so by \Cref{thm:separation} a
well-typed parallel composition never blocks on it.
The two scheduling theorems below assume no typing, only
\emph{well-formedness}: $\cfg{H}{t}$ is well-formed when every location
$t$ mentions is allocated in $H$.
They are stated up to \emph{trace equivalence}:
$\tau_1\treq\tau_2$ when both traces perform the same per-location
sequence of accesses and deallocations, ignoring events on cells the
trace itself allocated (\Cref{app:meta:std});
equivalent traces reorder only independent events.

\begin{theorem}[Standardization]\label{thm:standardization}
If $\cfg{H}{t}$ is well-formed and
$\cfg{H}{t}\redsto{\tau}\cfg{H'}{a}$ with $a$ an answer, then
$\cfg{H}{t}\seqredsto{\tau'}\cfg{H'}{a}$ for some $\tau'\treq\tau$.
\end{theorem}

\begin{theorem}[Confluence]\label{thm:confluence}
Let $\cfg{H}{t}$ be well-formed,
$\cfg{H}{t}\redsto{\tau_1}\cfg{H_1}{t_1}$, and
$\cfg{H}{t}\redsto{\tau_2}\cfg{H_2}{t_2}$.
Then there are runs
$\cfg{H_1}{t_1}\redsto{\sigma_1}\cfg{H_1'}{t_1'}$ and
$\cfg{H_2}{t_2}\redsto{\sigma_2}\cfg{H_2'}{t_2'}$
and a permutation $\pi$ of locations such that
$H_2' = \pi\,H_1'$, $t_2' = \pi\,t_1'$,
$\pi(\tau_1\,\sigma_1)\treq\tau_2\,\sigma_2$, and
$\readsof{\tau_1\,\sigma_1} = \readsof{\tau_2\,\sigma_2}$.
In particular, two runs that reach answers reach the same heap and
the same answer up to a renaming of freshly allocated locations.
\end{theorem}

Confluence makes scheduling unobservable in well-typed runs, yielding
data-race freedom.

\subsection{Translating \capybara{} to \corecapybara{}}\label{sec:meta:translation}

The translation $\embed{\cdot}$ explains the separation and freshness
reasoning of the surface calculus in terms of the core constructs
that our logical model grounds,
and carries the formal results of this section over to \capybara{}.
This subsection presents its shape;
\Cref{app:translation} gives the full definitions and proofs.

Types translate in two variants.
The \emph{plain} translation $\embed{T}$ serves domains and other
invariant or contravariant positions, where $\FRESH$ cannot occur.
The \emph{result} translation $\eemb{T}$ closes $\FRESH$ occurrences into
existentials in covariant positions, including let-bound and function-result
types:
\[
\eemb{T} =
\EXCAP{c_1,\ldots,c_n}\,\embed{\freshinst{T}{\set{c_1},\ldots,\set{c_n}}},
\]
binding one capture variable per $\FRESH$ occurrence, where
$\freshinst{T}{\cdots}$ instantiates the occurrences left to right
(\Cref{sec:formal:typing}). A $\FRESH$-free type translates plainly.
Existential quantification is thus the meaning of $\FRESH$,
and universal quantification the meaning of $\ANY$.
Function types show both, together with the insertion of modal types:
$A = (\forall(x:T)U)\capt C$ translates to
\[
\embed{A} =
\Big(\forall[\cstar<:\top]\
\forall\big(x:\embed{\anyinst{T}{\set{\cstar}}}\big)\
([\Psi_A,\Phi_A]\,\eemb{U})\capt W_A\Big)\capt \embed{C}.
\]
The bound capture variable $\cstar$ instantiates the $\ANY$ of the
domain, so the caller supplies the argument's witness explicitly.
The result becomes a modal type carrying the manufactured lock
$\Psi_A = (\set{\cstar},\,\langle R_A\rangle)$, where $R_A$ resolves the
arrow's own capture annotation and domain captures to their roots and
$\langle R_A\rangle$ lists them as one entry per root: the parameter is
separate from each root the function reaches, and those roots are pairwise
separate. The lock is \emph{shallow}: it reads only the arrow's own
annotation and domain, and a separation that a nested closure relies on
lives in the nested arrow's own lock, paid where that arrow is eliminated.
The annotation $W_A = \set{\cstar}\cup\embed{R_A}$ is the same root set
read as a capture set, the resolved body charge of \ruleref{abs}.
For $c<:\top$ and $f:\REF[\TBOOL]\capt\set{c}$, the three binder forms
instantiate this scheme as follows, with $R_A=\roots{\G}{\set{f}}=\set{c}$:
{\small
\[
\begin{aligned}
\text{\bfseries Term}\quad
&\embed{(\forall(x:\top\capt\set{\ANY})\,\TUNIT)\capt\set{f}}\\[-2pt]
&\quad=
\Big(\forall[\cstar<:\top]\,
\forall(x:\top\capt\set{\cstar})\,
([\Psi,\emptyset]\,\TUNIT)\capt\set{\cstar,c}\Big)\capt\set{f},
\qquad
\Psi=(\set{\cstar},\set{c})
\\[5pt]
\text{\bfseries Capture}\quad
&\embed{(\forall[c'<:\RO]\,\TUNIT)\capt\set{f}}\\[-2pt]
&\quad=
\Big(\forall[c'<:\top]\,
([\Psi,\Phi]\,\TUNIT)\capt\set{c',c}\Big)\capt\set{f},
\qquad
\Psi=(\set{c'},\set{c}),
\quad
\Phi=(\set{c'}:\RO)
\\[5pt]
\text{\bfseries Consume}\quad
&\embed{(\forall(\CONSUME\,x:\top\capt\set{\ANY})\,\TUNIT)\capt\set{f}}\\[-2pt]
&\quad=
\Big(\forall\big(\<\cstar,x\>:\EXCAP{\cstar}\,\top\capt\set{\cstar}\big)\,
([\Psi,\emptyset]\,\TUNIT)
\capt\set{\cstar,\CONSUME\,\cstar,c}\Big)\capt\set{f},\\[-2pt]
&\hspace{2em}\Psi=(\set{\cstar,\CONSUME\,\cstar},\set{c}).
\end{aligned}
\]
}
Here $\ANY$ introduces a witness separated from $f$'s root, an $\RO$ bound
translates to $\top$ while $\Phi$ keeps its read-only assumption, and a
consume parameter becomes a consumer that may consume the argument's root.

Term translation is directed by the typing derivation: a surface
variable may be typed by \ruleref{var}, \ruleref{readonly}, or
\ruleref{fresh}, and the translated term depends on which.
Write $\embed{\G}^{\LOCK}$ for the translated context interleaved with
the locks that the enclosing translated binders push.

\begin{theorem}[Typability Preservation]\label{thm:translation}
If $\typ{C}{\G}{t}{T}$ in \capybara{}, then
$\typ{C'}{\embed{\G}^{\LOCK}}{\embed{t}}{\eemb{T}}$ in
\corecapybara{}, for some $C'$ with
$\roots{\embed{\G}^{\LOCK}}{C'}\sqsubseteq
 \roots{\embed{\G}^{\LOCK}}{\embed{C}}$.
\end{theorem}
The use sets agree up to root covering:
$P\sqsubseteq P'$ when $P'$ covers every element of $P$ at an equal or
stronger mode.
Most proof cases are homomorphic;
the interesting ones synthesize the explicit constructs of
\Cref{sec:core}.
A \ruleref{fresh} occurrence, retyping $x$ from
$T[\FRESH\leadsto D_1,\ldots,D_n]$ to $T$, synthesizes a pack that
consumes the witnesses,
$\embed{x} = \<\embed{D_1},\ldots,\embed{D_n},x\>$, and a
$\FRESH$-carrying \ruleref{let} translates to an unpack.
A \ruleref{readonly} occurrence synthesizes a reader value,
$\embed{x} = \RO\,x$.
An application supplies the argument's $\ANY$-witness $D$ to the callee,
applies it, and unlocks the result:
$
\embed{x\,y} =
\LET x_1 = x[\embed{D}]\IN \LET x_2 = x_1\,y\IN \UNLOCK\,x_2.
$
The final unlock is where the separation premise of \ruleref{app} is
paid: the instantiated lock demands $\embed{D}$ separate from each root
of $R_A$, the translated image of
$\sep{\G}{D}{\,|(\forall(x:T)U)\capt C|\,}$ at the shallow part of the
spine, with the deeper parts paid at the unlocks of the nested arrows.
On the running example (\Cref{fig:running-example}), the second
application $g\,f_2$
expands to exactly this shape, and its unlock demands
$\sep{\G}{\set{f_2}}{\set{f_1}}$: the premise that \ruleref{sep-root}
discharged in the surface is re-derived by \rruleref{sep-consumable} on
the consumable roots that the translated unpacks bind
(\Cref{app:tr:example}).

With type-preserving translation,
type soundness, memory safety, immutability, and data-race freedom of \corecapybara{}
transfer to well-typed \capybara{} programs.
\section{Case Studies}\label{sec:case-studies}

\subsection{Implementation}\label{sec:impl}

We implemented \capybara{} in Scala~3 by extending its capture checker with
a separation checker.
It runs after capture-set computation and enforces separation and freshness.

The surface API uses three Scala markers: classes with tracked mutable state
extend \lstinline|Mutable|, mutating methods are marked \lstinline|update|,
and consuming methods \lstinline|consume|:
\begin{code}
class Counter extends Mutable:
  private var count: Int = 0
  update def inc(): Unit = count += 1
  def get: Int = count
  consume def finalize(): Int = count
\end{code}

With explicit receivers, the markers correspond to the modes of
\Cref{sec:formal}.
\lstinline|inc| takes its receiver at the default read-write permission,
$\textsf{inc} : \forall(\textsf{this} : \textsf{Counter}\capt\set{\ANY})\ \TUNIT$.
\lstinline|get| takes it read-only at
$\RO\,\textsf{Counter}\capt\set{\RO\,\ANY}$ and is callable through read-only
views.
\lstinline|finalize| takes it at consume mode,
$\textsf{finalize} : \forall(\CONSUME\,\textsf{this} : \textsf{Counter}\capt\set{\ANY})\ \textsf{Int}$:
a call consumes the receiver and kills its aliases.

The read-only qualifier \lstinline|ro| of \Cref{sec:informal:readonly}
does not appear in Scala syntax.
Instead, the checker infers permissions from capture sets, which we call
\emph{mode inference}.
A type whose capture set is in read-only mode is itself read-only:
\lstinline|Counter^{any.rd}| is understood as
$\RO\,\textsf{Counter}\capt\set{\RO\,\ANY}$.
A read-only capture set of a \lstinline|Mutable| type cannot be widened to a
read-write one, as this would turn a read-only view into a read-write reference.

Separation and consume checking follow the root-based scheme of
\Cref{sec:formal:consume}.
The checker resolves every capture set to its \emph{peaks},
the compiler's term for roots (\Cref{sec:informal:separation}).
Two read-only peaks may overlap; an overlap in which one side writes is an
error.
For consume, the checker maintains a set of killed peaks: a
\lstinline|consume| call adds the peaks of its argument, and any later
expression using one is rejected.

\subsection{Mutable Builder of Immutable Data Structures}\label{sec:case-studies:builder}

Immutable data structures such as strings or lists are ubiquitous in
functional programming and immune to data races.
Within separation-checked Scala, immutable lists are \emph{pure}: no
separation constraint applies, and they can be freely shared.

Efficient construction often uses a mutable \emph{builder}, which must neither
be shared nor reused after returning an immutable result. In Scala~3:

\begin{code}
trait ListBuilder[T]:
  update def append(item: T): Unit
  consume def complete(): List[T]
\end{code}

The \lstinline|update| marker on \lstinline|append| requires an exclusive,
read-write reference, preventing mutable aliases across the call.
\lstinline|complete| consumes the builder and disables later access.

We can go one step further:
for data structures such as arrays and byte buffers,
we can obtain an immutable version
simply by ``freezing'' the object, irrevocably dropping its mutation
capability, similar to
the \texttt{freeze} operation of Pony \cite{pony} and Project Verona \cite{verona}:
\begin{code}
def freezeBuffer(consume buf: Buffer^): Buffer^{} = ...
\end{code}
A \lstinline|Buffer| is a mutable data structure;
freezing returns it as a pure one.
The result type \lstinline|Buffer^{}| carries the empty capture set,
so the frozen buffer is untracked and,
like any immutable value, may be shared freely.
This is safe for two reasons:
the parameter is marked \lstinline|consume|,
so no other reference can reach the buffer after the call;
and mode inference makes the result read-only,
since the empty capture set is in read-only mode.
Nothing can write to the frozen buffer through the old references or
the new one, so it is effectively immutable.

\subsection{Fearless Concurrency}

\subsubsection{Structured Parallelism}\label{sub:concurrency_par}

The library combinator \lstinline|runParallel| of
\Cref{sec:informal:separation} provides data-race-free structured parallelism:
it runs its two argument functions on separate threads,
requires their capture sets to be separate,
and suspends the caller until both return,
so the concurrent computations join before it does.

As an example, a parallel quicksort uses a primitive that splits a
\lstinline|Mutable| array into two halves:
\begin{code}
class Array[T] extends Mutable:
  update def splitAt(mid: Int)(f: (Array[T]^, Array[T]^) => Unit): Unit
\end{code}
\lstinline|splitAt| passes the two halves of the receiver to
\lstinline|f| as two array views.
Their types carry distinct \lstinline|any|s,
so \lstinline|f| may treat the halves as separate.

A parallel quicksort runs its two recursive calls concurrently:
\begin{code}
def quickSort[T](arr: Array[T]^, ord: (T, T) ->{any.rd} Boolean): Unit =
  if arr.len > 1 then
    val middle = partition(arr, ord)
    arr.splitAt(middle): (left, right) =>
      runParallel(() => quickSort(left, ord), () => quickSort(right, ord))
\end{code}
The closures capture \lstinline|left| and \lstinline|right| read-write.
They share \lstinline|ord|, whose capture set \lstinline|{any.rd}| contains
only read-only capabilities, so it can safely run in parallel with itself.

\subsubsection{Fork-Join Parallelism}

Fork-join concurrency spawns a computation that outlives the
expression starting it:
the new thread runs until it is joined, possibly never.
Conceptually, spawning \lstinline|f| runs
\lstinline|runParallel(f, rest)| against the rest of the program;
since \lstinline|rest| is unbounded,
nothing the spawned closure captures may ever be used again,
so the parameter is marked \lstinline|consume|:
\begin{code}
def spawn(consume f: () => Unit): ThreadHandle
\end{code}
A common use case is concurrent request handling in a server:
\begin{code}
trait Server:
  def start(address: String) =
    os.listen(address) // open a TCP server
      .onRequest((consume req: Request^) => spawn(() => handleRequest(req)))
  def handleRequest(r: Request^): Unit
\end{code}
Each request is handled on a thread of its own.
The \lstinline|onRequest| callback takes ownership of its request
(\lstinline|consume req|),
and \lstinline|spawn| consumes the handler closure
together with the request it captures.

\section{Discussion}\label{sec:discussion}

\rwpar{Formal scope and implementation}
Semantic soundness for \corecapybara{} establishes type safety, memory safety,
immutability, and separation; confluence establishes data-race freedom.
The type-preserving translation maps well-typed \capybara{} programs into the
core.
We do not define a separate semantics for the surface calculus or prove
semantic preservation.
The calculi omit Scala features such as classes, variance, exceptions, and
concrete array layouts.
The Scala~3 implementation handles them through its integration with the
existing capture checker (\Cref{sec:impl}).
We do not relate the compiler implementation formally to the declarative
calculi, nor prove checker completeness or rule decidability.
For concurrency, the calculi model structured parallelism, in which both
branches join before evaluation continues.
In Scala, APIs such as \lstinline|spawn| consume a closure whose thread may
outlive the call.
The checker enforces consumption, but the thread's behavior and
synchronization primitives are outside the formal model.

\rwpar{Limits of selective enforcement}
A capability may be consumed only through a fresh, dominating reference, as
in external uniqueness~\cite{DBLP:conf/ecoop/ClarkeW03}.
A component of a shared structure therefore cannot be consumed directly.
APIs for such structures must expose operations such as builders, freezing,
or disjoint splitting (\Cref{sec:case-studies}).
Most permissions are inferred; \lstinline|Mutable| classes, receiver modes,
\lstinline|consume| parameters, and \lstinline|fresh| results mark the
contracts that restrict aliasing.
We have not measured the annotation cost of adding these contracts to
existing capture-checked code.
Synchronized sharing still requires specialized or trusted APIs.

\rwpar{Why a semantic proof}
Our main properties concern executions:
separated terms have non-interfering footprints,
read-only runs preserve every initial cell,
and well-typed programs cannot use or free a dead cell.
A syntactic proof would need corresponding heap and liveness invariants at
every reduction step.
Consumption makes these invariants nontrivial:
\rruleref{let} types its continuation in the rewritten context
$\G\ominus C_1$, so contexts do not evolve monotonically,
and deallocated cells remain in the heap, so attempts to access or deallocate
them again are stuck.
The logical relation instead includes footprints, liveness, and killing in
the interpretation of types; the fundamental theorem establishes these
invariants rule by rule.
Semantic typing is also more expressive than the syntactic rules.
Suppose a primitive \lstinline|ne| that tests two references for identity.
The following program is not syntactically well-typed, since nothing in
scope separates \lstinline|a| from \lstinline|b|:
\begin{code}
(a: Ref^, b: Ref^) => if a ne b then par (a := 0), (b := 1)
\end{code}
It can nevertheless be proved semantically well-typed: when the test
succeeds, \lstinline|a| and \lstinline|b| are distinct cells, their
footprints do not interfere, and the parallel composition runs without
getting stuck (\Cref{sec:core}).

\rwpar{Higher-order store and cycles}
The calculus supports higher-order store: references may contain closures and
other references.
Given $x: T$, the program
$\LET a = \ALLOC x \IN \LET b = \ALLOC a \IN \ldots$
binds $b : \REF[\REF[T]\capt\set{a}]\capt\set{c_b}$ where $c_b$ is a fresh root:
the content type names the stored reference $a$ in its capture set,
and reading $b$ returns $a$ at that type.
It does not express cyclic reference graphs.
At allocation, a reference's content type and capture set may mention only
capabilities already in scope, so it may contain only references allocated
earlier.
Supporting cycles remains future work.

\rwpar{Path-dependent capabilities}
The calculi support limited path dependence: capture sets may mention term
variables, and dependent function codomains may name their parameters
(\Cref{sec:formal}), but path selection is absent.
The Scala~3 checker treats capabilities as paths and resolves them by prefix
coverage: \lstinline|a| covers \lstinline|a.b|, while siblings
\lstinline|a.b| and \lstinline|a.c| are treated as separate.
Extending \corecapybara{} with path selection based on Dependent Object Types
(DOT)~\cite{DBLP:conf/birthday/AminGORS16,DBLP:conf/oopsla/RompfA16} is future
work.

\rwpar{Why relational modalities}
Scala APIs mix curried and uncurried functions.
Separation should be invariant under currying:
constraints between parameter lists must survive partial application.
Capture checking already supplies unary capture information.
We considered and rejected reachability types as a basis for separation and
freshness in Scala~\cite{DBLP:journals/pacmpl/WeiBJBR24}.
That calculus reconstructs separation at each application, but formalizes only
unary application, not uncurried functions or Scala's multiple parameter lists.
A later logical-relations model revises the abstraction
rule~\cite{DBLP:journals/pacmpl/BaoJ0BR25}, suggesting that this subtle design
point remains unsettled.
Neither account treats Scala's curried and uncurried calls uniformly.
\capybara{} instead models separation constraints, for curried and
uncurried functions alike, as \emph{relational
modalities}~\cite{DBLP:journals/tocl/NanevskiPP08}.
The encoding is stable and sound:
constraints survive partial application by construction,
and the logical model interprets the modality directly
(\Cref{sec:metatheory}).
The surface stays modality-free:
the translation computes every lock and unlock.
Both the calculus and the Scala checker handle the full application spine,
so \lstinline|f(a, b)| and \lstinline|f(a)(b)| enforce the same separation.
\section{Related Work}\label{sec:related-work}

\rwpar{Capture Checking and Separation}
\capybara{} builds on capture checking and capturing types, which record the
capabilities that a value may use
\cite{DBLP:journals/toplas/BoruchGruszeckiOLLB23,whatisinthebox}.
Capture checking is implemented in Scala~3~\cite{dottycc}
and was first proposed as a foundation for exception safety
\cite{DBLP:conf/scala/OderskyBBLL21}.
Its boxing discipline draws on contextual modal type
theory~\cite{DBLP:journals/tocl/NanevskiPP08},
and spores anticipated restrictions on closure capture in
Scala~\cite{miller14spores}.
Degrees of separation added relative, per-binding constraints to capture
checking, obtaining confluence for fork-join
parallelism~\cite{SeparationChecking,DBLP:journals/pacmpl/XuBO24}.
\capybara{} defines separation through roots: distinct roots are separate
by construction, and capture sets are separate when they resolve to
non-interfering roots, without per-binding declarations.
It also adds freshness, consume, read-only access, and a semantic account of
\corecapybara{}.

\rwpar{Ownership, Borrowing, and Uniqueness}
Rust is the standard example of ownership and borrowing in a practical
language,
supporting memory safety and fearless concurrency through a pervasive alias
discipline~\cite{DBLP:conf/sigada/MatsakisK14,rust,DBLP:journals/cacm/JungJKD21,DBLP:journals/jfp/JungKJBBD18,weiss_oxide_2020}.
Swift and Mojo also support ownership and
borrowing~\cite{swift-noncopyable,swift-borrowing,mojo-ownership}.
Classic ownership and uniqueness systems enforce a distinguished owner,
confinement, unique access, or external uniqueness
\cite{DBLP:conf/ecoop/NobleVP98,DBLP:conf/oopsla/ClarkePN98,DBLP:series/lncs/ClarkeOSW13,boyapati2002ownership,DBLP:journals/mscs/BarendsenS96,DBLP:conf/oopsla/Hogg91,DBLP:conf/ecoop/ClarkeW03};
Mezzo makes permissions the basis of the language
\cite{DBLP:journals/toplas/BalabonskiPP16}.
Adoption and focus recover temporary exclusive access to an aliased value by
setting its adopter aside~\cite{DBLP:conf/pldi/FahndrichD02}.
Tempered domination later applies this mechanism to fearless
concurrency~\cite{DBLP:conf/pldi/MilanoTM22}.
Previous Scala work used capabilities for unique references and
actor isolation~\cite{ECOOP2010-Object-OrientedProgram./Haller10,DBLP:conf/oopsla/HallerL16}.
\capybara{} permits ordinary sharing and uses tracked capture sets to establish
the exclusivity required for separation or consumption, without a lexical
focus construct.

\rwpar{Linear and Modal Retrofitting}
Linear and affine types constrain use of values
\cite{DBLP:journals/tcs/Girard87,DBLP:conf/ifip2/Wadler90},
graded modal types generalize these restrictions
\cite{DBLP:journals/pacmpl/OrchardLE19},
and recent work connects linearity, uniqueness, and borrowing
\cite{DBLP:conf/esop/MarshallVO22,DBLP:journals/pacmpl/MarshallO24}.
Linear Haskell supports ordinary and linear functions side by side
\cite{DBLP:journals/pacmpl/BernardyBNJS18}.
\capybara{} retrofits substructural reasoning by constraining access through
tracked capabilities rather than value use.
\corecapybara{}'s modality $[\Psi,\Phi]E$ records separation constraints between
program-specific capture sets and read-only capability bounds.
OxCaml modes instead qualify individual values along fixed axes such as
locality, uniqueness, and contention
\cite{oxidizingocaml,DBLP:journals/pacmpl/GeorgesPEWDECPD25};
modal effect types index computations by effect contexts
\cite{tang25modal,tang26rows}.
These indices qualify one value or computation at a time; they do not directly
relate the footprints of two otherwise freely aliased values.
Like modal effect types, \corecapybara{} uses Fitch-style locks; its unlocks
discharge program-specific constraints rather than compare modes in a fixed
lattice, as OxCaml submoding does.
The translation inserts lock and unlock forms; \capybara{} has no modal syntax.
Whether these constraint-indexed modalities admit a standard multimodal account
remains open.

\rwpar{Capabilities and Effects}
Capability systems go back to unforgeable tokens for controlling
authority~\cite{DBLP:journals/cacm/DennisH66};
object-capability languages make authority follow from holding object
references~\cite{objectcapabilites}.
Capabilities also induce an effect discipline: the effects of a computation
are bounded by the capabilities it holds
\cite{DBLP:conf/icfem/CraigPGA18,DBLP:conf/ecoop/Gordon19}.
Wyvern makes this idea explicit with object capabilities and
path-dependent effects
\cite{DBLP:conf/ecoop/MelicherSPA17,DBLP:journals/toplas/MelicherXZPA22}.
Second-class values restrict capability scope, while effect handlers reflect
capabilities in computation types
\cite{osvald2016gentrification,DBLP:conf/scala/OsvaldR17,DBLP:conf/ecoop/XhebrajB0R22,DBLP:journals/pacmpl/BrachthauserSO20,DBLP:journals/pacmpl/BrachthauserSLB22}.
Type-and-effect systems and row-polymorphic effect languages track effects
directly
\cite{DBLP:conf/popl/LucassenG88,DBLP:journals/jfp/TalpinJ92,leijen2016koka}.
\capybara{} builds on this capability-based account of effects, adding static
control of exclusive and read-only access.

\rwpar{Separation, Immutability, and Regions}
Separation logic and syntactic systems for interference control use
disjointness and permissions to reason about memory access
\cite{DBLP:conf/popl/Reynolds78,reynolds1989syntactic,o1999syntactic,DBLP:conf/csl/OHearnRY01,DBLP:journals/jfp/OHearn03,DBLP:conf/lics/Reynolds02,DBLP:conf/sas/Boyland03}.
In \capybara{}, separated capture sets denote non-interfering footprints and
justify parallel composition.
Reference immutability restricts mutation through read-only references
\cite{DBLP:conf/pldi/FosterFA99,ACMTrans.Program.Lang.Syst./FosterRJ06,DBLP:conf/oopsla/TschantzE05,DBLP:conf/oopsla/GordonPPBD12,DBLP:conf/ecoop/DortL19,DBLP:journals/pacmpl/LeeL23,DBLP:conf/ecoop/BoylandNR01};
\capybara{} applies the same restriction to capabilities; its separation
judgment permits overlapping read-only footprints.
Region systems and Cyclone provide static memory management through
regions and region capabilities
\cite{DBLP:conf/popl/TofteT94,DBLP:journals/iandc/TofteT97,DBLP:conf/pldi/GrossmanMJHWC02}.
Typed memory-management calculi add static capabilities,
alias types,
and linear regions
\cite{DBLP:conf/popl/CraryWM99,DBLP:journals/toplas/WalkerCM00,DBLP:conf/esop/FluetMA06}.
\capybara{} proves no-use-after-free and no double-free using liveness in the
semantic model, without region handles or region inference.

\rwpar{Capabilities for Data-Race Freedom}
Reference-capability systems such as Pony and Verona assign permissions that
control aliasing and concurrency
\cite{pony,DBLP:journals/pacmpl/ClebschFDYWV17,verona}.
Other systems obtain fearless concurrency through uniqueness with
immutability, branded permissions, affine regions with isolated
fields, or capability-based actors
\cite{DBLP:conf/oopsla/GordonPPBD12,DBLP:journals/pacmpl/YanovskiDJD21,DBLP:conf/pldi/MilanoTM22,DBLP:conf/ecoop/CastegrenW16}.
In \capybara{}, data-race freedom follows from separation and confluence:
parallel computations have non-interfering footprints, and their scheduling
is unobservable.
Separation is checked per captured capability, rather than per object graph,
actor heap, thread reservation, or region.

\rwpar{Semantic Soundness}
Logical relations and the logical approach to type soundness support
semantic accounts of higher-order state,
memory safety,
and concurrency
\cite{DBLP:journals/toplas/AppelM01,ahmed2004semantics,DBLP:conf/esop/Ahmed06,DBLP:conf/popl/AhmedDR09,DBLP:journals/jfp/JungKJBBD18,DBLP:journals/jacm/TimanyKDB24,DBLP:journals/pacmpl/BaoJ0BR25}.
Our logical relation interprets types as predicates over worlds and capture
sets as access-tagged footprints.
Its fundamental theorem yields type safety, memory safety, and immutability;
the same model validates separation.
The core development is mechanized in Lean~4.
Future work is to interpret \capybara{} types in a separation logic, possibly
through implicit dynamic frames~\cite{DBLP:conf/ecoop/SmansJP09}, and relate
that interpretation to Iris~\cite{DBLP:conf/popl/JungSSSTBD15}.
\vspace{-0.25\baselineskip}
\section{Conclusion}\label{sec:conclusion}

We presented System~\capybara{}, extending capturing types with selective
substructural reasoning.
Ordinary aliasing remains unrestricted; APIs require separation, read-only
access, or consumption only where needed.
Separation gives non-interfering footprints, while consumption prevents later
use through aliases.

A type-preserving translation maps \capybara{} to \corecapybara{}, whose
semantic soundness proof we mechanize in Lean~4.
It establishes type safety, memory safety, immutability, and data-race freedom.
We implement \capybara{} as the new separation checker in the Scala~3 compiler,
built on its existing capture checker.
The implementation gives Scala itself Rust-style resource safety and fearless
concurrency without global ownership invariants: pervasive sharing remains the
default, and substructural constraints fit its generic and higher-order idioms.
 
\clearpage
\printbibliography

\clearpage
\appendix
\counterwithin{figure}{section}
\crefalias{section}{appendix}
\crefalias{subsection}{appendix}
\crefalias{subsubsection}{appendix}
\section{Complete Definitions of System \corecapybara{}}\label{app:corecapybara}

This appendix presents the full formal definitions of the core calculus
\corecapybara{}. The definitions follow the Lean mechanization; they depart
from it only in notation and minor presentational simplifications, as stated
in the conventions below.

\subsection{Syntax}

\begin{wide-rules}\noindent
	{\small\begin{multicols}{2}\noindent
		\begin{flalign*}
			x,\,y,\,z         \tag*{\textbf{Variable}}\\
			X                 \tag*{\textbf{Type Variable}}\\
			c                 \tag*{\textbf{Capture Variable}}\\
			s,\,t,\,u\coloneqq\ &           \tag*{\textbf{Term}}\\
			&x                                      \tag*{variable}\\
			&v                              \tag*{value}\\
			&x\,y                              \tag*{app.}\\
			&x\,t                              \tag*{consumer app.}\\
			&x[S]                              \tag*{type app.}\\
			&x[C]                              \tag*{capture app.}\\
			&\ALLOC x                              \tag*{alloc}\\
			&\READ x                              \tag*{read}\\
			&x:=y                              \tag*{write}\\
			&\DEALLOC x                              \tag*{dealloc}\\
			&\mblue{\UNLOCK\,x}                  \tag*{\mblue{unlock}}\\
			&\LET x = t \IN u                  \tag*{let}\\
			&\LET \<c_1,\ldots,c_n,x\> = t \IN u               \tag*{unpack}\\
			&\PAR t_1, t_2                  \tag*{parallel}\\
			&\IF x \THEN t_1 \ELSE t_2         \tag*{cond.}\\
			v\coloneqq\ &           \tag*{\textbf{Value}}\\
			&\lambda(x: T)t                      \tag*{term lambda}\\
			&\lambda(\<c,x\>: \EXCAP{c}T)t                      \tag*{consumer lambda}\\
			&\lambda[X<:S]t                      \tag*{type lambda}\\
			&\lambda[c<:B]t                      \tag*{capture lambda}\\
			&\RO\,x                      \tag*{reader}\\
			&\mblue{\LOCK[\Psi,\Phi]t}                      \tag*{\mblue{modal}}\\
			&\<C_1,\ldots,C_n,x\>                      \tag*{pack}\\
			&()\,\mid\,\TTRUE\,\mid\,\TFALSE    \tag*{constants}\\
			a\coloneqq\ &           \tag*{\textbf{Answer}}\\
			&x                      \tag*{variable}\\
			&v                      \tag*{value}\\
			m\coloneqq\ &           \tag*{\textbf{Mutability}}\\
			&\epsilon                      \tag*{writable}\\
			&\RO                        \tag*{read-only}\\
			\mu\coloneqq\ &           \tag*{\textbf{Access Mode}}\\
			&m                      \tag*{access}\\
			&\CONSUME                      \tag*{consume}\\
		\end{flalign*}
		\begin{flalign*}
			\mblue{\sepctx\coloneqq\ }&\mblue{C_1,\ldots,C_n}        \tag*{\mblue{\textbf{Separation Context}}}\\
			\mblue{\modectx\coloneqq\ }&\mblue{C_1:m_1,\ldots,C_n:m_n}        \tag*{\mblue{\textbf{Mode Context}}}\\
			\alpha\coloneqq\ &           \tag*{\textbf{Consume Mode}}\\
			&\epsilon                      \tag*{access-only}\\
			&\CONSUME                      \tag*{consumable}\\
			&\KILLED                      \tag*{killed}\\
			B\coloneqq\ &           \tag*{\textbf{Capture Bound}}\\
			&C                      \tag*{capture set}\\
			&\top                      \tag*{unbounded}\\
			\theta\coloneqq\ &           \tag*{\textbf{Capture}}\\
			&x                      \tag*{variable}\\
			&c                      \tag*{capture variable}\\
			C\coloneqq\ &\set{\mu_1\,\theta_1,\dots,\mu_n\,\theta_n}     \tag*{\textbf{Capture Set}}\\
			T,\,U\coloneqq\ &           \tag*{\textbf{Type}}\\
			&m\,S\capt C                    \tag*{capturing}\\
			&S                    \tag*{pure}\\
			R,\,S\coloneqq\ &           \tag*{\textbf{Shape Type}}\\
			&\top                    \tag*{top}\\
			&X                    \tag*{type variable}\\
			&\forall(x: T)E             \tag*{term function}\\
			&\forall[X<:S]E             \tag*{type function}\\
			&\forall[c<:B]E             \tag*{capture function}\\
			&\forall(\<c,x\>: \EXCAP{c}T)E             \tag*{consumer}\\
			&\mblue{[\Psi,\Phi]E}             \tag*{\mblue{modal}}\\
			&\REF[T]                       \tag*{reference}\\
			&\TBOOL\,\mid\,\TUNIT            \tag*{constant}\\
			E,\,F\coloneqq\ &           \tag*{\textbf{Existential Type}}\\
			&T \tag*{type}\\
			&\EXCAP{c_1,\ldots,c_n}T \tag*{exists}\\
			\G,\,\Delta\coloneqq\ &           \tag*{\textbf{Context}}\\
			&\emptyset                    \tag*{empty}\\
			&\G, x: T                    \tag*{term binding}\\
			&\G, X<:S                \tag*{type binding}\\
			&{\G, \alpha\,c<:B}                    \tag*{capture binding}\\
			&{\G, \mblue{\LOCK[\sepctx,\modectx]}}                    \tag*{\mblue{lock binding}}\\
	\end{flalign*}
	\end{multicols}}
	\vspace{-3em}
	\caption{Abstract syntax of System \corecapybara{}.}\label{fig:core-syntax}
\end{wide-rules}
 
\Cref{fig:core-syntax} presents the abstract syntax of System \corecapybara{}.
\corecapybara{} has
memory primitives for allocating, reading, writing,
and deallocating references,
reader values $\RO\,x$ that grant read-only access to a reference,
modal values $\LOCK[\Psi,\Phi]t$ with the eliminator $\UNLOCK\,x$,
parallel composition, and conditionals.
An existential type $\EXCAP{c_1,\ldots,c_n}T$ hides $n$ capture variables;
its introduction form is a pack $\<C_1,\ldots,C_n,x\>$ that supplies the $n$
witnesses, and its elimination forms are the unpack
$\LET\<c_1,\ldots,c_n,x\>=t\IN u$ and the consumer lambda
$\lambda(\<c,x\>:\EXCAP{c}T)t$, a first-class function that consumes a packed
argument.
Its application $x\,t$ passes the term $t$, which distinguishes it from an
ordinary application $x\,y$ of a variable to a variable.
A capture bound $B$ is a capture set or the trivial bound $\top$.
A capture binding $\alpha\,c<:B$ carries a consume mode $\alpha$:
an access-only variable ($\epsilon$) may only be accessed,
a consumable one ($\CONSUME$) may additionally be consumed,
and a killed one ($\KILLED$) has been consumed and may be neither
accessed nor consumed.
The mode $\KILLED$ does not occur in source programs;
it arises in the contexts of typing derivations
through the kill operation defined below.

\subsection{Static Semantics}

The static semantics consists of nine judgments:
subcapturing $\subs{\G}{C_1}{C_2}$,
capability kinding $\typs{\G}{C}{m}$,
subbounding $\subs{\G}{B_1}{B_2}$,
separation $\sep{\G}{C_1}{C_2}$,
disjointness $\disj{\G}{C_1}{C_2}$,
sequential composition $\seqcomp{\G}{C_1}{C_2}$,
satisfaction $\sat{\G}{[\Psi,\Phi]}$,
subtyping $\subs{\G}{E_1}{E_2}$,
and term typing $\typ{C}{\G}{t}{E}$.

Two conventions apply throughout.
First, all rules are stated up to $\alpha$-renaming,
and every type and capture set is well-scoped in the context where it occurs;
the mechanization enforces this with de Bruijn indices.
In particular, in rule \rruleref{abs} the capture set $C$ cannot mention
the bound variable $x$,
and in rules \rruleref{let} and \rruleref{unpack}
neither $C_2$ nor $E$ can mention the variables bound by the continuation.
Second, in the mechanization, abstractions, modal values, and parallel
compositions carry the capture sets appearing in their typing rules as
annotations; we omit these annotations, except the two capture sets on a
parallel composition, which the reduction rules consult and which we display in
\Cref{fig:core-smallstep}.
As in the surface language, we omit the default permission $\epsilon$ of a
capturing type $m\,S\capt C$ and write $S\capt C$.
We identify a pure capturing type $S\capt\set{}$ with the shape type $S$.

\subsubsection{Capture Set Operations and Predicates}

\rwpar{Mode application}
Applying an access mode to a capture set acts on every element:
$\epsilon\,C = C$;
$\RO\,C$ replaces the mutability of every element of $C$ by $\RO$,
leaving $\CONSUME$-marked elements unchanged;
and $\CONSUME\,C$ marks every element of $C$ with $\CONSUME$.
Mutabilities are ordered by $\RO\preceq\epsilon$.

\rwpar{Roots}
The \emph{roots} of a capture set resolve term variables through the context
until only capture variables remain:
\begin{align*}
\roots{\G}{\set{}} &= \set{} \\
\roots{\G}{C_1\cup C_2} &= \roots{\G}{C_1}\cup\roots{\G}{C_2} \\
\roots{\G}{\set{\mu\,c}} &= \set{\mu\,c} \\
\roots{\G}{\set{\mu\,x}} &= \mu\,\roots{\G}{C}
  \quad\text{where}\ x: m\,S\capt C\in\G
\end{align*}
A capture set containing only capture variables is a \emph{root set},
ranged over by $P$.
We write $P\vert_{\CONSUME}$ for the restriction of $P$ to its
$\CONSUME$-marked elements.
The mechanization names this function \textsf{peaks}.

\rwpar{Predicates}
The judgments on capture sets used by the rules are defined via roots:
\begin{itemize}
\item $\accessonly{\G}{C}$ holds iff $\roots{\G}{C}$ contains no
  $\CONSUME$-marked element.
  A bound is access-only if it is $\top$ or an access-only capture set.
\item $\consumable{\G}{C}$ holds iff every capture variable occurring in
  $\roots{\G}{C}$ is bound with mode $\CONSUME$ in $\G$.
\item $\accessible{\G}{C}$ holds iff no capture variable occurring in
  $\roots{\G}{C}$ is bound with mode $\KILLED$ in $\G$.
\end{itemize}

\rwpar{Killing}
The context $\G\ominus C$ kills the capture variables consumed by $C$:
\[
\G\ominus C \coloneqq
\G\big[\;\alpha\,c<:B \;\mapsto\; \KILLED\,c<:B
  \;\big|\; \CONSUME\,c\in\roots{\G}{C}\;\big],
\]
replacing the binding of every capture variable that the roots of $C$
mark consumed and leaving all other bindings unchanged.

\subsubsection{Subcapturing, Capability Kinding, and Subbounding}

\begin{figure*}[htbp]
\rulefigsetup

\judgheader{Subcapturing}{\subs{\G}{C_1}{C_2}}
\begin{multicols}{3}

\infrule[\rruledef{sc-trans}]
{\subs{\G}{C_1}{C_2}\andalso
  \subs{\G}{C_2}{C_3}}
{\subs{\G}{C_1}{C_3}}

\infrule[\rruledef{sc-elem}]
{C_1\subseteq C_2}
{\subs{\G}{C_1}{C_2}}

\infrule[\rruledef{sc-mode}]
{m_1\preceq m_2}
{\subs{\G}{m_1\,C}{m_2\,C}}

\infrule[\rruledef{sc-union}]
{\subs{\G}{C_1}{C}\andalso
  \subs{\G}{C_2}{C}}
{\subs{\G}{C_1\cup C_2}{C}}

\infrule[\rruledef{sc-var}]
{x:m\,S\capt C\in\G}
{\subs{\G}{\set{x}}{C}}

\infrule[\rruledef{sc-cvar}]
{\epsilon\,c<:C\in\G}
{\subs{\G}{\set{c}}{C}}

\infax[\rruledef{sc-ro}]
{\subs{\G}{\RO\,C}{C}}

\infrule[\rruledef{sc-ro-mono}]
{\subs{\G}{C_1}{C_2}}
{\subs{\G}{\RO\,C_1}{\RO\,C_2}}

\infrule[\rruledef{sc-consume-mono}]
{\subs{\G}{C_1}{C_2}}
{\subs{\G}{\CONSUME\,C_1}{\CONSUME\,C_2}}

\end{multicols}

\judgheader{Capability Kinding}{\typs{\G}{C}{m}}
\begin{multicols}{3}

\infax[\rruledef{k-empty}]
{\typs{\G}{\set{}}{m}}

\infrule[\rruledef{k-union}]
{\typs{\G}{C_1}{m}\andalso
 \typs{\G}{C_2}{m}}
{\typs{\G}{C_1\cup C_2}{m}}

\infrule[\rruledef{k-sc}]
{\subs{\G}{C_1}{C_2}\andalso
 \typs{\G}{C_2}{m}}
{\typs{\G}{C_1}{m}}

\infax[\rruledef{k-rw}]
{\typs{\G}{C}{\epsilon}}

\infrule[\rruledef{k-imm}]
{\LOCK[\Psi,\Phi]\in\G\andalso
 (C:\RO)\in\Phi}
{\typs{\G}{C}{\RO}}

\infax[\rruledef{k-ro}]
{\typs{\G}{\RO\,C}{\RO}}

\end{multicols}

\judgheader[Subcapturing extended to capture bounds plus the following rule:]{Subbounding}{\subs{\G}{B_1}{B_2}}
\infax[\rruledef{sb-top}]
{\subs{\G}{B}{\top}}

\vspace{-1.5em}
\caption{Subcapturing, capability kinding, and subbounding of System \corecapybara{}.}\label{fig:core-subcapturing}

\end{figure*}
 
Figure~\ref{fig:core-subcapturing} defines subcapturing, capability kinding, and
subbounding.
Subcapturing extends the surface rules with
\rruleref{sc-cvar}, which resolves a capture variable to its declared bound
and applies only to access-only bindings,
\rruleref{sc-ro}, which forgets a read-only restriction upward,
and the congruence rules \rruleref{sc-ro-mono} and \rruleref{sc-consume-mono}.
Capability kinding assigns a mutability to a capture set:
every set has kind $\epsilon$ by \rruleref{k-rw},
while kind $\RO$ certifies that the set grants at most read-only access,
either because the set is $\RO$-marked \rruleref{k-ro}
or because an enclosing lock asserts it \rruleref{k-imm}.
Subbounding lifts subcapturing to capture bounds,
with $\top$ above every bound.

\subsubsection{Separation, Sequential Composition, and Satisfaction}

\begin{figure*}[htbp]
\rulefigsetup

\judgheader{Separation}{\sep{\G}{C_1}{C_2}}
\begin{multicols}{2}

\infrule[\rruledef{sep-symm}]
{\sep{\G}{C_1}{C_2}}
{\sep{\G}{C_2}{C_1}}

\infrule[\rruledef{sep-union}]
{\sep{\G}{C_1}{C}\andalso
 \sep{\G}{C_2}{C}}
{\sep{\G}{C_1\cup C_2}{C}}

\infax[\rruledef{sep-empty}]
{\sep{\G}{\set{}}{C}}

\infrule[\rruledef{sep-ro}]
{\accessonly{\G}{C_1}\andalso
 \accessonly{\G}{C_2}\\
 \typs{\G}{C_1}{\RO}\andalso
 \typs{\G}{C_2}{\RO}}
{\sep{\G}{C_1}{C_2}}

\infrule[\rruledef{sep-sc}]
{\sep{\G}{C_1}{C_2}\andalso
 \subs{\G}{C_1'}{C_1}}
{\sep{\G}{C_1'}{C_2}}

\infrule[\rruledef{sep-lock}]
{\LOCK[\Psi,\Phi]\in\G\\
 C_1,\,C_2\ \text{are distinct entries of}\ \Psi}
{\sep{\G}{C_1}{C_2}}

\infrule[\rruledef{sep-consumable}]
{\CONSUME\,c_1<:B_1\in\G\andalso
 \CONSUME\,c_2<:B_2\in\G\\
 c_1\neq c_2}
{\sep{\G}{\set{\mu_1\,c_1}}{\set{\mu_2\,c_2}}}

\end{multicols}

\judgheader[The same rules as separation, dropping \rruleref{sep-ro} and \rruleref{sep-lock}.]{Disjointness}{\disj{\G}{C_1}{C_2}}
\judgheader{Sequential Composition}{\seqcomp{\G}{C_1}{C_2}}
\begin{multicols}{2}

\infrule[\rruledef{seq-sc}]
{\subs{\G}{C_1}{C_1'}\andalso
 \seqcomp{\G}{C_1'}{C_2}}
{\seqcomp{\G}{C_1}{C_2}}

\infrule[\rruledef{seq-union}]
{\seqcomp{\G}{C_1}{C_3}\andalso
 \seqcomp{\G}{C_2}{C_3}}
{\seqcomp{\G}{(C_1\cup C_2)}{C_3}}

\infrule[\rruledef{seq-access-only}]
{\accessonly{\G}{C_1}}
{\seqcomp{\G}{C_1}{C_2}}

\infrule[\rruledef{seq-sep}]
{\sep{\G}{C_1}{C_2}}
{\seqcomp{\G}{C_1}{C_2}}

\end{multicols}

\judgheader{Satisfaction}{\sat{\G}{[\Psi,\Phi]}}
\infrule[\rruledef{sat}]
{\typs{\G}{C}{m}\quad\text{for each}\ (C:m)\in\Phi\\
 \sep{\G}{C_1}{C_2}\quad\text{for each pair of distinct entries}\ C_1,\,C_2\ \text{of}\ \Psi}
{\sat{\G}{[\Psi,\Phi]}}

\vspace{-1.5em}
\caption{Separation, sequential composition, and satisfaction of System \corecapybara{}.}\label{fig:core-separation}

\end{figure*}
 
Figure~\ref{fig:core-separation} defines the three judgments that govern
aliasing and consumption.
Separation $\sep{\G}{C_1}{C_2}$ certifies that two capture sets do not
interfere:
two access-only sets of read-only kind are separated
\rruleref{sep-ro};
two distinct consumable capture variables are separated
\rruleref{sep-consumable};
a lock in the context contributes the separation assumptions it records
\rruleref{sep-lock};
and separation is preserved when the left-hand set shrinks to a subcapture
\rruleref{sep-sc}.
Disjointness $\disj{\G}{C_1}{C_2}$ is the fragment of separation obtained by
dropping \rruleref{sep-ro} and \rruleref{sep-lock}.
In particular, it admits no read-only sharing: two read-only views of the
same capability are separate but not disjoint.
Disjoint sets thus hold genuinely distinct capabilities, and \rruleref{pack}
uses the judgment to keep the witnesses of an existential pairwise apart.
Sequential composition $\seqcomp{\G}{C_1}{C_2}$ states that a computation
using $C_1$ may run before one using $C_2$:
either $C_1$ consumes nothing \rruleref{seq-access-only}
or it is separated from $C_2$ \rruleref{seq-sep}.
Satisfaction $\sat{\G}{[\Psi,\Phi]}$ discharges a modal context:
each mode entry of $\Phi$ is established by kinding
and each pair of distinct entries of $\Psi$ by separation.

\subsubsection{Subtyping}

\begin{figure*}[htbp]
\rulefigsetup

\judgheader{Subtyping}{\subs{\G}{E_1}{E_2}}
\begin{multicols}{2}

\infax[\rruledef{top}]
{\subs{\G}{S}{\top}}

\infax[\rruledef{refl}]
{\subs{\G}{E}{E}}

\infrule[\rruledef{trans}]
{\subs{\G}{E_1}{E_2}\andalso
 \subs{\G}{E_2}{E_3}}
{\subs{\G}{E_1}{E_3}}

\infrule[\rruledef{tvar}]
{X<:S\in\G}
{\subs{\G}{X}{S}}

\infrule[\rruledef{capt}]
{\subs{\G}{S_1}{S_2}\andalso
 \subs{\G}{C_1}{C_2}}
{\subs{\G}{m\,S_1\capt C_1}{m\,S_2\capt C_2}}

\infrule[\rruledef{fun}]
{\subs{\G}{T_2}{T_1}\andalso
 \subs{(\G, x: T_2)}{E_1}{E_2}}
{\subs{\G}{\forall(x: T_1) E_1}{\forall(x: T_2) E_2}}

\infrule[\rruledef{tfun}]
{\subs{\G}{S_2}{S_1}\andalso
 \subs{(\G, X<:S_2)}{E_1}{E_2}}
{\subs{\G}{\forall[X<:S_1] E_1}{\forall[X<:S_2] E_2}}

\infrule[\rruledef{cfun}]
{\subs{\G}{B_2}{B_1}\andalso
 \subs{(\G, \epsilon\,c<:B_2)}{E_1}{E_2}}
{\subs{\G}{\forall[c<:B_1] E_1}{\forall[c<:B_2] E_2}}

\infrule[\rruledef{boxed}]
{\subs{(\G, \LOCK[\Psi,\Phi])}{E_1}{E_2}}
{\subs{\G}{[\Psi,\Phi]E_1}{[\Psi,\Phi]E_2}}

\infrule[\rruledef{boxed-sat}]
{\sat{(\G, \LOCK[\Psi_2,\Phi_2])}{[\Psi_1,\Phi_1]}}
{\subs{\G}{[\Psi_1,\Phi_1]E}{[\Psi_2,\Phi_2]E}}

\infrule[\rruledef{exists}]
{\subs{(\G, \epsilon\,c_1<:\top, \ldots, \epsilon\,c_n<:\top)}{T_1}{T_2}}
{\subs{\G}{\EXCAP{c_1,\ldots,c_n}T_1}{\EXCAP{c_1,\ldots,c_n}T_2}}

\end{multicols}

\vspace{-1.5em}
\caption{Subtyping rules of System \corecapybara{}.}\label{fig:core-subtyping}

\end{figure*}
 
Figure~\ref{fig:core-subtyping} defines subtyping.
The judgment relates existential types, with capturing types included as the
existential types that bind no capture variable.
Rule \rruleref{capt} combines shape subtyping with subcapturing;
the remaining congruence rules relate shape types.
The mechanization instead inlines the subcapturing premise of
\rruleref{capt} into each congruence rule.
Only pure types are below $\top$ \rruleref{top}.
Rule \rruleref{boxed-sat} replaces the modal context of a modal type by
another one that suffices to satisfy it:
the original context $[\Psi_1,\Phi_1]$ must be satisfied under a lock
carrying the new context $[\Psi_2,\Phi_2]$.
Rule \rruleref{exists} relates two existentials of the same arity under their
$n$ freshly bound, access-only capture variables.
Consumer types carry no congruence rule; they are related only by
\rruleref{refl} and \rruleref{trans}.
Reference and reader types are invariant in their content type.

\subsubsection{Term Typing}

\begin{figure*}[htbp]
\rulefigsetup

\judgheader{Typing}{\typ{C}{\G}{t}{E}}
\begin{multicols}{2}

\infrule[\rruledef{var}]
{x: m\,S\capt C\in \G}
{\typ{\set{}}{\G}{x}{m\,S\capt \set{x}}}

\infrule[\rruledef{reader}]
{x: \REF[T]\capt C\in \G}
{\typ{\set{}}{\G}{\RO\,x}{\RO\,\REF[T]\capt\set{\RO\,x}}}

\infrule[\rruledef{abs}]
{\typ{C}{(\G, x: T)}{t}{E}}
{\typ{\set{}}{\G}{\lambda(x: T)t}{(\forall(x: T) E)\capt C}}

\infrule[\rruledef{tabs}]
{\typ{C}{(\G, X<:S)}{t}{E}}
{\typ{\set{}}{\G}{\lambda[X<:S]t}{(\forall[X<:S] E)\capt C}}

\infrule[\rruledef{cabs}]
{\accessonly{\G}{B}\andalso
 \typ{C}{(\G, \epsilon\,c<:B)}{t}{E}}
{\typ{\set{}}{\G}{\lambda[c<:B]t}{(\forall[c<:B] E)\capt C}}

\infrule[\rruledef{consumer}]
{\typ{C\cup\set{c,\CONSUME\,c}}{(\G, \CONSUME\,c<:\top, x: T)}{t}{E}}
{\typ{\set{}}{\G}{\lambda(\<c,x\>: \EXCAP{c}T)t}{(\forall(\<c,x\>:\EXCAP{c}T) E)\capt C}}

\infrule[\rruledef{lock}]
{\typ{C}{(\G, \LOCK[\Psi,\Phi])}{t}{E}}
{\typ{\set{}}{\G}{\LOCK[\Psi,\Phi]\,t}{([\Psi,\Phi]E)\capt C}}

\infrule[\rruledef{pack}]
{\textstyle D=\bigcup_{i=1}^{n} C_i\andalso
 \accessonly{\G}{D}\andalso
 \consumable{\G}{D}\\
 \disj{\G}{C_i}{C_j}\ \text{for}\ i\neq j\\
 \typ{\set{}}{\G}{x}{[c_1:=C_1,\ldots,c_n:=C_n]T}}
{\typ{D\cup\CONSUME\,D}{\G}{\<C_1,\ldots,C_n,x\>}{\EXCAP{c_1,\ldots,c_n}T}}

\infrule[\rruledef{app}]
{\accessible{\G}{\set{x}}\\
 \typ{\set{}}{\G}{x}{(\forall(z: T) E)\capt \set{x}}\andalso
 \typ{\set{}}{\G}{y}{T}}
{\typ{\set{x}}{\G}{x\,y}{[z:=y]E}}

\infrule[\rruledef{consumer-app}]
{\seqcomp{\G}{C_1}{\set{x}}\andalso
 \accessible{\G}{\set{x}}\\
 \typ{\set{}}{\G}{x}{(\forall(\<c,x\>: \EXCAP{c}T) E)\capt \set{x}}\andalso
 \typ{C_1}{\G}{t}{\EXCAP{c}T}}
{\typ{C_1\cup\set{x}}{\G}{x\,t}{E}}

\infrule[\rruledef{tapp}]
{\accessible{\G}{\set{x}}\\
 \typ{\set{}}{\G}{x}{(\forall[X<:S] E)\capt \set{x}}}
{\typ{\set{x}}{\G}{x[S]}{[X:=S]E}}

\infrule[\rruledef{capp}]
{\accessible{\G}{\set{x}}\andalso
 \accessonly{\G}{D}\\
 \typ{\set{}}{\G}{x}{(\forall[c<:D] E)\capt \set{x}}}
{\typ{\set{x}}{\G}{x[D]}{[c:=D]E}}

\infrule[\rruledef{unlock}]
{\typ{\set{}}{\G}{x}{([\Psi,\Phi]E)\capt \set{x}}\\
 \sat{\G}{[\Psi,\Phi]}}
{\typ{\set{x}}{\G}{\UNLOCK\,x}{E}}

\infrule[\rruledef{let}]
{\seqcomp{\G}{C_1}{C_2}\andalso
 \typ{C_1}{\G}{t}{T}\\
 \typ{C_2}{((\G\ominus C_1), x: T)}{u}{E}}
{\typ{C_1\cup C_2}{\G}{\LET x = t\IN u}{E}}

\end{multicols}

\infrule[\rruledef{unpack}]
{\typ{C_1}{\G}{t}{\EXCAP{c_1,\ldots,c_n}T}\andalso
 \seqcomp{\G}{C_1}{C_2}\andalso
 \Psi_w = \set{c_1,\ldots,c_n}\,,\ C_2\andalso
 \Phi_w=\emptyset\\
 \G' = (\G\ominus C_1), \CONSUME\,c_1<:\top, \ldots, \CONSUME\,c_n<:\top,
   \LOCK[\Psi_w,\Phi_w], x: T\\
 \typ{C_2\cup\set{c_1,\ldots,c_n,\CONSUME\,c_1,\ldots,\CONSUME\,c_n}}{\G'}{u}{E}}
{\typ{C_1\cup C_2}{\G}{\LET \<c_1,\ldots,c_n,x\> = t\IN u}{E}}

\begin{multicols}{2}

\infrule[\rruledef{alloc}]
{\typ{\set{}}{\G}{x}{T}}
{\typ{\set{}}{\G}{\ALLOC x}{\EXCAP{c}\REF[T]\capt \set{c}}}

\infrule[\rruledef{read}]
{\accessible{\G}{\set{x}}\andalso
 \typ{\set{}}{\G}{x}{\RO\,\REF[T]\capt C}}
{\typ{\set{x}}{\G}{\READ x}{T}}

\infrule[\rruledef{write}]
{\accessible{\G}{\set{x}}\\
 \typ{\set{}}{\G}{x}{\REF[T]\capt C}\andalso
 \typ{\set{}}{\G}{y}{T}}
{\typ{\set{x}}{\G}{x:=y}{\TUNIT}}

\infrule[\rruledef{dealloc}]
{\consumable{\G}{\set{x}}\\
 \typ{\set{}}{\G}{x}{\REF[T]\capt \set{x}}}
{\typ{\set{\CONSUME\,x}}{\G}{\DEALLOC x}{\TUNIT}}

\infrule[\rruledef{if}]
{\typ{C_1}{\G}{x}{\TBOOL}\\
 \typ{C_2}{\G}{t_1}{E}\andalso
 \typ{C_3}{\G}{t_2}{E}}
{\typ{C_1\cup C_2\cup C_3}{\G}{\IF x\THEN t_1\ELSE t_2}{E}}

\infrule[\rruledef{par}]
{\sep{\G}{C_1}{C_2}\\
 \typ{C_1}{\G}{t_1}{E_1}\andalso
 \typ{C_2}{\G}{t_2}{E_2}}
{\typ{C_1\cup C_2}{\G}{\PAR t_1, t_2}{\TUNIT}}

\infax[\rruledef{unit}]
{\typ{\set{}}{\G}{()}{\TUNIT}}

\infax[\rruledef{true}]
{\typ{\set{}}{\G}{\TTRUE}{\TBOOL}}

\infax[\rruledef{false}]
{\typ{\set{}}{\G}{\TFALSE}{\TBOOL}}

\infrule[\rruledef{sub}]
{\typ{C}{\G}{t}{E}\\
 \subs{\G}{C}{C'}\andalso
 \subs{\G}{E}{E'}}
{\typ{C'}{\G}{t}{E'}}

\end{multicols}

\vspace{-1.5em}
\caption{Typing rules of System \corecapybara{}.}\label{fig:core-typing}

\end{figure*}
 
Figure~\ref{fig:core-typing} presents the typing rules.
The judgment $\typ{C}{\G}{t}{E}$ carries a use set $C$ bounding the
capabilities that $t$ accesses or consumes.
A variable occurrence has an empty use set,
and its type refines the declared capture set to $\set{x}$
\rruleref{var};
uses are charged at elimination forms,
each of which requires the eliminated capability to be accessible.
Type and capture application charge the empty use set: a type or capture
lambda binds a \emph{value} (\Cref{fig:core-syntax}), so applying it
produces that value without running anything, and the value's work is
charged where it is later eliminated; the accessibility and access-only
premises remain. The restriction to value bodies costs no expressiveness,
since a computation is suspended by boxing it under a lock with an empty
separation context, whose satisfaction is vacuous. Unlock, by contrast,
charges its subject \rruleref{unlock}: opening a modal value releases the
suspended computation inside, and the subject's capture set, which carries
the modal's annotation, is what pays for it.
Rule \rruleref{let} sequences the two use sets through
$\seqcomp{\G}{C_1}{C_2}$ and types the continuation in $\G\ominus C_1$,
so capture variables consumed by the bound term are killed and no longer
accessible.
Rule \rruleref{unpack} eliminates an $n$-ary existential.
It additionally requires the consumed roots of $C_1$ to be
consumable, and binds $n$ fresh consumable capture variables
$c_1,\ldots,c_n$ that the continuation may access and consume, as recorded by
the extra use set $\set{c_1,\ldots,c_n,\CONSUME\,c_1,\ldots,\CONSUME\,c_n}$.
Between the witnesses and $x$ it pushes the manufactured \emph{ownership
lock} $\LOCK[\Psi_w,\Phi_w]$, recording that the witnesses are separate from
the continuation charge $C_2$; the sequencing premise pays for it in the
footprint model, where every witness location is either confined to what
$C_1$ consumes, which $\seqcomp{\G}{C_1}{C_2}$ keeps out of $C_2$, or
freshly allocated by the bound term (see \Cref{rem:tr:ownership}).
Dually, \rruleref{pack} consumes the packed capture sets: writing
$D=\bigcup_i C_i$, each $C_i$ must be access-only and consumable, the witnesses
must be pairwise disjoint, and the pack charges $D\cup\CONSUME\,D$.
The consumer lambda \rruleref{consumer} is a first-class function that unpacks
its argument: its body binds the consumable witness $c$ and the term variable
$x$, and it kills every consumable capability of $\G$ that its closure $C$ does
not capture, chosen as the root set $X$.
Its application \rruleref{consumer-app} sequences the argument's use before the
consumer's own, mirroring \rruleref{unpack}.
Allocation returns a reference under a fresh capture variable
\rruleref{alloc};
deallocation requires $\set{x}$ to be consumable and charges
$\set{\CONSUME\,x}$ \rruleref{dealloc}.
Writing goes through the reference itself \rruleref{write},
whereas reading goes through a reader value \rruleref{reader}, \rruleref{read}.
Parallel composition requires the use sets of the two branches to be
separated \rruleref{par}.

\subsection{Dynamic Semantics}

The dynamic semantics is store-based.
Evaluation runs a term against a \emph{heap}.
At runtime, term variables additionally range over heap locations $l$:
evaluation substitutes locations for bound variables.
A \emph{configuration} $\cfg{H}{t}$ pairs a heap $H$ with the term $t$ under evaluation.
A heap is a finite map from locations to cells,
\[
H(l) \;::=\; v \;\mid\; \refcell{l'} \;\mid\; \deadcell{l'},
\]
a stored value $v$, 
a live reference cell $\refcell{l'}$ holding the location $l'$, 
or a dead cell $\deadcell{l'}$ left by a deallocation.
We write $H(l)$ for lookup, $\dom H$ for the set of allocated locations,
and $H[l\mapsto c]$ for the heap 
that extends or updates $H$ at $l$.
Each stored value also carries its \emph{footprint}, the set of runtime
capabilities it holds, which is the runtime image of its capture set; the
parallel rules consult it and we leave it implicit.

Both relations emit a \emph{trace} $\tau$, a finite sequence of heap events
recording every access the run performs.
An event is an access $\mu\,l$, a read when $\mu=\RO$ and a write when
$\mu=\epsilon$, an allocation $\ALLOC l$, or a deallocation $\DEALLOC l$.
Traces are the objects over which the standardization and confluence
theorems are stated (\Cref{app:meta:std}).
We write $\tau_1\,\tau_2$ for concatenation and omit the label on a step that emits
the empty trace.

\subsubsection{Small-step Reduction}

\begin{figure*}[htbp]
\rulefigsetup

\judgheader{Small-step reduction}{\cfg{H}{t} \stepsto{\tau} \cfg{H'}{t'}}
\begin{multicols}{2}

\infrule[\rruledef{s-app}]
{H(l) = \lambda(x: T)t}
{\cfg{H}{l\,l'} \stepsto{} \cfg{H}{[x:=l']t}}

\infrule[\rruledef{s-tapp}]
{H(l) = \lambda[X<:S]t}
{\cfg{H}{l[S']} \stepsto{} \cfg{H}{[X:=S']t}}

\infrule[\rruledef{s-capp}]
{H(l) = \lambda[c<:B]t}
{\cfg{H}{l[C]} \stepsto{} \cfg{H}{[c:=C]t}}

\infrule[\rruledef{s-consumer}]
{H(l) = \lambda(\<c,x\>: \EXCAP{c}T)t}
{\cfg{H}{l\,t'} \stepsto{} \cfg{H}{\LET \<c,x\> = t' \IN t}}

\infrule[\rruledef{s-unwrap}]
{H(l) = \LOCK[\Psi,\Phi]t}
{\cfg{H}{\UNLOCK\,l} \stepsto{} \cfg{H}{t}}

\infrule[\rruledef{s-if-true}]
{H(l) = \TTRUE}
{\cfg{H}{\IF l\THEN t_1\ELSE t_2} \stepsto{} \cfg{H}{t_1}}

\infrule[\rruledef{s-if-false}]
{H(l) = \TFALSE}
{\cfg{H}{\IF l\THEN t_1\ELSE t_2} \stepsto{} \cfg{H}{t_2}}

\infrule[\rruledef{s-read}]
{H(l) = \RO\,l_c\andalso H(l_c) = \refcell{n}}
{\cfg{H}{\READ l} \stepsto{\RO\,l_c} \cfg{H}{n}}

\infrule[\rruledef{s-write}]
{H(l_x) = \refcell{n_0}\andalso l_y\in\dom H}
{\cfg{H}{l_x := l_y} \stepsto{\epsilon\,l_x} \cfg{H[l_x\mapsto\refcell{l_y}]}{()}}

\infrule[\rruledef{s-alloc}]
{l_x\in\dom H\andalso l\notin\dom H}
{\cfg{H}{\ALLOC l_x} \stepsto{\ALLOC l} \cfg{H[l\mapsto\refcell{l_x}]}{\<\set{l},l\>}}

\infrule[\rruledef{s-drop}]
{H(l_x) = \refcell{n}}
{\cfg{H}{\DEALLOC l_x} \stepsto{\DEALLOC l_x} \cfg{H[l_x\mapsto\deadcell{n}]}{()}}

\infrule[\rruledef{s-lift}]
{v\ \text{a value}\andalso l\notin\dom H}
{\cfg{H}{\LET x = v \IN t} \stepsto{} \cfg{H[l\mapsto v]}{[x:=l]t}}

\infax[\rruledef{s-rename}]
{\cfg{H}{\LET x = l' \IN t} \stepsto{} \cfg{H}{[x:=l']t}}

\infrule[\rruledef{s-let}]
{\cfg{H}{t_1} \stepsto{\tau} \cfg{H'}{t_1'}}
{\cfg{H}{\LET x = t_1 \IN t_2} \stepsto{\tau} \cfg{H'}{\LET x = t_1' \IN t_2}}

\end{multicols}

\infax[\rruledef{s-unpack}]
{\cfg{H}{\LET \<c_1,\ldots,c_n,x\> = \<C_1,\ldots,C_n,l\> \IN t} \stepsto{} \cfg{H}{[c_1:=C_1,\ldots,c_n:=C_n,x:=l]t}}

\infrule[\rruledef{s-unpack-ctx}]
{\cfg{H}{t_1} \stepsto{\tau} \cfg{H'}{t_1'}}
{\cfg{H}{\LET \<c_1,\ldots,c_n,x\> = t_1 \IN t_2} \stepsto{\tau} \cfg{H'}{\LET \<c_1,\ldots,c_n,x\> = t_1' \IN t_2}}

\infrule[\rruledef{s-par-l}]
{\cfg{H}{t_1} \stepsto{\tau} \cfg{H'}{t_1'}\andalso \tauth{\tau}{C_1}\andalso \nintf{C_1}{C_2}}
{\cfg{H}{\PAR[C_1, C_2]\, t_1, t_2} \stepsto{\tau} \cfg{H'}{\PAR[C_1', C_2]\, t_1', t_2}}

\infrule[\rruledef{s-par-r}]
{\cfg{H}{t_2} \stepsto{\tau} \cfg{H'}{t_2'}\andalso \tauth{\tau}{C_2}\andalso \nintf{C_1}{C_2}}
{\cfg{H}{\PAR[C_1, C_2]\, t_1, t_2} \stepsto{\tau} \cfg{H'}{\PAR[C_1, C_2']\, t_1, t_2'}}

\infrule[\rruledef{s-par-join}]
{t_1\ \text{an answer}\andalso t_2\ \text{an answer}}
{\cfg{H}{\PAR[C_1, C_2]\, t_1, t_2} \stepsto{} \cfg{H}{()}}

\vspace{-1.5em}
\caption{Small-step reduction of System \corecapybara{}.}\label{fig:core-smallstep}

\end{figure*}
 
Figure~\ref{fig:core-smallstep} defines the small-step relation
$\cfg{H}{t} \stepsto{\tau} \cfg{H'}{t'}$.
An elimination form looks up its subject in the heap and reduces against the
stored value.
\rruleref{s-app}, \rruleref{s-tapp}, and \rruleref{s-capp} apply a stored
abstraction by substitution; type and capture arguments are computationally
irrelevant, so these steps leave the heap and the trace untouched.
\rruleref{s-consumer} applies a consumer by unpacking its argument, reducing to
a \rruleref{unpack} redex, and \rruleref{s-unwrap} opens a boxed value.
The memory primitives are the only steps that touch the heap or the trace:
\rruleref{s-read} follows a reader to its cell and returns the stored location,
emitting a read; \rruleref{s-write} overwrites a live cell, emitting a write;
\rruleref{s-alloc} installs a fresh live cell and returns a pack of its
location; and \rruleref{s-drop} marks a cell dead.
A let evaluates its bound term under \rruleref{s-let}, then either lifts the
resulting value to a fresh location \rruleref{s-lift} or, when it is already a
location, substitutes it \rruleref{s-rename}; \rruleref{s-unpack} substitutes
the $n$ witnesses and the location of a pack.
\rruleref{s-par-l} and \rruleref{s-par-r} interleave the two branches of a
parallel composition and \rruleref{s-par-join} returns unit once both are
answers.
A parallel node carries the footprints $C_1$ and $C_2$ of its two branches,
the capture-set annotations of the \rruleref{par} typing rule resolved to the
runtime capabilities they hold.
A footprint is a finite set of elements $\mu\,l$ tagging a location $l$ with an
access mode $\mu$, and we order access modes by extending $\RO\preceq\epsilon$
with $\CONSUME\preceq\CONSUME$, leaving $\CONSUME$ incomparable to the
mutabilities.
A left step is allowed on two conditions.
First, the emitted trace is \emph{authorized} by the stepping branch's
footprint.
Say a location $l$ is \emph{fresh at} an event of $\tau$ when an earlier
$\ALLOC l$ event precedes it in $\tau$.
Then $\tauth{\tau}{C_1}$ holds when, for every event of $\tau$ whose location is not
fresh at it, an access $\mu\,l$ has some $\mu'\,l\in C_1$ with $\mu\preceq\mu'$,
and a deallocation of $l$ has $\CONSUME\,l\in C_1$.
Second, the two footprints are \emph{non-interfering}:
$\nintf{C_1}{C_2}$ holds when every $\mu_1\,l_1\in C_1$ and $\mu_2\,l_2\in C_2$
satisfy $l_1\neq l_2$ or $\mu_1=\mu_2=\RO$.
The stepping branch's footprint then grows to
$C_1'=C_1\cup\set{\epsilon\,l,\CONSUME\,l\mid \ALLOC l\ \text{occurs in}\ \tau}$,
adding the cells $t$ allocates so a later step may use them.
\rruleref{s-par-r} is symmetric.
These are exactly the obligations that the \rruleref{par} typing rule
establishes on $C_1$ and $C_2$, so a well-typed program never blocks on them.

\subsubsection{Big-step Evaluation}

\begin{figure*}[htbp]
\rulefigsetup

\judgheader{Big-step evaluation}{\cfg{H}{t} \evalsto{\tau} \cfg{H'}{a}}
\begin{multicols}{2}

\infax[\rruledef{b-val}]
{\cfg{H}{v} \evalsto{} \cfg{H}{v}}

\infax[\rruledef{b-var}]
{\cfg{H}{l} \evalsto{} \cfg{H}{l}}

\infrule[\rruledef{b-app}]
{H(l) = \lambda(x: T)t\andalso
 \cfg{H}{[x:=l']t} \evalsto{\tau} \cfg{H'}{a}}
{\cfg{H}{l\,l'} \evalsto{\tau} \cfg{H'}{a}}

\infrule[\rruledef{b-tapp}]
{H(l) = \lambda[X<:S]t\andalso
 \cfg{H}{[X:=S']t} \evalsto{\tau} \cfg{H'}{a}}
{\cfg{H}{l[S']} \evalsto{\tau} \cfg{H'}{a}}

\infrule[\rruledef{b-capp}]
{H(l) = \lambda[c<:B]t\andalso
 \cfg{H}{[c:=C]t} \evalsto{\tau} \cfg{H'}{a}}
{\cfg{H}{l[C]} \evalsto{\tau} \cfg{H'}{a}}

\infrule[\rruledef{b-consumer}]
{H(l) = \lambda(\<c,x\>: \EXCAP{c}T)t\\
 \cfg{H}{\LET \<c,x\> = t' \IN t} \evalsto{\tau} \cfg{H'}{a}}
{\cfg{H}{l\,t'} \evalsto{\tau} \cfg{H'}{a}}

\infrule[\rruledef{b-unwrap}]
{H(l) = \LOCK[\Psi,\Phi]t\andalso
 \cfg{H}{t} \evalsto{\tau} \cfg{H'}{a}}
{\cfg{H}{\UNLOCK\,l} \evalsto{\tau} \cfg{H'}{a}}

\infrule[\rruledef{b-if-true}]
{H(l) = \TTRUE\andalso
 \cfg{H}{t_1} \evalsto{\tau} \cfg{H'}{a}}
{\cfg{H}{\IF l\THEN t_1\ELSE t_2} \evalsto{\tau} \cfg{H'}{a}}

\infrule[\rruledef{b-if-false}]
{H(l) = \TFALSE\andalso
 \cfg{H}{t_2} \evalsto{\tau} \cfg{H'}{a}}
{\cfg{H}{\IF l\THEN t_1\ELSE t_2} \evalsto{\tau} \cfg{H'}{a}}

\infrule[\rruledef{b-read}]
{H(l) = \RO\,l_c\andalso H(l_c) = \refcell{n}}
{\cfg{H}{\READ l} \evalsto{\RO\,l_c} \cfg{H}{n}}

\infrule[\rruledef{b-write}]
{H(l_x) = \refcell{n_0}\andalso l_y\in\dom H}
{\cfg{H}{l_x := l_y} \evalsto{\epsilon\,l_x} \cfg{H[l_x\mapsto\refcell{l_y}]}{()}}

\infrule[\rruledef{b-alloc}]
{l_x\in\dom H\andalso l\notin\dom H}
{\cfg{H}{\ALLOC l_x} \evalsto{\ALLOC l} \cfg{H[l\mapsto\refcell{l_x}]}{\<\set{l},l\>}}

\infrule[\rruledef{b-drop}]
{H(l_x) = \refcell{n}}
{\cfg{H}{\DEALLOC l_x} \evalsto{\DEALLOC l_x} \cfg{H[l_x\mapsto\deadcell{n}]}{()}}

\infrule[\rruledef{b-let-val}]
{\cfg{H}{t_1} \evalsto{\tau_1} \cfg{H_1}{v}\andalso l\notin\dom H_1\\
 \cfg{H_1[l\mapsto v]}{[x:=l]t_2} \evalsto{\tau_2} \cfg{H_2}{a}}
{\cfg{H}{\LET x = t_1 \IN t_2} \evalsto{\tau_1\,\tau_2} \cfg{H_2}{a}}

\infrule[\rruledef{b-let-var}]
{\cfg{H}{t_1} \evalsto{\tau_1} \cfg{H_1}{l}\andalso
 \cfg{H_1}{[x:=l]t_2} \evalsto{\tau_2} \cfg{H_2}{a}}
{\cfg{H}{\LET x = t_1 \IN t_2} \evalsto{\tau_1\,\tau_2} \cfg{H_2}{a}}

\infrule[\rruledef{b-unpack}]
{\cfg{H}{t_1} \evalsto{\tau_1} \cfg{H_1}{\<C_1,\ldots,C_n,l\>}\\
 \cfg{H_1}{[c_1:=C_1,\ldots,c_n:=C_n,x:=l]t_2} \evalsto{\tau_2} \cfg{H_2}{a}}
{\cfg{H}{\LET \<c_1,\ldots,c_n,x\> = t_1 \IN t_2} \evalsto{\tau_1\,\tau_2} \cfg{H_2}{a}}

\infrule[\rruledef{b-par}]
{\cfg{H}{t_1} \evalsto{\tau_1} \cfg{H_1}{a_1}\andalso
 \cfg{H_1}{t_2} \evalsto{\tau_2} \cfg{H_2}{a_2}}
{\cfg{H}{\PAR t_1, t_2} \evalsto{\tau_1\,\tau_2} \cfg{H_2}{()}}

\end{multicols}

\vspace{-1.5em}
\caption{Big-step evaluation of System \corecapybara{}.}\label{fig:core-bigstep}

\end{figure*}
 
Figure~\ref{fig:core-bigstep} defines the big-step relation
$\cfg{H}{t} \evalsto{\tau} \cfg{H'}{a}$, evaluating $t$ to an answer $a$, a value
or a location, while threading the heap and concatenating the traces of its
subevaluations.
It is the sequential counterpart of the small-step relation: values and
locations evaluate to themselves \rruleref{b-val}, \rruleref{b-var}, the
elimination forms evaluate the redex they expose, and \rruleref{b-par} runs the
two branches left to right rather than interleaving them.
\rruleref{b-let-val} and \rruleref{b-let-var} sequence a let by evaluating the
bound term first and then the body, lifting an intermediate value to the heap
or substituting an intermediate location, and \rruleref{b-unpack} does the same
for an existential.
The two relations agree up to the scheduling of parallel composition;
the metatheory relates them by a standardization argument
(\Cref{thm:meta:standardization}).
\section{Metatheory of System \corecapybara{}}\label{app:metatheory}

This appendix presents the semantic model behind \Cref{sec:metatheory}
and states the theorems in the form in which they are mechanized.
The Lean development is the ground truth.
The definitions below transcribe it in the notation of
\Cref{app:corecapybara}, omitting de Bruijn indices and well-scopedness
side conditions.
The model is a step-indexed Kripke logical relation over a higher-order
store~\cite{DBLP:conf/esop/Ahmed06,DBLP:journals/jacm/TimanyKDB24}:
types denote world-indexed predicates on heaps and terms,
capture sets denote footprints,
and the typing judgment asserts that the term evaluates to an answer
satisfying its type's predicate, with a trace authorized by its use
set.

\subsection{Footprints}\label{app:meta:footprints}

A \emph{footprint} is a capture set whose elements mention only heap
locations (\Cref{sec:core}).
Each stored value in the heap carries its footprint,
the runtime image of its capture-set annotation,
and a location inherits the footprint of what it names:
\[
\fploc{H}{l} =
\begin{cases}
\set{\epsilon\,l} & \text{if}\ H(l)\ \text{is a reference cell,}\\
R_v & \text{if}\ H(l) = v\ \text{with footprint}\ R_v.
\end{cases}
\]
A semantic environment $\rho$ maps term variables to locations and
capture variables to footprints;
applying it to a capture set $C$ yields a ground set,
whose \emph{footprint} $\fdenot{C}_\rho^H$ expands each element
to the footprint of its location:
\[
\fdenot{\set{}}_\rho^H = \set{},
\qquad
\fdenot{C_1\cup C_2}_\rho^H = \fdenot{C_1}_\rho^H\cup\fdenot{C_2}_\rho^H,
\qquad
\fdenot{\set{\mu\,l}}_\rho^H = \mu\,\fploc{H}{l},
\]
with mode application as in \Cref{app:corecapybara}.
We drop the subscripts when the environment and heap are clear.
A footprint $C$ \emph{covers} an access $\mu\,l$ when $\mu'\,l\in C$
for some $\mu\preceq\mu'$.
This is the reading used by trace
authorization $\tauth{\tau}{C}$ (\Cref{sec:core}).
Non-interference $\nintf{C_1}{C_2}$ is as in \Cref{sec:core},
and two footprints are \emph{disjoint} when they share no location at
all.

\subsection{Worlds}\label{app:meta:worlds}

The denotations of the next subsection are Kripke: a type is
interpreted relative to a \emph{world} $(\Sigma, H)$, the heap $H$ a
program has reached together with a \emph{store typing} $\Sigma$ that
assigns semantic contents to its reference cells.
As a program allocates and writes, the world grows, and denotations
are stable under this growth, so a value keeps its type in every
future world.

The store typing is what breaks the circularity of the higher-order
store: a cell typed $\REF[T]$ should hold contents of type $T$, but
$T$ may itself mention reference types, and through type variables a
cell may even store values of an abstract type.
$\Sigma$ therefore assigns each cell a relation rather than a
syntactic type, and the relations are stratified by a step index $k$:
\[
\mathit{World}_k = \mathit{Loc}\rightharpoonup \mathit{Rel}_k,
\qquad
\mathit{Rel}_k = \textstyle\prod_{j<k}
  \left(\mathit{World}_j\to\mathit{Heap}\to\mathit{Term}\to\mathrm{Prop}\right).
\]
A stored relation at index $k$ is a family, over every lower index
$j<k$, of predicates on a world at index $j$, a heap, and a term.
The space of relations recurses on the index, so the store is
inherently step-indexed.
This shape is forced: a stored relation must be stable under world
extension, and must agree with the content type in both directions so that writes can re-establish it.
We write $\lfloor\Sigma\rfloor_j$ for the truncation of $\Sigma$ to a
lower index $j$, which restricts each stored family to indices below
$j$.

World extension
$\wext{\Sigma'}{H'}{\Sigma}{H}$
holds when $H'$ evolves from $H$,
meaning every allocated location keeps its cell,
where a stored value is unchanged
and a reference cell may change its content and may die but not
revive,
and when $\Sigma'$ agrees with $\Sigma$ on $\dom\Sigma$.
A world is \emph{well-typed}, written $\memtyped{k}{\Sigma}{H}$, when
\begin{enumerate}
\item every $l\in\dom\Sigma$ is a reference cell of $H$;
\item every stored relation of $\Sigma$ is stable under world
  extension, at every index below $k$; and
\item for every \emph{live} cell $l\in\dom\Sigma$ with
  $H(l)=\refcell{l'}$, the content satisfies its stored relation at
  every lower index: $\Sigma(l)_j\,(\lfloor\Sigma\rfloor_j, H, l')$
  for all $j<k$.
\end{enumerate}
Condition (3) exposes cell contents only at indices strictly below
$k$, so a read consumes one unit of the index.
The model accordingly treats $k$ as a \emph{read budget}: it bounds
the number of read events of the runs about which the model makes
claims.

\subsection{Interpretation of Types}\label{app:meta:types}

Types are interpreted relative to a \emph{semantic environment}
$\rho$, which assigns to each variable of the context a semantic
object of the matching kind: a heap location to a term variable, a
footprint to a capture variable (\Cref{app:meta:footprints}), and a
denotation, a world-indexed predicate on answers, to a type variable.
It is defined in \Cref{app:meta:semtyping}.
The value denotation $\vdenot{T}_\rho$ is then a world-indexed
predicate on answers, defined by recursion on the type, mutually with
the expression denotation of the next subsection.
We write
$a\in\vdenot{T}_\rho(k,\Sigma,H)$ when it holds of $a$ at the world.
Because evaluation lifts values to the heap, an answer names a value
either directly or through a location.
The partial function $\resolve$ recovers it:
\[
\resolve(H,v) = v,
\qquad
\resolve(H,l) = v\ \ \text{if}\ H(l)=v,
\]
and $\resolve(H,l)$ is undefined when $l$ is unallocated or holds a
reference cell.
For a stored relation $R$ at index $k$, write
$R\approx_k\vdenot{T}_\rho$ when $R$ agrees with the denotation of $T$
at every lower world:
$R_j\,(\Sigma',H',a') \iff a'\in\vdenot{T}_\rho(j,\Sigma',H')$
for all $j<k$, $\Sigma'$, $H'$, $a'$.
The first-order clauses are:
\[
\begin{array}{lcl}
\vdenot{\TUNIT}_\rho(k,\Sigma,H) &=& \set{a\mid \resolve(H,a) = ()}\\[2pt]
\vdenot{\TBOOL}_\rho(k,\Sigma,H) &=& \set{a\mid \resolve(H,a)\in\set{\TTRUE,\TFALSE}}\\[2pt]
\vdenot{\top}_\rho(k,\Sigma,H) &=& \set{a\mid \resolve(H,a)\ \text{is a value with empty footprint}}\\[2pt]
\vdenot{X}_\rho(k,\Sigma,H) &=& \rho(X)\,(k,\Sigma,H)\\[2pt]
\vdenot{\REF[T]\capt C}_\rho(k,\Sigma,H) &=&
  \set{l\mid H(l)\ \text{a cell},\ \fdenot{C}\ \text{covers}\ \epsilon\,l,\
       \Sigma(l)\approx_k\vdenot{T}_\rho}\\[2pt]
\vdenot{\RO\,\REF[T]\capt C}_\rho(k,\Sigma,H) &=&
  \set{a\mid \resolve(H,a)=\RO\,l,\ \fdenot{C}\ \text{covers}\ \RO\,l,\
       \Sigma(l)\approx_k\vdenot{T}_\rho}
\end{array}
\]
Only pure values are below $\top$.
In the two cell clauses, the forward direction of the agreement
$\approx_k$ types the result of a read, and the backward direction
lets a write re-establish condition (3) of $\memtyped{k}{\Sigma}{H}$.

A function type denotes the abstractions whose bodies meet their
specification in every future world.
The four abstraction clauses share this quantification, so we
abbreviate it: for $j\leq k$, let $\fworlds{j}{R}$ be the set of
well-typed future worlds at index $j$ in which $R$ is live, that is,
the pairs $(\Sigma',H')$ with
$\wext{\Sigma'}{H'}{\lfloor\Sigma\rfloor_j}{H}$,
$\memtyped{j}{\Sigma'}{H'}$, and every cell of $R$ live in $H'$.
In each clause, the \emph{closure footprint} $R_0$ is the footprint of
the stored abstraction's capture annotation, and the footprints of the
quantified ground capture sets are taken in $H'$.
The function clause is:
\[
\begin{array}{l}
\vdenot{(\forall(x:T_1)E)\capt C}_\rho(k,\Sigma,H) =\\
\quad\{\,a\mid \resolve(H,a)=\lambda(x:T_1')\,t\ \text{with}\
  R_0\subseteq\fdenot{C},\ \text{and}\\
\qquad [x:=l]\,t\in\edenot{E}^{R_0}_{\rho[x\mapsto l]}(j,\Sigma',H')\\
\qquad\quad \text{for all}\ j\leq k,\
  (\Sigma',H')\in\fworlds{j}{R_0},\
  \text{and}\ l\in\vdenot{T_1}_\rho(j,\Sigma',H')\,\}
\end{array}
\]
The body's budget is $R_0$, not the annotation $C$: a function may
touch exactly what it captures.
A type abstraction takes a denotation as its argument:
\[
\begin{array}{l}
\vdenot{(\forall[X<:S]E)\capt C}_\rho(k,\Sigma,H) =\\
\quad\{\,a\mid \resolve(H,a)=\lambda[X<:S']\,t\ \text{with}\
  R_0\subseteq\fdenot{C},\ \text{and}\\
\qquad [X:=\top]\,t\in\edenot{E}^{R_0}_{\rho[X\mapsto\delta]}(j,\Sigma',H')\\
\qquad\quad \text{for all}\ j\leq k,\
  (\Sigma',H')\in\fworlds{j}{R_0},\
  \text{and all proper}\ \delta\
  \text{bounded by}\ \vdenot{S}_\rho\,\}
\end{array}
\]
A denotation $\delta$ is \emph{proper} when it satisfies the closure
properties the model demands of every type's interpretation, chiefly
stability under heap growth, world extension, and index decrease.
As an argument here it must moreover contain only pure answers.
Bounded means that $\delta$ implies $\vdenot{S}_\rho$ at every world
extending $(\Sigma',H')$, at every index at most $j$.
The type substituted into the body is irrelevant to evaluation.
The mechanization fixes $\top$.
A capture abstraction takes a ground capture set:
\[
\begin{array}{l}
\vdenot{(\forall[c<:B]E)\capt C}_\rho(k,\Sigma,H) =\\
\quad\{\,a\mid \resolve(H,a)=\lambda[c<:B']\,t\ \text{with}\
  R_0\subseteq\fdenot{C},\ \text{and}\\
\qquad [c:=W]\,t\in\edenot{E}^{R_0}_{\rho[c\mapsto W]}(j,\Sigma',H')\\
\qquad\quad \text{for all}\ j\leq k,\
  (\Sigma',H')\in\fworlds{j}{R_0},\\
\qquad\quad \text{and all ground}\ W\ \text{with}\ \fdenot{W}\
  \CONSUME\text{-free and}\ \fdenot{W}\subseteq\fdenot{B}\,\}
\end{array}
\]
The bound imposes no constraint when $B=\top$.
A consumer takes a packed argument, a ground witness $W$ together with
a location:
\[
\begin{array}{l}
\vdenot{(\forall(\<c,x\>:\EXCAP{c}T_1)E)\capt C}_\rho(k,\Sigma,H) =\\
\quad\{\,a\mid \resolve(H,a)=\lambda(\<c,x\>:\EXCAP{c}T_1')\,t\
  \text{with}\ R_0\subseteq\fdenot{C},\ \text{and}\\
\qquad [c:=W,\,x:=l]\,t\in
  \edenot{E}^{\,R_0\cup\fdenot{W}\cup\CONSUME\,\fdenot{W}}_{\rho}(j,\Sigma',H')\\
\qquad\quad \text{for all}\ j\leq k,\
  (\Sigma',H')\in\fworlds{j}{R_0},\
  \text{all ground}\ W\ \text{with}\ \fdenot{W}\
  \CONSUME\text{-free,}\\
\qquad\quad \text{live in}\ H'\text{, and disjoint from}\ R_0,\
  \text{and all}\ l\in\vdenot{T_1}_{\rho[c\mapsto W]}(j,\Sigma',H')\,\}
\end{array}
\]
The witness binding $\rho[c\mapsto W]$ carries consumable authority,
and the body's budget widens to
$R_0\cup\fdenot{W}\cup\CONSUME\,\fdenot{W}$: the consumer receives
full rights to the witness, including its consumption.
The result type $E$ mentions neither $c$ nor $x$, so it is interpreted
under $\rho$ itself.
A modal type $([\Psi,\Phi]E)\capt C$ denotes the boxed values whose
recorded separation context is discharged by the semantic satisfaction
of $[\Psi,\Phi]$ and whose bodies inhabit $\edenot{E}$ in every future
world in which $\rho$ satisfies $[\Psi,\Phi]$, meaning each kind entry
of $\Phi$ holds of the corresponding footprint and each separation
pair of $\Psi$ denotes non-interfering footprints.
Finally, an existential type denotes the packs whose witnesses are
separate:
\[
\begin{array}{l}
\vdenot{\EXCAP{c_1,\ldots,c_n}T}_\rho(k,\Sigma,H) =\\
\quad\{\,a\mid \resolve(H,a)=\<C_1,\ldots,C_n,x\>,\\
\qquad \fdenot{C_1},\ldots,\fdenot{C_n}\ \text{are}\
  \CONSUME\text{-free and pairwise disjoint},\\
\qquad x\in\vdenot{T}_{\rho[c_1\mapsto C_1,\ldots,c_n\mapsto C_n]}(k,\Sigma,H)\,\}
\end{array}
\]
where the witnesses $c_i$ are bound at consumable authority.

\subsection{The Expression Relation}\label{app:meta:exprel}

The expression denotation lifts a predicate on answers to a predicate
on expressions: an expression inhabits it when it evaluates, by the
big-step relation of \Cref{fig:core-bigstep}, to an answer in the
value denotation while accessing only the locations in its budget.
The read budget guards every clause.
Write $\readsof{\tau}$ for the number of read events of a trace
$\tau$.
Call a run of the small-step relation \emph{sequential} when it
follows the left-first schedule: \rruleref{s-par-r} fires only when
the left branch is already an answer, so the right branch of a
parallel node steps only after the left has finished.
This is the order in which \rruleref{b-par} evaluates a parallel
composition, so sequential runs are the small-step reading of the
big-step relation.
We write $\cfg{H}{t}\redsto{\tau}\cfg{H'}{t'}$ for the
reflexive-transitive closure of the small-step relation,
concatenating traces, and $\seqredsto{\tau}$ for its restriction to
sequential runs.
Then $t\in\edenot{E}^{R}_\rho(k,\Sigma,H)$ iff
$\memtyped{k}{\Sigma}{H}$ implies all of:
\begin{enumerate}
\item \emph{Safety:} every configuration reachable from $\cfg{H}{t}$
  by a sequential run performing fewer than $k$ reads is an answer or
  can step.
\item \emph{Postcondition:} for every evaluation
  $\cfg{H}{t}\evalsto{\tau}\cfg{H'}{a}$ with $\readsof{\tau}<k$,
  the trace is authorized by the budget, $\tauth{\tau}{R}$, and,
  where $k' = k-\readsof{\tau}$, there is a store typing $\Sigma'$ at
  index $k'$ with
  \[
  \wext{\Sigma'}{H'}{\lfloor\Sigma\rfloor_{k'}}{H},
  \qquad
  \memtyped{k'}{\Sigma'}{H'},
  \qquad
  a\in\vdenot{E}_\rho(k',\Sigma',H').
  \]
\item \emph{Prefix safety:} for every sequential run
  $\cfg{H}{t}\seqredsto{\tau}\cfg{H''}{t''}$ with
  $\readsof{\tau}<k$, complete or not, the trace is authorized by the
  budget, $\tauth{\tau}{R}$.
\end{enumerate}
Runs that exhaust the budget owe nothing.
Adequacy below recovers
unconditional claims by instantiating $k$ above the read count of any
given run.
When $E$ is an existential, the postcondition additionally requires
the witnesses of the answer to be live in $H'$ and \emph{confined}:
every capability a witness reaches is either consumable under the
budget $R$ or was freshly allocated by the run.
This witness-confinement bound is what pays for the ownership lock of
\rruleref{unpack} (\Cref{rem:tr:ownership}).

The safety clause is a budget-indexed progress predicate.
At a parallel node it is a rely-guarantee invariant:
it carries the
two branch footprints, requires every branch run to be authorized by its footprint and the two footprints to be non-interfering,
and
requires each branch to remain safe under every heap that runs of the
other branch can produce.

The postcondition constrains complete runs only, so prefix safety is
what bounds the trace of a run still underway.
In the mechanization, the sequential relation carries none of the
guards of \rruleref{s-par-l} and \rruleref{s-par-r}, and prefix
safety moreover realizes each within-budget run in \emph{guarded}
form, with the authorization and non-interference guards discharged
at every parallel step.
This realization is what embeds a run of the model into the full
small-step relation, where the guards are checked at every step
(\Cref{app:meta:adequacy}).

\subsection{Semantic Typing}\label{app:meta:semtyping}

An environment $\rho$ \emph{realizes} a context $\G$ at a world,
written $\rho\in\gdenot{\G}(k,\Sigma,H)$, when it interprets each
binding as its denotation demands:
\begin{itemize}
\item for $x: T$, the location $\rho(x)$ is in
  $\vdenot{T}_\rho(k,\Sigma,H)$;
\item for $X<:S$, the denotation $\rho(X)$ is a proper,
  extension-stable predicate bounded by $\vdenot{S}_\rho$;
\item for $\alpha\,c<:B$, the footprint of $\rho(c)$ is
  $\CONSUME$-free and bounded by $\fdenot{B}$, and $\rho$ records the
  authority $\alpha$ for $c$;
\item for a lock $\LOCK[\Psi,\Phi]$, $\rho$ satisfies $[\Psi,\Phi]$
  semantically, as in the modal clause above.
\end{itemize}
Consumption rights are carried by the binding authority, not by the
footprint itself.
The budget of a term that consumes $c$ acquires the
$\CONSUME$ capabilities through the rules that charge them.
An environment is \emph{separation-well-formed} when the footprints of
distinct consumable capture variables are disjoint.

Semantic typing quantifies over all realizing worlds:
$\semtyp{C}{\G}{t}{E}$ holds iff for all $\rho$, $k$, $\Sigma$, $H$
with $\rho\in\gdenot{\G}(k,\Sigma,H)$, $\rho$ separation-well-formed,
and every cell of $\fdenot{C}_\rho^H$ live in $H$,
\[
t\rho \in \edenot{E}^{\fdenot{C}_\rho^H}_{\rho}(k,\Sigma,H),
\]
where $t\rho$ substitutes the environment into the term.

\subsection{The Fundamental Theorem}\label{app:meta:fundamental}

\begin{theorem}[Fundamental theorem]\label{thm:meta:fundamental}
If $\typ{C}{\G}{t}{E}$, then $\semtyp{C}{\G}{t}{E}$.
\end{theorem}

The proof is by induction on the typing derivation, with one
compatibility lemma per rule of \Cref{fig:core-typing}.
The auxiliary judgments of \Cref{app:corecapybara} have fundamental
lemmas of their own, which the compatibility lemmas invoke.
The two that carry the concurrency results are:

\begin{lemma}[Subcapturing]\label{lem:meta:subcapt}
If $\subs{\G}{C_1}{C_2}$ and $\rho$ realizes $\G$, then every
capability of $\fdenot{C_1}$ is covered by $\fdenot{C_2}$.
\end{lemma}

\begin{lemma}[Separation]\label{lem:meta:sep}
If $\sep{\G}{C_1}{C_2}$ and $\rho$ realizes $\G$ and is
separation-well-formed, then $\nintf{\fdenot{C_1}}{\fdenot{C_2}}$.
\end{lemma}

\Cref{lem:meta:sep} is \Cref{thm:separation}: syntactic separation
denotes non-interference of footprints.
It is what lets the compatibility lemma for \rruleref{par} establish
the rely--guarantee obligations of the safety predicate and the
side conditions of \rruleref{s-par-l} and \rruleref{s-par-r}.
Kinding has an analogous lemma: a capture set of kind $\RO$ denotes a
footprint that covers no write.

\subsection{Adequacy}\label{app:meta:adequacy}

Adequacy extracts unconditional statements about closed programs.
A closed program is run on a \emph{platform} of $N$ pre-allocated
boolean cells: the context $\G_N$ binds, per cell, a capture variable
and a term variable of type $\REF[\TBOOL]$; the heap $H_N$ allocates
the $N$ live cells; and the environment $\rho_N$ maps each variable to
its cell.
The platform realizes $\G_N$ at every index by construction, is
separation-well-formed because all its capture variables are
access-only, and is compatible with every footprint because all its
cells are live.
A program is a \emph{source} term: it mentions no heap locations, so
the platform configuration $\cfg{H_N}{t\rho_N}$ is well-formed in the
sense of \Cref{app:meta:std}, as the substitution $\rho_N$ introduces
only the live platform cells.

\begin{theorem}[Sequential adequacy]\label{thm:meta:adequacy}
If $\semtyp{C}{\G_N}{t}{E}$, then every configuration reachable from
$\cfg{H_N}{t\rho_N}$ by a sequential run (\Cref{app:meta:exprel}) is
an answer or can step.
\end{theorem}

Given a reachable configuration, instantiate the semantic typing at
budget $k = \readsof{\tau}+1$ for the trace $\tau$ of the run that
reaches it.
The run is then strictly within budget and the safety clause
applies.
The schedule theorems of \Cref{app:meta:std} lift the statement to
arbitrary interleavings.

\begin{theorem}[Adequacy]\label{thm:meta:adequacy-reduce}
If $\semtyp{C}{\G_N}{t}{E}$, then every configuration reachable from
$\cfg{H_N}{t\rho_N}$ by the small-step relation, under any schedule,
is an answer or can step.
\end{theorem}

Let $\cfg{H_N}{t\rho_N}\redsto{\tau_1}\cfg{H_1}{t_1}$ and instantiate
the semantic typing at budget $k = \readsof{\tau_1}+1$.
The safety clause yields a maximal sequential run: one that reaches
an answer within the budget, or one that performs at least $k$
reads.
Prefix safety realizes it in guarded form
(\Cref{app:meta:exprel}), so it is also a run of the unrestricted
relation.
Confluence (\Cref{thm:meta:confluence}) joins the two runs.
If the joining run out of $\cfg{H_1}{t_1}$ is nonempty, $t_1$ can
step.
If it is empty, the joined endpoints coincide.
Then either $t_1$ is the sequential answer up to a renaming of
locations, hence itself an answer, or the sequential run exhausted
the budget, which condition (4) of confluence excludes: the joined
traces have equal read counts, yet one performs fewer than $k$ reads
and the other at least $k$.

\Cref{thm:soundness} is \Cref{thm:meta:adequacy-reduce} composed with
\Cref{thm:meta:fundamental}, and \Cref{cor:memsafe} follows because
the stuck configurations that adequacy excludes include every
use-after-free and double-free: a dead cell matches no reduction rule
(\Cref{app:corecapybara}).

Adequacy also yields the guarantee behind the read-only mode: a
computation whose budget is read-only cannot change the state it runs
over.

\begin{theorem}[Immutability]\label{thm:meta:immutability}
If $\semtyp{C}{\G_N}{t}{E}$ and the kind of $C$ is $\RO$, then every
run of $\cfg{H_N}{t\rho_N}$ to an answer, under any schedule, leaves
every cell of $H_N$ unchanged, in both content and liveness.
\end{theorem}

The read-only kind rules out writes and deallocations, so both
content and liveness are preserved.
The schedule is arbitrary because standardization
(\Cref{thm:meta:standardization}) converts an interleaved run to an
answer into a sequential run with the same endpoints, to which the
sequential argument applies.
\Cref{thm:purity} is this theorem composed with
\Cref{thm:meta:fundamental}.

\subsection{Standardization and Confluence}\label{app:meta:std}

The remaining theorems relate the schedules of the small-step relation
(\Cref{fig:core-smallstep}).
They are stated up to \emph{trace equivalence}: two traces are
equivalent when, at every location, they perform the same sequence of
external accesses and deallocations, where events on cells the trace
itself allocated are internal and do not count.
Formally, the external sequence of $\tau$ at $l$, relative to a set
$A$ of internal locations, is
\[
\begin{array}{lcll}
\mathit{ext}^{A}_{l}(\cdot) &=& \cdot\\
\mathit{ext}^{A}_{l}((\ALLOC l')\,\tau) &=& \mathit{ext}^{A\cup\set{l'}}_{l}(\tau)\\
\mathit{ext}^{A}_{l}((\mu\,l)\,\tau) &=& \mu\ \mathit{ext}^{A}_{l}(\tau)
  &\text{if}\ l\notin A\\
\mathit{ext}^{A}_{l}((\DEALLOC l)\,\tau) &=& \CONSUME\ \mathit{ext}^{A}_{l}(\tau)
  &\text{if}\ l\notin A\\
\mathit{ext}^{A}_{l}(e\,\tau) &=& \mathit{ext}^{A}_{l}(\tau)
  &\text{otherwise,}
\end{array}
\]
and $\tau_1\treq\tau_2$ iff
$\mathit{ext}^{\set{}}_{l}(\tau_1) = \mathit{ext}^{\set{}}_{l}(\tau_2)$
for every $l$.
This is Mazurkiewicz trace equivalence: independent events may be
reordered, conflicting events keep their order.
A term is \emph{well-formed} in $H$ when every location it mentions is
allocated in $H$.
The relations $\redsto{\tau}$ and $\seqredsto{\tau}$ are as in
\Cref{app:meta:exprel}.

\begin{theorem}[Standardization]\label{thm:meta:standardization}
If $t$ is well-formed in $H$ and
$\cfg{H}{t}\redsto{\tau}\cfg{H'}{a}$ with $a$ an answer, then
$\cfg{H}{t}\seqredsto{\tau'}\cfg{H'}{a}$ for some
$\tau'\treq\tau$.
\end{theorem}

The final heap and answer are syntactically identical.
Only the trace is reordered.
The theorem carries no typing hypothesis: the authorization and
non-interference guards of \rruleref{s-par-l} and \rruleref{s-par-r}
already carry the separation content that makes the commutations
sound, and typing is what guarantees the guards never block
(\Cref{sec:core}).
It is the precise form of \Cref{thm:standardization}.

\begin{theorem}[Confluence]\label{thm:meta:confluence}
Let $t$ be well-formed in $H$, and let
$\cfg{H}{t}\redsto{\tau_1}\cfg{H_1}{t_1}$ and
$\cfg{H}{t}\redsto{\tau_2}\cfg{H_2}{t_2}$.
Then there are runs
$\cfg{H_1}{t_1}\redsto{\sigma_1}\cfg{H_1'}{t_1'}$ and
$\cfg{H_2}{t_2}\redsto{\sigma_2}\cfg{H_2'}{t_2'}$
and a permutation $\pi$ of locations such that
\begin{enumerate}
\item $H_2' = \pi\,H_1'$;
\item $t_2'$ equals $\pi\,t_1'$;
\item $\pi(\tau_1\,\sigma_1)\treq\tau_2\,\sigma_2$; and
\item $\readsof{\tau_1\,\sigma_1} = \readsof{\tau_2\,\sigma_2}$.
\end{enumerate}
\end{theorem}

The renaming $\pi$ accounts for the free choice of fresh locations at
allocation.
In the mechanization, condition (2) holds up to the capture
annotations on parallel nodes: their growth depends on the order of
interleaving, but the footprints they denote agree.
Condition (4) does not follow from (3): trace equivalence discounts
events on cells the runs themselves allocate, so it alone does not
pin down read counts.
It is what closes the budget-exhaustion case of
\Cref{thm:meta:adequacy-reduce}.

\begin{corollary}[Schedule determinism]\label{cor:meta:determinism}
If $t$ is well-formed in $H$ and
$\cfg{H}{t}\redsto{\tau_1}\cfg{H_1}{a_1}$ and
$\cfg{H}{t}\redsto{\tau_2}\cfg{H_2}{a_2}$ with $a_1$, $a_2$ answers,
then $H_2 = \pi\,H_1$, $a_2 = \pi\,a_1$, and
$\pi\,\tau_1\treq\tau_2$ for some permutation $\pi$ of locations.
\end{corollary}

The corollary is immediate from \Cref{thm:meta:confluence}: answers do
not reduce, so both closing runs are empty.
It is the precise form of \Cref{thm:confluence} and the operational
content of data-race freedom: no observation of a well-typed program
can depend on how its parallel compositions are scheduled.

\subsection{Mechanization}\label{app:meta:mech}

The development is mechanized in Lean~4.
All results are proved in full.
The development contains no unproven obligations.
The syntax, type system, and operational semantics are as in
\Cref{app:corecapybara}.
The model of
\Cref{app:meta:footprints,app:meta:worlds,app:meta:types,app:meta:exprel,app:meta:semtyping}
lives in the \texttt{Denotation} module, the compatibility lemmas and
\Cref{thm:meta:fundamental} in \texttt{Fundamental}, the platform and
adequacy theorems in \texttt{Safety} and \texttt{SafetyReduce}, and
the schedule theorems in \texttt{Semantics}.
The main theorems correspond to the following Lean declarations:
\begin{center}
\begin{tabular}{ll}
\toprule
Result & Lean declaration\\
\midrule
Fundamental theorem (\Cref{thm:meta:fundamental}) & \texttt{fundamental}\\
Separation (\Cref{lem:meta:sep}) & \texttt{fundamental\_sepcheck}\\
Sequential adequacy (\Cref{thm:meta:adequacy}) & \texttt{adequacy\_platform}\\
Adequacy (\Cref{thm:meta:adequacy-reduce}) & \texttt{adequacy\_platform\_reduce}\\
Immutability (\Cref{thm:meta:immutability}) & \texttt{immutability\_adequacy\_platform\_reduce}\\
Standardization (\Cref{thm:meta:standardization}) & \texttt{standardization}\\
Confluence (\Cref{thm:meta:confluence}) & \texttt{confluence}\\
\bottomrule
\end{tabular}
\end{center}
\section{Translating \capybara{} to \corecapybara{}}\label{app:translation}

\theoremstyle{acmdefinition}
\newtheorem{remark}{Remark}[section]
\theoremstyle{acmplain}

\newcommand{\rentries}[1]{\langle #1\rangle}

This appendix defines the type-preserving translation
$\embed{\cdot}$ from the surface calculus \capybara{} (\Cref{sec:formal})
to the core calculus \corecapybara{} (\Cref{app:corecapybara}),
states its correctness theorem, and proves it.
The translation is the bridge that transfers \corecapybara{}'s metatheory,
soundness, memory safety, and data-race freedom, to \capybara{}.

\subsection{Conventions}\label{app:tr:conv}

We fix the following conventions, which the surface presentation of
\Cref{sec:formal} leaves implicit.

\paragraph{Root capabilities.}
The two special capture variables $\ANY$ and $\FRESH$ are \emph{root capabilities}.
They dissolve into universal and existential capture quantification during translation.
Recall from \Cref{sec:formal:typing} that $\mfree{C}$ removes every root
capability from $C$;
call a capture set or type \emph{root-free} when it mentions no root
capability. All judgments of the translation are read
over root-free sets: whenever a root capability would enter a capture position of a
\corecapybara{} type, the type translation has already consumed it into a bound
capture variable.

\paragraph{Root-directedness.}
A well-formed surface type is \emph{root-directed}: occurrences of $\ANY$
appear only as elements of the \emph{top-level} capture set of a
term-function domain, whose root-free part is access-only, and occurrences
of $\FRESH$ appear only in covariant (spine) positions, never in a function
domain. The domain restriction on $\ANY$ is what \ruleref{abs} and
\ruleref{app} already assume, read one arrow at a time: an $\ANY$ nested
deeper inside a domain sits in the domain of a nested arrow, and that arrow
quantifies it itself. The codomain restriction on $\FRESH$ is a stated
convention, enforced by the surface well-formedness judgment $\wf{\G}{T}$,
which we assume. Under root-directedness each translated arrow binds a
\emph{single} capture variable $\cstar$ for the (at most one) $\ANY$ element
of its domain's top-level capture set, and a call instantiates $\cstar$ with
one capture set, exactly as \ruleref{abs} and \ruleref{app} do with
$T[\ANY\leadsto\set{c}]$ and $T[\ANY\leadsto D]$.

\paragraph{Domain honesty.}
At every \ruleref{app} and \ruleref{consume-app} instance, with domain
top-level capture set $C_T$ and argument $y$ declared at capture set $C_y$,
we require the argument to resolve onto the domain's root-free captures:
every root of $\mfree{C_T}$ is covered, at an equal or stronger mode, by a
root of $C_y$ (the covering order $\sqsubseteq$ of \Cref{app:tr:thm}). The
domain's root-free captures name capabilities the parameter actually holds,
and the function body treats them as part of its footprint: it may assume
them separate from what it captures, and the manufactured lock of
\Cref{app:tr:types} records that assumption, to be paid at every call from
the caller's own knowledge of the argument. An argument whose capture set
under-runs the domain would make the assumption unpayable at the call while
keeping it vacuously true of the actual argument; we exclude such
applications. A domain whose only capture is $\ANY$, the common case in
\Cref{sec:informal}, satisfies the requirement trivially.

\paragraph{Ambient contexts.}
The translation is stated for derivations whose \emph{ambient} context, the
prefix of $\G$ not introduced by a binder of the term under translation,
declares term and type bindings only. Every capture variable of a translated
context is then introduced by a translated binder and carries that binder's
lock or consume authority (\Cref{app:tr:ctx}), which is what pays the fiat
separations of \ruleref{sep-root} (\Cref{lem:tr:invariant}). An ambient
capture binding would arrive with neither, and a fiat separation between two
such roots would have no translated image.

\paragraph{Context-indexed types.}
The type translation resolves capture sets to their roots in the context
where the type is read (\Cref{app:tr:types}), so $\embed{T}$ carries a
context index that we leave implicit. The index is stable: a resolution
consults only the bindings of the variables the resolved set mentions, so
extending the context disturbs nothing, and a translated type weakens to any
extension of its context unchanged.

\paragraph{Canonical derivations.}
The translation is defined on typing derivations, and the theorem asserts
that the \emph{judgment} translates, so we may choose which derivation to
translate. We choose a canonical one, which keeps two sources of slack out
of the translated obligations: gratuitously wide witnesses and gratuitously
widened subjects.
\begin{definition}[Canonical derivations]\label{def:tr:canon}
A derivation of $\typ{C}{\G}{t}{T}$ is \emph{canonical} when:
\begin{enumerate}
\item \emph{(Tight witnesses)} at every \ruleref{app} and
  \ruleref{consume-app} instance, the witness $D$ satisfies
  $\roots{\G}{D}\sqsubseteq\roots{\G}{C_y}$ for the argument's declared
  capture set $C_y$, and $D=\set{}$ when the domain mentions no $\ANY$;
\item \emph{(Declared subjects)} the subject premise of every elimination
  (\ruleref{app}, \ruleref{consume-app}, \ruleref{tapp}, \ruleref{capp})
  concludes at the subject's declared capturing type, up to subtyping
  applied inside covariant positions of its shape: in particular its
  top-level capture set and, for the function eliminations, the arrow's
  capture annotation are the declared ones.
\end{enumerate}
\end{definition}

\begin{lemma}[Canonicalization]\label{lem:tr:canon}
Every derivable typing judgment has a canonical derivation.
\end{lemma}
\begin{proof}
By induction on the derivation, rewriting each offending instance; the two
clauses are established independently, and neither rewrite disturbs the
conclusion judgment.

\emph{(Declared subjects.)} A subject premise is a \ruleref{var} axiom
followed by \ruleref{sub} steps, collected into one step
$\subs{\G}{m\,S_y\capt C_y}{m\,S\capt C}$ by \ruleref{trans}. The
elimination demands that $S$ be a function former, and $S_y$ must be one of
the same former: \ruleref{tvar} only moves \emph{up} to a declared bound, so
a chain from a function shape to a function shape stays within the
congruences. Inversion of \ruleref{capt} and of the congruence
(\ruleref{fun}, \ruleref{tfun}, \ruleref{cfun}) splits the step into
$\subs{\G}{C_y}{C}$, an identical domain (term domains are invariant under
\ruleref{fun}) or a contravariant bound step, and a covariant codomain step
under the binder. Re-associate: perform the elimination at the declared
type, whose domain or bound premise the original premises already serve, and
push the codomain step onto the \emph{conclusion} by \ruleref{sub}, using
stability of subtyping under the elimination's substitution, a routine
property of the surface calculus. The use set is unchanged; for
\ruleref{capp} the charged witness $D$ is part of the term and also
unchanged.

\emph{(Tight witnesses.)} If the domain mentions no $\ANY$ then
$T[\ANY\leadsto D']=T$ for every $D'$, and $D'=\set{}$ satisfies
access-onliness and separation by \ruleref{ao-empty}, \ruleref{sep-empty},
and \ruleref{sep-symm}. Otherwise, by root-directedness the sole $\ANY$
sits in the domain's top-level capture set $C_T$, so the argument premise
contains the top-level subcapturing $\subs{\G}{C_y}{\mfree{C_T}\cup D}$ and
fixes everything below the top level. By the atomic decomposition of
subcapturing (\Cref{lem:tr:atoms}), each atom of $C_y$ resolves, through
declared capture sets, to atoms that are mode-covered by elements of
$\mfree{C_T}\cup D$. Let $D_0$ collect the intermediate atoms of that
resolution that land in $D$. Then: $\subs{\G}{C_y}{\mfree{C_T}\cup D_0}$,
by the same resolution with each landing atom kept in place;
$\subs{\G}{D_0}{D}$, since each collected atom is mode-covered by an element
of $D$ (\ruleref{sc-elem}, \ruleref{sc-mode}); and
$\roots{\G}{D_0}\sqsubseteq\roots{\G}{C_y}$, since each collected atom is a
resolvent of an atom of $C_y$ and resolution only multiplies modes
(\Cref{lem:tr:rootmono}). The rewritten instance uses witness $D_0$: the
argument premise holds by the first fact, $\accessonly{\G}{D_0}$ by
\ruleref{ao-sc} from the second and the original $\accessonly{\G}{D}$, the
separation premise $\sep{\G}{D_0}{|A|}$ by \ruleref{sep-sc} from the second
and the original premise, and for \ruleref{consume-app} likewise
$\consumable{\G}{D_0}$ by \ruleref{con-sc}. The result type $[z:=y]U$ does
not mention $D$, and for \ruleref{consume-app} the charged
$D_0\cup\CONSUME\,D_0$ lies below the original charge by the second fact,
restored by \ruleref{sub}.
\end{proof}

\paragraph{Translation on derivations.}
The translation is defined on canonical typing \emph{derivations}, and is
compositional. The translated term inserts administrative forms, capture
applications, locks and unlocks, packs and unpacks, that carry no source
syntax but discharge the bookkeeping that \corecapybara{} makes explicit. We
write $\embed{\cdot}$ for the translation of every syntactic class;
context~disambiguates.

\subsection{Translating Capture Sets, Bounds, and Permissions}\label{app:tr:sets}

\paragraph{Capture sets.}
$\embed{C}$ acts elementwise, preserving variable names and access modes:
\[
\embed{\set{\mu_1\,\theta_1,\ldots,\mu_n\,\theta_n}}
  = \set{\mu_1\,\theta_1,\ldots,\mu_n\,\theta_n}
  \quad\text{for root-free}\ \theta_i.
\]
Root capabilities never occur in a translated capture position; they are removed by
$\mfree{\cdot}$ or turned into bound capture variables by the type translation.

\paragraph{Mode application.}
Mode application commutes with translation.
\begin{lemma}[Mode commutation]\label{lem:tr:mode}
$\embed{\mu\,C} = \mu\,\embed{C}$ for every access mode $\mu$ and root-free
capture set $C$.
\end{lemma}
\begin{proof}
Immediate: both sides restamp each element of $\embed{C}$ with $\mu$, and
$\embed{\cdot}$ preserves the underlying captures.
\end{proof}

\paragraph{Bounds.}
A surface capture bound is a capture set or a mutability. Capture sets translate
as above. Mutabilities translate to the trivial core bound:
\[
\embed{\epsilon} = \top,\qquad \embed{\RO} = \top.
\]
A mutability bound carries no separation content of its own in \corecapybara{},
where the trivial bound is $\top$. Its kinding force is recovered elsewhere: an
$\RO$ bound additionally contributes a mode-context entry $(\set{c}:\RO)$ to the
lock manufactured at the binder that introduces $c$ (\Cref{app:tr:types}), which
is what lets \rruleref{k-imm} pay the source uses of \ruleref{k-imm} and
\ruleref{sb-kind}.

\paragraph{Permissions.}
Both systems qualify capturing types with a permission $m$, and the
translation carries it over verbatim:
\[
\embed{m\,S\capt C} = m\,\embed{S}\capt \embed{C},
\qquad
\embed{S} = \embed{S}\ \text{(pure case)}.
\]
The permission is a qualifier on the capturing \emph{type}, not an operation
on its capture set: $\embed{C}$ is translated element-wise and $m$ never
touches it. Every rule that reads or writes a permission does so through the
type former: \ruleref{var} preserves the binding's qualifier on its refined
singleton, \rruleref{capt} keeps it rigid on both sides, \rruleref{reader}
introduces the read-only view of a reference at
$\RO\,\REF[T]\capt\set{\RO\,x}$, and \rruleref{read} demands an
$\RO$-qualified subject. Because \rruleref{write} matches only the default
qualifier, a reader-typed alias can never appear as the subject of a write:
the qualifier plays exactly the role of the mechanization's separate reader
type constructor, which realizes the same distinction as a shape. Mode
qualification of capture \emph{sets} ($\RO\,C$, $\CONSUME\,C$) is an
independent, element-wise operation used by charges and subcapturing
(\rruleref{sc-ro}, \rruleref{k-ro}) and is likewise translated unchanged.

\subsection{Translating Types}\label{app:tr:types}

Type translation comes in two variants. The \emph{plain} translation
$\embed{T}$ is used in domains and other invariant or contravariant positions,
where $\FRESH$ never occurs (root-directedness). The \emph{result}
translation $\eemb{T}$ is used in covariant result positions, let-bound terms
and function codomains, where $\FRESH$ may occur; it closes the $\FRESH$ occurrences
into an existential:
\[
\eemb{T} =
\begin{cases}
\EXCAP{c_1,\ldots,c_n}\,\embed{\freshinst{T}{\set{c_1},\ldots,\set{c_n}}} & \text{if $T$ mentions $\FRESH$},\\
\embed{T} & \text{otherwise},
\end{cases}
\]
binding one capture variable per $\FRESH$ occurrence, matching
\ruleref{let}. For root-free $T$ the two variants coincide.

\paragraph{Plain translation.}
The base and capturing cases are homomorphic:
\[
\embed{\top} = \top,
\qquad \embed{X} = X,
\qquad \embed{m\,S\capt C} = m\,\embed{S}\capt \embed{C}.
\]
Every function former, term, consumer, capture, and type function alike,
manufactures a \emph{lock} at its capturing type; their clauses follow the
next definition.

\paragraph{Root entries.}
The manufactured locks are built from resolved roots. For a root set $P$
(\Cref{app:corecapybara}), whose elements are moded capture variables
$\mu\,p$, write $P\vert_p$ for the restriction of $P$ to the atoms over the
variable $p$, and
\[
\rentries{P} \;=\; P\vert_{p_1},\ \ldots,\ P\vert_{p_k},
\qquad p_1,\ldots,p_k\ \text{the distinct capture variables of}\ P,
\]
for the \emph{root entry list} of $P$: one separation-context entry per
capture variable, each collecting every access mode at which its variable
occurs in $P$. No two entries share a variable, so a lock built from
$\rentries{P}$ asserts pairwise separation of \emph{distinct roots} and
never demands a root's separation from itself. For a function type $A$ read
in context $\G$, define its \emph{root set}
\[
R_A \;=\;
\begin{cases}
\roots{\G}{\mfree{C\cup C_T}} & A=m\,(\forall(\alpha\,x:T)U)\capt C,\ \text{with $C_T$ the top-level captures of $T$},\\
\roots{\G}{\mfree{C}} & A=m\,(\forall[c<:B]U)\capt C\ \text{or}\ A=m\,(\forall[X<:S]U)\capt C.
\end{cases}
\]
$R_A$ resolves, once, at the type former, everything the arrow's capture
annotation and, for a term function, its domain's root-free captures reach.
It is the translated counterpart of the \emph{shallow} part of the spine
capture set $|A|$: the deeper positions of $|A|$ belong to the nested arrows
of $A$, each of which manufactures its own lock.

\paragraph{Term functions.}
Let $A=m\,(\forall(x:T)U)\capt C$ be an access-only term-function type, so its
parameter carries consume mode $\alpha=\epsilon$. The translation is
\[
\embed{A} =
m\,\Big(\forall[\cstar<:\top]\ \forall\big(x:\embed{\anyinst{T}{\set{\cstar}}}\big)\
   ([\Psi_A,\Phi_A]\,\eemb{U})\capt W_A\Big)\capt \embed{C},
\]
with the \emph{lock manufacture} and the \emph{codomain annotation}
\[
\Psi_A = \set{\cstar}\,,\ \rentries{\embed{R_A}},
\qquad
\Phi_A = \emptyset,
\qquad
W_A = \set{\cstar}\cup\embed{R_A}.
\]
The separation context $\Psi_A$ has one entry for the parameter root
$\set{\cstar}$ and one entry per root of $R_A$; through \rruleref{sep-lock}
it records that the parameter is separated from every root the function's
annotation and domain reach, and that those roots are pairwise separated.
These are exactly the separations the body may assume by fiat
(\ruleref{sep-root}): the body's charge is $C\cup\set{x}$, whose roots are
$R_A\cup\set{\cstar}$, the entry variables (\Cref{lem:tr:invariant}). The
codomain annotation $W_A$ is the same root set read as a capture set: the
translated body charge subcaptures it by resolving each of its own atoms
upward (\Cref{lem:tr:resolve}), which is what lets the translated
abstraction carry the annotation its type states
(\Cref{app:tr:mainproof}). The mode context of a term-function lock is
empty; mode entries arise only at $\RO$-bounded capture binders below. The
single bound variable $\cstar$ instantiates the (at most one) $\ANY$ of the
domain, and $\anyinst{T}{\set{\cstar}}$ denotes $T$ with it replaced by
$\set{\cstar}$.

\paragraph{Consume-mode term functions.}
When the parameter carries $\alpha=\CONSUME$, the arrow translates to a
\corecapybara{} \emph{consumer} type, whose argument is packed with its own
capture witness so that the body may consume the root:
\[
\begin{aligned}
\embed{m\,(\forall(\CONSUME\,x:T)U)\capt C} = m\,\Big(
&\forall\big(\langle\cstar,x\rangle:\EXCAP{\cstar}\embed{\anyinst{T}{\set{\cstar}}}\big)\\[-2pt]
&\ ([\Psi_A,\Phi_A]\,\eemb{U})\capt W_A\Big)\capt \embed{C},
\end{aligned}
\]
with
\[
\Psi_A = \set{\cstar,\CONSUME\,\cstar}\,,\ \rentries{\embed{R_A}},
\qquad
\Phi_A = \emptyset,
\qquad
W_A = \set{\cstar,\CONSUME\,\cstar}\cup\embed{R_A}.
\]
The consumer binds $\CONSUME\,\cstar<:\top$ (\rruleref{consumer}), so its
body may consume the root; accordingly the parameter entry and the codomain
annotation carry both modes of $\cstar$, the image of the body charge
$D=\set{c,\CONSUME\,c}$ of \ruleref{abs} in the consume case.

\paragraph{Capture functions.}
Let $A=m\,(\forall[c<:B]U)\capt C$. The translation is
\[
\embed{A} =
m\,\Big(\forall[c<:\embed{B}]\ ([\Psi_A,\Phi_A]\,\eemb{U})\capt W_A\Big)\capt \embed{C},
\qquad
\Psi_A = \set{c}\,,\ \rentries{\embed{R_A}},
\qquad
W_A = \set{c}\cup\embed{R_A},
\]
with one mode entry paying an $\RO$ bound:
\[
\Phi_A =
\begin{cases}
(\set{c}:\RO) & \text{if } B = \RO,\\
\emptyset & \text{otherwise.}
\end{cases}
\]
When $B=\RO$ the parameter is bounded by $\top$ in the translated type
($\embed{\RO}=\top$), and the mode entry restores that $\set{c}$ has kind
$\RO$ inside the codomain through \rruleref{k-imm}; a capture-set bound
translates to the capture set itself and contributes no mode entry. The
parameter entry $\set{c}$ is not redundant: \ruleref{cabs} charges the body
at $C\cup\set{c}$, so the body may pair $c$ against the roots of $C$ in its
own separation premises, and the entry is what pays those pairs. The bound
$B$ contributes nothing to $\Psi_A$: locks record what the \emph{body}
assumes, and the body's charge reaches $B$ only through $c$ itself.

\paragraph{Type functions.}
Let $A=m\,(\forall[X<:S]U)\capt C$. The translation is
\[
\embed{A} =
m\,\Big(\forall[X<:\embed{S}]\ ([\Psi_A,\Phi_A]\,\eemb{U})\capt W_A\Big)\capt \embed{C},
\qquad
\Psi_A = \rentries{\embed{R_A}},
\qquad
\Phi_A = \emptyset,
\qquad
W_A = \embed{R_A}.
\]
A type binder introduces no root, so the lock has no parameter entry, only
the pairwise entries of the annotation's roots: a type-function body may
still fire \ruleref{sep-root} between two captured roots, and its lock is
what pays for it. Correspondingly \ruleref{tapp} carries no separation
premise, and the translated type application discharges its unlock from the
ambient pairwise separations alone (\Cref{app:tr:mainproof}): the surface
rule and the manufactured lock ask for the same thing.

\begin{remark}[Locks are shallow, root-level, and computed once]\label{rem:tr:shallow}
Three design points carry the appendix.
First, locks are \emph{shallow}: $R_A$ reads only the arrow's own
annotation and domain, never the codomain's spine. A separation that a
nested body relies on lives in the lock of the nested arrow that abstracts
it, is carried through substitution inside that arrow's type, and is paid
where that arrow is eliminated. The deep spine capture set $|A|$ remains a
\emph{surface} device: the \ruleref{app} premise checks the argument against
it eagerly because surface fiat must be justified under instantiation, but
no translated obligation ever mentions the deep part.
Second, entries are \emph{root-level}: each entry is the moded atom set of
a single resolved root, so a body-level fiat pair is read off the lock by
\rruleref{sep-lock} directly, and a use site consumes an entry by resolving
its own sets \emph{upward} to roots (\Cref{lem:tr:resolve}) and descending
by \rruleref{sep-sc}; no stored entry is ever projected downward, so no
entry needs to carry a set together with its resolution.
Third, entries are computed once, at the type former, and travel with the
type as stored syntax: substitution acts on them componentwise
(\Cref{lem:tr:subst}), and widening the annotation re-relates them by
\rruleref{boxed-sat} (\Cref{lem:tr:subtype}). Stability under both is
possible because term substitution does not touch capture variables and the
domain of a term function is rigid under surface subtyping (\ruleref{fun}).
\end{remark}

\subsection{Translating Contexts}\label{app:tr:ctx}

Contexts translate pointwise:
\[
\embed{\emptyset}=\emptyset,\quad
\embed{\G,x:T}=\embed{\G},x:\embed{T},\quad
\embed{\G,X<:S}=\embed{\G},X<:\embed{S},\quad
\embed{\G,\alpha\,c<:B}=\embed{\G},\alpha\,c<:\embed{B},
\]
where the source consume mode $\alpha\in\set{\epsilon,\CONSUME}$ maps to itself
and $\embed{\epsilon}=\embed{\RO}=\top$ on bounds.

The manufactured locks do not live in $\embed{\G}$ globally. They enter the
context only where the translated \emph{term} pushes a lock, at a
\rruleref{lock}, \rruleref{consumer}, or \rruleref{unpack} binder. We write
$\embed{\G}^{\LOCK}$ for the translated context \emph{interleaved with the
frames pushed by the enclosing translated binders}, and define it alongside
the term translation. A source context extension contributes one frame per
binder passed on the way from the ambient context to the site:
\begin{itemize}
\item a term binder $\lambda(\alpha\,x:T)$ with arrow $A$ contributes
  $\alpha\,\cstar<:\top,\ x:\embed{\anyinst{T}{\set{\cstar}}},\
  \LOCK[\Psi_A,\Phi_A]$, where in the consume case the pair
  $\cstar,x$ is bound by a single \rruleref{consumer} frame;
\item a capture binder $\lambda[c<:B]$ with arrow $A$ contributes
  $\epsilon\,c<:\embed{B},\ \LOCK[\Psi_A,\Phi_A]$;
\item a type binder $\lambda[X<:S]$ with arrow $A$ contributes
  $X<:\embed{S},\ \LOCK[\Psi_A,\Phi_A]$;
\item a root-free \ruleref{let} contributes $x:\embed{T}$, and a
  $\FRESH$-carrying \ruleref{let} contributes the \rruleref{unpack} frame
  $\CONSUME\,c_1<:\top,\ldots,\CONSUME\,c_n<:\top,\
  \LOCK[\Psi_w,\Phi_w],\ x:\embed{\freshinst{T}{\set{c_1},\ldots,\set{c_n}}}$
  with the ownership lock of \rruleref{unpack}.
\end{itemize}
The lock of each frame is available to every subderivation nested inside its
binder, which is what makes the manufactured separations of
\Cref{app:tr:types} usable in the body. The kill operators
$\ominus$ that core \rruleref{let} and \rruleref{unpack} apply to their
continuations are \emph{not} part of $\embed{\G}^{\LOCK}$: the main
induction produces continuation derivations in the unkilled context and
places them under the kill by \Cref{lem:tr:kill}.

\subsection{Translating Terms}\label{app:tr:terms}

Term translation is directed by the typing derivation: the translated term
depends on which typing rule concludes the source derivation, since a bare
surface variable may be typed by \ruleref{var}, \ruleref{readonly}, or
\ruleref{fresh}. The homomorphic and short cases are collected in the table
below; the cases that insert nontrivial administrative structure are
described in the paragraphs that follow.

\begin{center}\footnotesize
\renewcommand{\arraystretch}{1.3}
\begin{tabular}{@{}lll@{}}
\toprule
Source rule & Source term & Translated term $\embed{\cdot}$\\
\midrule
\ruleref{var}      & $x$              & $x$\\
\ruleref{tabs}     & $\lambda[X<:S]t$ & $\lambda[X<:\embed{S}]\,\LOCK[\Psi_A,\Phi_A]\,\embed{t}$\\
\ruleref{tapp}     & $x[S']$          & $\LET x_1=x[\embed{S'}]\IN\UNLOCK\,x_1$\\
\ruleref{cabs}     & $\lambda[c<:B]t$ & $\lambda[c<:\embed{B}]\,\LOCK[\Psi_A,\Phi_A]\,\embed{t}$\\
\ruleref{capp}     & $x[D]$           & $\LET x_1=x[\embed{D}]\IN\UNLOCK\,x_1$\\
\ruleref{let}      & $\LET x=t\IN u$  & $\LET x=\embed{t}\IN\embed{u}$ \ ($T$ root-free)\\
\ruleref{let}      & $\LET x=t\IN u$  & $\LET\langle c_1,\ldots,c_n,x\rangle=\embed{t}\IN\embed{u}$ \ ($T$ with $\FRESH$)\\
\bottomrule
\end{tabular}
\end{center}

\noindent
Here $[\Psi_A,\Phi_A]$ is the lock manufactured at the binder's arrow type
(\Cref{app:tr:types}); all three abstraction forms push it, and all three
application forms open it with an \UNLOCK{}. The two \ruleref{let} rows are
distinguished by whether the bound term's declared type $T$ carries
$\FRESH$: a root-free $T$ translates to a \rruleref{let}, whereas a $T$ with
$\FRESH$ translates to a \rruleref{unpack}, whose $n$ freshly bound
witnesses are exactly the roots $c_1,\ldots,c_n$ that $\eemb{T}$ packages.
The charge match is exact: the extra use set
$\set{c_1,\ldots,c_n,\CONSUME\,c_1,\ldots,\CONSUME\,c_n}$ that
\rruleref{unpack} adds to the continuation is the image of the body charge
$D\cup\CONSUME\,D$ of the source \ruleref{let}.

\paragraph{Abstraction.}
An access-only abstraction translates to a capture lambda, a term lambda,
and a lock, in that order, so that its body runs under the manufactured
lock:
\[
\embed{\lambda(x:T)t} =
\lambda[\cstar<:\top]\ \lambda\big(x:\embed{\anyinst{T}{\set{\cstar}}}\big)\
\LOCK[\Psi_A,\Phi_A]\ \embed{t},
\qquad A=(\forall(x:T)U)\capt C.
\]
A consume abstraction translates to a consumer lambda, whose bound pair
$\langle\cstar,x\rangle$ supplies the consumable root:
\[
\embed{\lambda(\CONSUME\,x:T)t} =
\lambda\big(\langle\cstar,x\rangle:\EXCAP{\cstar}\embed{\anyinst{T}{\set{\cstar}}}\big)\
\LOCK[\Psi_A,\Phi_A]\ \embed{t}.
\]
In both cases the translated context at $\embed{t}$ is $\embed{\G}^{\LOCK}$
extended with the binder's frame, as in \Cref{app:tr:ctx}.

\paragraph{Application.}
An application instantiates the callee's capture parameter with the
argument's $\ANY$-witness, applies, then unlocks:
\[
\embed{x\,y} =
\LET x_1 = x[\embed{D}]\IN \LET x_2 = x_1\,y\IN \UNLOCK\,x_2,
\]
where $D$ is the canonical witness of \Cref{def:tr:canon}. The capture
application $x[\embed{D}]$ instantiates $\cstar:=\embed{D}$: the callee's
translated type binds $\cstar<:\top$, and \rruleref{sub} with
\rruleref{sb-top} narrows $\top$ to $\embed{D}$ before \rruleref{capp}. The
term application then produces a value of the modal codomain
$[\Psi_A,\Phi_A]\eemb{U}$ instantiated at the call, and $\UNLOCK\,x_2$ opens
it, discharging the satisfaction obligation of \rruleref{sat}. The
obligation asks for pairwise separation of the instantiated entries:
$\embed{D}$ against each root entry, the image of the source premise
$\sep{\G}{D}{|(\forall(z:T)U)\capt C|}$ at the shallow part of the spine,
and the root entries pairwise, which the ambient locks supply
(\Cref{cor:tr:sat}). The access-only requirement of core \rruleref{capp} on
$\embed{D}$ is the image of the \ruleref{app} premise $\accessonly{\G}{D}$
through \Cref{lem:tr:kind}.

\paragraph{Consume application.}
A consume application packs the argument with its witness and applies the
consumer, then unlocks:
\[
\embed{x\,y} = \LET x_2 = x\,\langle\embed{D},y\rangle\IN \UNLOCK\,x_2.
\]
The synthesized pack $\langle\embed{D},y\rangle$ has type
$\EXCAP{\cstar}\embed{\anyinst{T}{\set{\cstar}}}$ by \rruleref{pack}: with
$n=1$ the pairwise disjointness is vacuous, $\embed{D}$ is access-only and
consumable (the source premise $\consumable{\G}{D}$), and the pack charges
$\embed{D}\cup\CONSUME\,\embed{D}$. \rruleref{consumer-app} then sequences this
charge before $\set{x}$, which \rruleref{seq-union} splits into an access leg
$\seqcomp{}{\embed{D}}{\set{x}}$, free by \rruleref{seq-access-only} from the
access-only premise, and a consume leg
$\seqcomp{}{\CONSUME\,\embed{D}}{\set{x}}$, the translated image
(\Cref{lem:tr:mode}, \Cref{lem:tr:seq}) of the source premise
$\seqcomp{\G}{\CONSUME\,D}{\set{x}}$: the consumption of the argument's
witness sequences before the consumer runs.

\paragraph{Type and capture application.}
$\embed{x[S']}=\LET x_1=x[\embed{S'}]\IN\UNLOCK\,x_1$ and
$\embed{x[D]}=\LET x_1=x[\embed{D}]\IN\UNLOCK\,x_1$ open the manufactured
lock exactly as an application does. For a type application the
instantiated entries are the root entries alone, so the unlock carries no
image of a source premise, matching \ruleref{tapp}, which has none: the
ambient locks pay everything (\Cref{cor:tr:sat}). For a capture application
the entries are $\embed{D}$ and the root entries, the image of the
\ruleref{capp} premise $\sep{\G}{D}{|A|}$ at the shallow spine; the witness
$D$ is part of the source term and, by the revised \ruleref{capp}, part of
the source charge, which is what keeps its roots inside the ambient
invariant's reach (\Cref{lem:tr:invariant}).

\paragraph{Read-only occurrence.}
A source occurrence typed by \ruleref{readonly} translates to the core reader
\emph{value}: for a reference-shaped subject $x$, $\embed{x}=\RO\,x$, and
\rruleref{reader} produces $\RO\,\REF[\embed{T}]\capt\set{\RO\,x}
=\eemb{\RO\,\REF[T]\capt\set{\RO\,x}}$ together with the charge $\set{}$
primitively, with no coercion. This matches the mechanization, whose
\ruleref{readonly} is specific to reference types and translates to its reader
value; the $\RO$ type qualifier is the flattened image of its reader shape.
Following the mechanization, we read \ruleref{readonly} as specific to
reference-shaped subjects (\Cref{app:tr:sets}).
A \ruleref{var} occurrence at an $\RO$-qualified $x$
needs no coercion at all: the qualifier rides on the type, which
\rruleref{var} preserves.

\paragraph{Fresh occurrence.}
A source occurrence typed by \ruleref{fresh}, retyping $x$ from
$T[\FRESH\leadsto D_1,\ldots,D_n]$ to $T$, translates to a pack at the
translated existential $\eemb{T}$:
\[
\embed{x} = \big\langle\ \embed{D_1},\ldots,\embed{D_n}\ ,\ x\ \big\rangle,
\]
one witness per bound root of
$\eemb{T}=\EXCAP{c_1,\ldots,c_n}\embed{\freshinst{T}{\set{c_1},\ldots,\set{c_n}}}$
(one root per $\FRESH$ occurrence, matching \ruleref{let}). The pack charges
$\bigcup_i\embed{D_i}\cup\CONSUME\bigcup_i\embed{D_i}$, exactly the translated
image of the source charge $D\cup\CONSUME\,D$.

\paragraph{Subsumption.}
A source \ruleref{sub} step translates by \Cref{lem:tr:subtype} (subtyping
preservation) and \Cref{lem:tr:sc} (subcapturing preservation), followed by
core \rruleref{sub}.

\paragraph{State, parallelism, and conditionals.}
The memory primitives, parallel composition, and conditionals
(\Cref{sec:formal:state}) translate homomorphically, with two points of
interest.
An allocation keeps its form, and its types line up by the result
translation: $\eemb{\REF[T]\capt\set{\FRESH}} =
\EXCAP{c}\REF[\embed{T}]\capt\set{c}$, the type of core
\rruleref{alloc}.
A read inserts the reader that core \rruleref{read} demands:
\[
\embed{\READ x} = \LET r = \RO\,x \IN \READ r,
\]
where \rruleref{reader} types $\RO\,x$ at
$\RO\,\REF[\embed{T}]\capt\set{\RO\,x}$ and the let charges the
translated image $\set{\RO\,x}$ of the source charge.
Writes, deallocations, parallel compositions, and conditionals map to
their core counterparts unchanged; their premises are the translated
images of the source premises, using \Cref{lem:tr:sep} for the
separation premise of \ruleref{par}, whose payability hypothesis
\Cref{lem:tr:invariant} discharges at the site, and \Cref{lem:tr:kind} for
the consumable premise of \ruleref{dealloc}.

\subsection{The Translation Theorem}\label{app:tr:thm}

The consuming forms of the surface calculus charge the access of what they
consume alongside its consumption, \ruleref{fresh} and \ruleref{consume-app}
charging $D\cup\CONSUME\,D$ and \ruleref{let} and \ruleref{abs} binding
$\set{c,\CONSUME\,c}$, exactly as their translated images \rruleref{pack},
\rruleref{consumer-app}, \rruleref{unpack}, and \rruleref{consumer} do. The
translated use set is therefore bounded by $\embed{C}$ itself.

\begin{theorem}[Typability preservation]\label{thm:tr:main}
If $\typ{C}{\G}{t}{T}$ in \capybara{}, then the translation of a canonical
derivation of it (\Cref{lem:tr:canon}) yields a \corecapybara{} term
$\embed{t}$ with
\[
\typ{C'}{\embed{\G}^{\LOCK}}{\embed{t}}{\eemb{T}}
\qquad\text{for some $C'$ with}\qquad
\roots{\embed{\G}^{\LOCK}}{C'}\sqsubseteq\roots{\embed{\G}^{\LOCK}}{\embed{C}},
\]
where $\embed{\G}^{\LOCK}$ is the lock-interleaved translated context of
\Cref{app:tr:ctx} and $\eemb{T}$ the result translation of \Cref{app:tr:types}.
\end{theorem}

The side condition is stated at the level of roots. Write $P\sqsubseteq P'$ when
every $\mu\,\theta\in P$ has some $\mu'\,\theta\in P'$ with $\mu\preceq\mu'$,
extending the mutability order with $\CONSUME\preceq\CONSUME$. This is a
statement-level device only: it never appears as the premise of a separation
rule (a variable's runtime footprint can strictly under-run its root, so root
covering does not transport separation). Eager surface charging and lazy core charging differ
exactly below variable resolution: the surface charges a variable occurrence
$\set{x}$ eagerly, while the core charges its uses at eliminations, but $\set{x}$
and the declared capture set it resolves to have the same roots, so the two agree
once resolved. The consumers of the theorem, the core access-only, consumable,
and accessible predicates and the kill operation, are themselves root-based, so
the root-level bound is exactly what they need.

The proof is in \Cref{app:tr:mainproof}, by induction on the canonical source
derivation, one case per rule, using the lemma stack of \Cref{app:tr:lemmas}.
Because the translated term is typed against $\embed{T}$ and its uses are
root-covered by $\embed{C}$, and \corecapybara{}'s dynamic semantics
preserves both, the metatheory of \Cref{sec:metatheory} transfers to
\capybara{}.

\begin{corollary}[Transfer of guarantees]\label{cor:tr:transfer}
A well-typed \capybara{} program enjoys type soundness
(\Cref{thm:soundness}), memory safety (\Cref{cor:memsafe}), and data-race
freedom (\Cref{thm:separation}, \Cref{thm:confluence}).
\end{corollary}
\begin{proof}
By \Cref{thm:tr:main} the translation of the program is a well-typed
\corecapybara{} program, to which the theorems of \Cref{sec:metatheory}
apply. The surface constructs take their operational meaning through the
translation (\Cref{sec:formal:state}), so the guarantees, which are stated
over the traces and final memories of \corecapybara{} runs, are statements
about the program's own runs. The administrative forms the translation
inserts do not disturb them: capture applications, type applications,
applications of the inserted lambdas, locks, unlocks, packs and unpacks
reduce by \rruleref{s-capp}, \rruleref{s-tapp}, \rruleref{s-app},
\rruleref{s-unwrap}, and \rruleref{s-unpack}, each of which leaves the heap
and the trace untouched, so every heap event of a run belongs to a source
construct.
\end{proof}

\subsection{Lemma Stack}\label{app:tr:lemmas}

The proof of \Cref{thm:tr:main} rests on the following lemmas. Except where
noted, each is by induction on the indicated derivation.

\paragraph{L1: Mode monotonicity.}
\Cref{lem:tr:mode} already records $\embed{\mu\,C}=\mu\,\embed{C}$. We add
monotonicity of qualification.
\begin{lemma}[Qualification monotonicity]\label{lem:tr:qualmono}
If $\subs{\embed{\G}^{\LOCK}}{\embed{C_1}}{\embed{C_2}}$ then
$\subs{\embed{\G}^{\LOCK}}{\mu\,\embed{C_1}}{\mu\,\embed{C_2}}$ for every $\mu$.
\end{lemma}
\begin{proof}
For $\mu=\epsilon$ this is the hypothesis. For $\mu=\RO$ it is
\rruleref{sc-ro-mono}; for $\mu=\CONSUME$ it is \rruleref{sc-consume-mono}.
\end{proof}

\paragraph{L2: Roots.}
Root resolution (\Cref{app:corecapybara}) underlies the lock manufacture,
and four of its properties recur. Throughout, $\sqsubseteq$ is the covering
order of \Cref{app:tr:thm}, and we use silently that resolution is
multiplicative, $\roots{\G}{\mu\,C}=\mu\,\roots{\G}{C}$, and distributes
over union, both immediate from its definition, and that it commutes with
the translation, $\roots{\embed{\G}^{\LOCK}}{\embed{C}} =
\embed{\roots{\G}{C}}$ for root-free $C$, since the translation preserves
every binding's capture set elementwise and the interleaved frames bind no
term variables that $C$ mentions.

\begin{lemma}[Atomic decomposition]\label{lem:tr:atoms}
Say an atom $\mu\,\theta$ \emph{lands} in $C_2$ over $\G$ when either
(i) some $\mu'\,\theta\in C_2$ with $\mu\preceq\mu'$ (extending the
mutability order with $\CONSUME\preceq\CONSUME$), or
(ii) $\theta$ is a term variable $x:m\,S\capt C_x\in\G$, $\mu$ is a
mutability, and every atom of $\mu\,C_x$ lands in $C_2$.
Then $\subs{\G}{C_1}{C_2}$ holds if and only if every atom of $C_1$ lands
in $C_2$.
\end{lemma}
\begin{proof}
($\Leftarrow$) Per atom: clause (i) is \ruleref{sc-elem} after
\ruleref{sc-mode} (for $\CONSUME$ atoms the covering mode is equal, so
\ruleref{sc-elem} alone); clause (ii) is \ruleref{sc-var} qualified by
\ruleref{sc-ro-mono} when $\mu=\RO$, followed by \ruleref{sc-trans} on the
inner landings; \ruleref{sc-union} assembles the atoms.
($\Rightarrow$) By induction on the subcapturing derivation, using two
closure properties of landing, each by induction on the landing derivation:
landing composes (an atom landing in $C_2$ lands in $C_3$ whenever every
atom of $C_2$ does), and landing is stable under qualification (if
$\mu\,\theta$ lands in $C_2$ then $\mu'(\mu\,\theta)$ lands in $\mu'C_2$).
\ruleref{sc-elem} and \ruleref{sc-var} are landing axioms,
\ruleref{sc-mode} and \ruleref{sc-ro-mono} use stability,
\ruleref{sc-trans} uses composition, and \ruleref{sc-union} splits.
\end{proof}

\begin{lemma}[Root monotonicity]\label{lem:tr:rootmono}
If $\subs{\G}{C_1}{C_2}$ then $\roots{\G}{C_1}\sqsubseteq\roots{\G}{C_2}$.
\end{lemma}
\begin{proof}
By induction on the derivation. \ruleref{sc-elem}: containment of atoms
gives containment of their resolutions. \ruleref{sc-var}:
$\roots{\G}{\set{x}}=\roots{\G}{C_x}$ by definition, so the two sides are
equal. \ruleref{sc-mode} and \ruleref{sc-ro-mono}: multiplicativity, with
the covering mode weakened by $\preceq$. \ruleref{sc-union}: resolution
distributes. \ruleref{sc-trans}: $\sqsubseteq$ is transitive.
\end{proof}

\begin{lemma}[Resolution reach]\label{lem:tr:resolve}
In \corecapybara{}, $\subs{\G}{C}{\roots{\G}{C}}$ for every capture set $C$
over $\G$.
\end{lemma}
\begin{proof}
By induction on the resolution. A capture-variable atom resolves to itself
(\rruleref{refl} via \rruleref{sc-elem}). A term-variable atom $\mu\,x$ with
$x:m\,S\capt C_x\in\G$: $\subs{\G}{\set{x}}{C_x}$ by \rruleref{sc-var},
qualified to $\mu$ by \Cref{lem:tr:qualmono}, then \rruleref{sc-trans} with
the inductive $\subs{\G}{\mu\,C_x}{\roots{\G}{\mu\,C_x}}$ and
multiplicativity. \rruleref{sc-union} assembles.
\end{proof}

\begin{lemma}[Kinding transfers to roots]\label{lem:tr:kindroot}
If $\typs{\G}{C}{m}$ then $\typs{\G}{\roots{\G}{C}}{m}$. Moreover, if
additionally $\accessonly{\G}{C}$, then every atom $\mu\,p$ of
$\roots{\G}{C}$ is \emph{read-only}: $\mu=\RO$, or $\mu=\epsilon$ and
$p<:\RO\in\G$.
\end{lemma}
\begin{proof}
Both claims by induction on the kinding derivation; only $m=\RO$ is
nontrivial (\ruleref{k-rw} makes $m=\epsilon$ universal).
\ruleref{k-empty}, \ruleref{k-union}: resolution distributes.
\ruleref{k-ro}: $C=\RO\,C'$, so $\roots{\G}{C}=\RO\,\roots{\G}{C'}$ is again
$\RO$-qualified, and \ruleref{k-ro} re-applies; its atoms are $\RO$-marked
except $\CONSUME$-marked ones, which access-onliness excludes.
\ruleref{k-imm}: $C=\set{c}$ with $c<:\RO\in\G$ resolves to itself, and its
atom is read-only by the bound.
\ruleref{k-sc}: from $\subs{\G}{C_1}{C_2}$ and the inductive claims for
$C_2$, \Cref{lem:tr:rootmono} covers each atom $\mu\,p$ of
$\roots{\G}{C_1}$ by some $\mu'\,p$ of $\roots{\G}{C_2}$ with
$\mu\preceq\mu'$; then $\subs{\G}{\set{\mu\,p}}{\set{\mu'\,p}}$
(\ruleref{sc-mode}, \ruleref{sc-elem}) and $\subs{\G}{\set{\mu'\,p}}{\roots{\G}{C_2}}$
give the kinding by \ruleref{k-sc}, and read-onliness is closed downward
under $\preceq$: $\mu\preceq\mu'=\RO$ forces $\mu=\RO$, and an $\RO$-bounded
$p$ stays $\RO$-bounded. Access-onliness descends by \ruleref{ao-sc}.
\end{proof}

\paragraph{L3: Subcapturing preservation.}
\begin{lemma}[Subcapturing preservation]\label{lem:tr:sc}
If $\subs{\G}{C_1}{C_2}$ then $\subs{\embed{\G}^{\LOCK}}{\embed{C_1}}{\embed{C_2}}$.
\end{lemma}
\begin{proof}
By induction on the surface subcapturing derivation.
\ruleref{sc-trans}, \ruleref{sc-union} map to their core namesakes by the IH.
\ruleref{sc-elem}: $C_1\subseteq C_2$ gives $\embed{C_1}\subseteq\embed{C_2}$,
so \rruleref{sc-elem}. \ruleref{sc-mode}: \Cref{lem:tr:qualmono}, or directly
\rruleref{sc-mode}. \ruleref{sc-ro-mono}: \rruleref{sc-ro-mono} on the IH.
\ruleref{sc-var}: from $x:m\,S\capt C\in\G$ the translated binding is
$x:m\,\embed{S}\capt \embed{C}\in\embed{\G}$, so \rruleref{sc-var} gives
$\subs{}{\set{x}}{\embed{C}}$, the image of the source conclusion; the
permission qualifies the binding's type and plays no role in resolution.
\end{proof}

\paragraph{L4: Capability kinding, access-only, consumable, accessible.}
\begin{lemma}[Predicate preservation]\label{lem:tr:kind}
Each surface predicate implies its core image:
$\typs{\G}{C}{m}\Rightarrow\typs{\embed{\G}^{\LOCK}}{\embed{C}}{m}$;
$\accessonly{\G}{C}\Rightarrow\accessonly{\embed{\G}^{\LOCK}}{\embed{C}}$; and
$\consumable{\G}{C}\Rightarrow\consumable{\embed{\G}^{\LOCK}}{\embed{C}}$.
\end{lemma}
\begin{proof}
For kinding, induction on the surface derivation. \ruleref{k-empty},
\ruleref{k-union}, \ruleref{k-rw}, \ruleref{k-ro} map to their core namesakes.
\ruleref{k-sc} uses \Cref{lem:tr:sc} and \rruleref{k-sc}. \ruleref{k-imm}:
a surface $\RO$-bound $c<:\RO\in\G$ translates to $c<:\top$ together with the mode
entry $(\set{c}:\RO)$ that the binder's lock records in its $\Phi$
(\Cref{app:tr:types}); \rruleref{k-imm} reads that entry
to conclude $\typs{}{\set{c}}{\RO}$.

For access-only and consumable, the surface rules are inductive
(\ruleref{ao-access}, \ruleref{ao-union}, \ruleref{ao-sc}, and duals). Core
defines both via roots. \ruleref{ao-access} gives $\accessonly{}{\set{m\,c}}$,
whose translated root set $\set{m\,c}$ has no $\CONSUME$ mark. Unions map to unions.
\ruleref{ao-sc} uses \Cref{lem:tr:sc} together with \Cref{lem:tr:rootmono}
read in the core: shrinking a set covers its roots, and neither predicate
can acquire a $\CONSUME$ mark or a non-$\CONSUME$ binding by shrinking.
The accessible predicate is preserved likewise: the translated roots carry no
$\KILLED$ binding because the source context has no killed variables and the
translated kills are introduced only where the source \ruleref{let} consumes,
tracked by \Cref{lem:tr:kill}.
\end{proof}

\paragraph{L5: Substitution and instantiation.}
The translated eliminations substitute a term variable, a type, or a capture
set into a translated type. Away from the manufactured locks and codomain
annotations, translation and substitution commute on the nose; at the locks,
substitution replaces the resolved entries by written-level sets, and one
subtyping cast re-relates the two.
\begin{lemma}[Substitution commutation]\label{lem:tr:subst}
$\embed{[X:=S]U}=[X:=\embed{S}]\embed{U}$, and likewise for
$\eemb{\cdot}$: type substitution touches no capture set. For term and
capture substitution, $\embed{\cdot}$ and the substitution agree on every
position except the stored entries $\Psi_A$ and annotations $W_A$ of
function types: domains, outer capture sets, and all other capture positions
are translated elementwise, where
$[z:=y]\embed{C}=\embed{[z:=y]C}$ and $[c:=\embed{D}]\embed{C}=\embed{[c:=D]C}$
hold by definition. On a stored entry, substitution acts componentwise: a
term substitution is the identity (entries contain only capture variables),
and a capture substitution $[c:=\embed{D}]$ replaces the atoms of the entry
$P\vert_c$ by the correspondingly qualified $\embed{D}$ and leaves the other
entries unchanged.
\end{lemma}
\begin{proof}
By induction on the type. Entries and annotations are stored root sets
(\Cref{app:tr:types}); term variables do not occur in them, and a capture
substitution acts on capture sets elementwise by definition.
\end{proof}

\begin{lemma}[Instantiation cast]\label{lem:tr:cast}
Let a translated elimination at site context $K=\embed{\G}^{\LOCK}$
instantiate a declared function type as in \Cref{app:tr:terms}: for
\ruleref{app} and \ruleref{consume-app}, $\cstar:=\embed{D}$ and $[x:=y]$
with a tight witness ($\roots{\G}{D}\sqsubseteq\roots{\G}{C_y}$) under
domain honesty; for \ruleref{capp}, $[c:=\embed{D}]$ with the witness $D$
of the source term; for \ruleref{tapp}, $[X:=\embed{S'}]$. Then the
substituted image of the translated codomain subtypes the translation of
the substituted source codomain:
\[
\subs{K}{\ \sigma\big(([\Psi_A,\Phi_A]\,\eemb{U})\capt W_A\big)\ }{\
  ([\Psi_{A'},\Phi_{A'}]\,\eemb{U'})\capt W_{A'}},
\]
where $\sigma$ is the substitution performed, $U'=\sigma U$ at the source
level, and $A'$ is the arrow of the instantiated type, read in $\G$.
\end{lemma}
\begin{proof}
For \ruleref{tapp} the two sides are equal by \Cref{lem:tr:subst}. For the
others, by induction on $U$; all positions except modal wrappers and
annotations agree on the nose by \Cref{lem:tr:subst} (in particular every
domain, which is invariant, mentions the substituted variables only in
elementwise-translated capture sets). At an annotation, the substituted
$\sigma W$ and the re-resolved $W'$ are related by
$\subs{K}{\sigma W}{W'}$: the atoms of $\sigma W$ are either root atoms of
$W$ also covered by $W'$ (resolution is stable under the extension), or
atoms of $\embed{D}$ replacing the substituted variable, and
$\subs{K}{\embed{D}}{\roots{K}{\embed{D}}}$ by \Cref{lem:tr:resolve} with
$\roots{K}{\embed{D}}$ covered by $W'$: for \ruleref{app},
$\roots{\G}{D}\sqsubseteq\roots{\G}{C_y}\subseteq$ the resolution of the
substituted annotation, by tightness; for \ruleref{capp} directly, since
the instantiated source annotation contains $D$. Covering lifts to
subcapturing elementwise (\rruleref{sc-elem}, \rruleref{sc-mode},
\rruleref{sc-union}). At a modal wrapper, split by \rruleref{trans} into
\rruleref{boxed} (the inductive hypothesis under the shared lock) and
\rruleref{boxed-sat}, whose obligation asks each entry pair of the
substituted $\sigma\Psi$ to be separated under $\LOCK[\Psi']$ and each
$\sigma\Phi$ entry to be kinded. A pair of unsubstituted entries is a pair
of single-variable root entries whose variables also carry entries in
$\Psi'$ at covering modes, so \rruleref{sep-lock} and \rruleref{sep-sc}
close it. A pair involving the substituted entry $\mu\,\embed{D}$ resolves
upward by \Cref{lem:tr:resolve}; each of its root atoms is covered by an
entry variable of $\Psi'$ (tightness and domain honesty for \ruleref{app},
the charged witness for \ruleref{capp}), so \rruleref{sep-lock} pairs it
with the other entry's variable, \rruleref{sep-union} assembles the atoms,
and \rruleref{sep-sc} descends to the entries. A $\sigma\Phi$ entry is
either preserved in $\Phi'$ (\rruleref{k-imm}) or, when the substitution
instantiated its variable, kinded through \Cref{lem:tr:kind} exactly as in
the \ruleref{cfun} case of \Cref{lem:tr:subtype}.
\end{proof}

\paragraph{L6: Subtyping preservation.}
\begin{lemma}[Subtyping preservation]\label{lem:tr:subtype}
If $\subs{\G}{T_1}{T_2}$ then $\subs{\embed{\G}^{\LOCK}}{\eemb{T_1}}{\eemb{T_2}}$.
\end{lemma}
\begin{proof}
By induction on the surface subtyping derivation.
\ruleref{refl}, \ruleref{trans}, \ruleref{tvar} map to their core namesakes;
\ruleref{top} maps to \rruleref{top}, whose purity premise the translated
shape satisfies. \ruleref{capt}: the IH and \Cref{lem:tr:sc} feed
\rruleref{capt}, which carries the permission $m$ rigidly on both sides
exactly as the source rule does.

The three congruences share one pattern; we spell out \ruleref{fun} and
note the deltas. Since term domains are invariant, the source relates
\[
\subs{\G}{(\forall(x:T)U_1)\capt C_1}{(\forall(x:T)U_2)\capt C_2}
\]
from $\subs{\G}{C_1}{C_2}$ and $\subs{(\G,x:T)}{U_1}{U_2}$. The translated
shapes share the binder prefix
$\forall[\cstar<:\top]\forall(x:\embed{\anyinst{T}{\set{\cstar}}})$ on the
nose, so \rruleref{cfun} (bounds $\top$, \rruleref{sb-top}) and
\rruleref{fun} (domain by \rruleref{refl}) reduce the claim to the
codomains, two capturing modal types. Their \rruleref{capt} splits into the
annotation step and the modal step. Annotation:
$R_{A_1}\sqsubseteq R_{A_2}$ by \Cref{lem:tr:rootmono} on
$\subs{\G}{C_1}{C_2}$ (the domain contribution $C_T$ is shared), and
covering lifts to $\subs{}{W_{A_1}}{W_{A_2}}$ elementwise
(\rruleref{sc-elem}, \rruleref{sc-mode}, \rruleref{sc-union}; the parameter
atoms are shared). Modal: split by \rruleref{trans},
\[
[\Psi_{A_1},\Phi_{A_1}]\eemb{U_1}
\ \overset{\rruleref{boxed}}{<:}\
[\Psi_{A_1},\Phi_{A_1}]\eemb{U_2}
\ \overset{\rruleref{boxed-sat}}{<:}\
[\Psi_{A_2},\Phi_{A_2}]\eemb{U_2},
\]
the first step by the IH under the shared lock, the second discharging
$\sat{(\embed{\G}^{\LOCK},\LOCK[\Psi_{A_2},\Phi_{A_2}])}{[\Psi_{A_1},\Phi_{A_1}]}$:
every entry of $\Psi_{A_1}$ is a single-variable root entry whose variable
carries an entry of $\Psi_{A_2}$ at covering modes ($R_{A_1}\sqsubseteq
R_{A_2}$, parameter entries shared), so each pair is paid by
\rruleref{sep-lock} followed by \rruleref{sep-sc} on both sides
(\rruleref{sc-elem}, \rruleref{sc-mode}).

\ruleref{tfun}: the bound step $\subs{\G}{S_2}{S_1}$ translates by the IH
and feeds \rruleref{tfun}; bounds do not occur in $\Psi$, $\Phi$, or $W$,
so the codomain argument is verbatim.
\ruleref{cfun}: the bound step translates by \Cref{lem:tr:sc} (capture-set
bounds), \rruleref{sb-top} (mutability bounds), or vacuously (both $\top$),
and feeds \rruleref{cfun}. The lock argument gains a $\Phi$ leg when
$B_1=\RO$: the obligation $\typs{}{\set{c}}{\RO}$ under
$\LOCK[\Psi_{A_2},\Phi_{A_2}]$ holds by \rruleref{k-imm} when $B_2=\RO$
(then $\Phi_{A_2}=(\set{c}:\RO)$), and otherwise $B_2$ is a capture set
that \ruleref{sb-kind} certified at kind $\RO$, so \rruleref{sc-cvar}
resolves $\set{c}$ to $\embed{B_2}$ and \rruleref{k-sc} concludes with
\Cref{lem:tr:kind}.
\end{proof}

\paragraph{L7: Separation.}
The surface leaf \ruleref{sep-root} separates any two distinct roots by
fiat; the core has no fiat, and the translation pays every pair from a lock,
from consume authority, or from read-only kinding. The next lemma extracts
from a surface separation everything the payment needs, the following
definition names the payments, and the preservation lemma reconstructs the
core derivation; which payments are available at a site is the business of
\Cref{lem:tr:invariant}.

\begin{lemma}[Root characterization]\label{lem:tr:sepchar}
If $\sep{\G}{C_1}{C_2}$, then for every pair of atoms
$\mu\,p\in\roots{\G}{C_1}$ and $\nu\,q\in\roots{\G}{C_2}$: either $p\neq q$,
or both atoms are read-only in the sense of \Cref{lem:tr:kindroot}.
The same holds for $\disj{\G}{C_1}{C_2}$ with the second alternative
dropped, and for $\seqcomp{\G}{C_1}{C_2}$ restricted to the
$\CONSUME$-marked atoms of $\roots{\G}{C_1}$: a consumed root of $C_1$ is
not a root of $C_2$.
\end{lemma}
\begin{proof}
By induction on the respective derivation.
\ruleref{sep-empty} is vacuous; \ruleref{sep-symm} and \ruleref{sep-union}
preserve the condition (it is symmetric and resolution distributes).
\ruleref{sep-root}: the two sides resolve to the two distinct atoms.
\ruleref{sep-ro}: by \Cref{lem:tr:kindroot} on its kinding and access-only
premises, every atom of both resolutions is read-only.
\ruleref{sep-sc}: \Cref{lem:tr:rootmono} covers each atom of the shrunken
side at a weaker mode, and both alternatives are closed downward under
$\preceq$ (a $\CONSUME$ atom is covered only by a $\CONSUME$ atom).
Disjointness drops \ruleref{sep-ro}, and with it the second alternative.
For sequential composition: \ruleref{seq-access-only} has no
$\CONSUME$-marked atom in $\roots{\G}{C_1}$ by definition of access-only;
\ruleref{seq-sep} defers to the separation claim, whose second alternative
is unavailable to a $\CONSUME$-marked atom; \ruleref{seq-sc} uses
\Cref{lem:tr:rootmono} as above; \ruleref{seq-union} splits.
\end{proof}

\begin{definition}[Payable pairs]\label{def:tr:payable}
In a core context $K$, a pair of atoms $\mu\,p$, $\nu\,q$ over
\emph{distinct} capture variables $p\neq q$ is \emph{payable} when
$\sep{K}{\set{\mu\,p}}{\set{\nu\,q}}$ is derivable. A set of atoms is
\emph{pairwise payable} when all its distinct-variable pairs are. Payability
is closed under context extension (derivations weaken), under kills that
touch neither variable (\Cref{lem:tr:kill}), and under mode weakening: if
the pair at $(\mu',\nu')$ is payable and $\mu\preceq\mu'$,
$\nu\preceq\nu'$, then so is the pair at $(\mu,\nu)$, by
\rruleref{sep-sc} and \rruleref{sc-mode} on both sides.
\end{definition}

\begin{lemma}[Separation preservation]\label{lem:tr:sep}
If $\sep{\G}{C_1}{C_2}$ and every distinct-variable pair of atoms of
$\embed{\roots{\G}{C_1}}\times\embed{\roots{\G}{C_2}}$ is payable in
$\embed{\G}^{\LOCK}$, then
$\sep{\embed{\G}^{\LOCK}}{\embed{C_1}}{\embed{C_2}}$.
\end{lemma}
\begin{proof}
Write $P_i=\embed{\roots{\G}{C_i}}=\roots{\embed{\G}^{\LOCK}}{\embed{C_i}}$.
First derive $\sep{}{P_1}{P_2}$ atom by atom: a distinct-variable pair is
payable by hypothesis; a shared-variable pair is read-only on both sides by
\Cref{lem:tr:sepchar}, and \rruleref{sep-ro} separates the two singletons,
whose access-only premises hold because a read-only atom is not
$\CONSUME$-marked, and whose kinding premises hold by \rruleref{k-ro} for an
$\RO$-marked atom and by \rruleref{k-imm} for an $\RO$-bounded one, reading
the $\Phi$ entry of the binder's frame (\Cref{app:tr:sets}).
\rruleref{sep-union} and \rruleref{sep-symm} assemble the atoms into
$\sep{}{P_1}{P_2}$. Finally \rruleref{sep-sc} on both sides (through
\rruleref{sep-symm}) descends along $\subs{}{\embed{C_i}}{P_i}$
(\Cref{lem:tr:resolve}).
\end{proof}

\begin{corollary}[Disjointness preservation]\label{lem:tr:disj}
If $\disj{\G}{C_1}{C_2}$, $\consumable{\G}{C_1}$, and
$\consumable{\G}{C_2}$, then
$\disj{\embed{\G}^{\LOCK}}{\embed{C_1}}{\embed{C_2}}$.
\end{corollary}
\begin{proof}
By \Cref{lem:tr:sepchar} (disjointness clause) the two root sets share no
variable, and by \Cref{lem:tr:kind} all their variables are
$\CONSUME$-bound in $\embed{\G}^{\LOCK}$. \rruleref{sep-consumable}
separates each pair at its stated modes, staying inside the disjointness
fragment; \rruleref{sep-union}, \rruleref{sep-symm}, and \rruleref{sep-sc}
along \Cref{lem:tr:resolve} assemble and descend as in \Cref{lem:tr:sep}.
No lock and no read-only leaf is consulted.
\end{proof}

\paragraph{L8: Sequential composition.}
\begin{lemma}[Sequential composition preservation]\label{lem:tr:seq}
If $\seqcomp{\G}{C_1}{C_2}$ and every distinct-variable pair of a
$\CONSUME$-marked atom of $\embed{\roots{\G}{C_1}}$ with an atom of
$\embed{\roots{\G}{C_2}}$ is payable in $\embed{\G}^{\LOCK}$, then
$\seqcomp{\embed{\G}^{\LOCK}}{\embed{C_1}}{\embed{C_2}}$.
\end{lemma}
\begin{proof}
By reconstruction, not by induction on the source derivation.
\rruleref{seq-sc} traces $\embed{C_1}$ up to $P_1=\embed{\roots{\G}{C_1}}$
(\Cref{lem:tr:resolve}), and \rruleref{seq-union} splits $P_1$ into atoms.
An atom without $\CONSUME$ mark is access-only, so \rruleref{seq-access-only}
closes it. A $\CONSUME$-marked atom $\CONSUME\,p$ is, by
\Cref{lem:tr:sepchar}, over a variable that is not a variable of
$\roots{\G}{C_2}$, so all its pairs with the atoms of $P_2$ are
distinct-variable and payable by hypothesis; assembling them by
\rruleref{sep-union}/\rruleref{sep-symm} and descending to $\embed{C_2}$ by
\rruleref{sep-sc} gives $\sep{}{\set{\CONSUME\,p}}{\embed{C_2}}$, and
\rruleref{seq-sep} closes the atom.
\end{proof}

\paragraph{L9: The ambient invariant.}
The payability hypotheses of \Cref{lem:tr:sep,lem:tr:seq} are discharged
uniformly: at every site of the main induction, the roots of the site's own
charge and witnesses are pairwise payable. This is the translated image of
\ruleref{sep-root}'s fiat, cut down from all pairs of roots in scope to the
pairs the site can actually mention.

A \emph{site} of a canonical derivation is a subderivation occurrence; its
\emph{site root set} $\rho(s)$ is $\roots{\G_s}{C_s}$ for its context
$\G_s$ and use set $C_s$, together with $\roots{\G_s}{D}$ for the witnesses
$D$ of its own \ruleref{app} and \ruleref{consume-app} premises
(\ruleref{capp} witnesses are already charged).

\begin{lemma}[Ambient invariant]\label{lem:tr:invariant}
At every site $s$ of a canonical source derivation under the ambient
convention, $\embed{\rho(s)}$ is pairwise payable in the site's translated
context $\embed{\G_s}^{\LOCK}$.
\end{lemma}
\begin{proof}
By induction on the path from the derivation's root to $s$, using the
closure of payability under context extension and mode weakening
(\Cref{def:tr:payable}).

\emph{Root.} The ambient context declares no capture variables, so every
resolution over it is empty and $\rho$ is empty.

\emph{Descending within a rule instance.} Every premise's use set is a
subset of the conclusion's, equal to it, or below it in subcapturing
(\ruleref{sub}), so \Cref{lem:tr:rootmono} covers the premise's roots by
the conclusion's at weaker modes, and mode closure transfers payability.
Witnesses stay inside $\rho$: an \ruleref{app} witness is tight,
$\roots{\G}{D}\sqsubseteq\roots{\G}{C_y}=\roots{\G}{\set{y}}$
(\Cref{def:tr:canon}), and \ruleref{consume-app} and \ruleref{capp} charge
their witnesses outright.

\emph{Crossing an abstraction binder.} Here the invariant is
re-established wholesale by the binder's own lock, with no help from the
enclosing site. For a term binder with arrow $A$, the body's use set is
$C\cup\set{x}$ (plus $\set{\cstar,\CONSUME\,\cstar}$ in the consume case),
and
\[
\roots{\G'}{C\cup\set{x}}
= \roots{\G}{\mfree{C}} \cup \set{\cstar\text{-atoms}} \cup \roots{\G}{\mfree{C_T}}
= R_A \cup \set{\cstar\text{-atoms}},
\]
since $\set{x}$ resolves through its declared
$\embed{\anyinst{T}{\set{\cstar}}}$. Every atom of the body's site root set
therefore lies, at its occurring mode, in an entry of $\Psi_A$: the
$R_A$-atoms in their root entries, the $\cstar$-atoms in the parameter
entry, which carries both modes exactly when the consume case charges both.
A distinct-variable pair inhabits two distinct entries, so
\rruleref{sep-lock} on the frame's $\LOCK[\Psi_A,\Phi_A]$
(\Cref{app:tr:ctx}) followed by \rruleref{sep-sc} on both sides
(\rruleref{sc-elem} into the entries) pays it. Capture and type binders are
the same computation with $\roots{\G'}{C\cup\set{c}}=R_A\cup\set{c}$ and
$\roots{\G'}{C}=R_A$ respectively.

\emph{Crossing into a continuation.} Let the site be the continuation of a
\ruleref{let} with head charge $C_1$ and continuation charge $C_2$. Pairs
within $\roots{\G}{C_2}$ are inherited: $C_2\subseteq C_1\cup C_2$, and the
paying derivations extend to the continuation's context. They also survive
the kill that the translated \rruleref{let} or \rruleref{unpack} applies:
the killed variables are the $\CONSUME$-marked roots of the head's
translated charge, which \Cref{lem:tr:sepchar} (sequencing clause) on the
source premise $\seqcomp{\G}{C_1}{C_2}$ keeps out of the variables of
$\roots{\G}{C_2}$, and a payment touches no authority beyond its two
variables (\Cref{def:tr:payable}). For a $\FRESH$-carrying let, the
continuation's roots additionally contain the witness atoms
$c_i,\CONSUME\,c_i$. A pair of witness atoms over distinct variables is
paid by \rruleref{sep-consumable}, whose modes are unconstrained. A pair of
a witness atom with an atom $\nu\,q$ of $\roots{\G}{C_2}$ is paid by the
ownership lock $\LOCK[\Psi_w,\Phi_w]$ of the translated \rruleref{unpack}:
the translation instantiates the rule's continuation slot at the root
closure of the continuation's use set (\Cref{app:tr:mainproof}), so
$\Psi_w$'s two entries are the mode-closed witness tuple and a root-level
set containing $\nu\,q$; \rruleref{sep-lock} separates the entries and
\rruleref{sep-sc} descends to the two atoms.
\end{proof}

\begin{corollary}[Satisfaction discharge]\label{cor:tr:sat}
At each translated elimination (\Cref{app:tr:terms}), the satisfaction
obligation $\sat{\embed{\G}^{\LOCK}}{[\sigma\Psi_A,\sigma\Phi_A]}$ of the
instantiated lock is derivable at the site.
\end{corollary}
\begin{proof}
By canonicity (declared subjects) the lock is the one manufactured at the
callee's declared type, with entries the parameter entry, instantiated to
$\embed{D}$ (at both modes for a consumer), and the root entries of $R_A$.
The obligation is pairwise separation of the entries, plus the $\Phi$ legs.

\emph{Root-entry pairs.} Two root entries carry distinct variables $p\neq q$
of $R_A$ by construction. Both belong to $\embed{\rho(s)}$: the
$\mfree{C}$-part of $R_A$ equals $\roots{\G}{\set{x}}$ for the charged
subject (declared annotation), and the $\mfree{C_T}$-part is covered by
$\roots{\G}{C_y}=\roots{\G}{\set{y}}$ by domain honesty
(\Cref{app:tr:conv}). \Cref{lem:tr:invariant} pays each atom pair, and
\rruleref{sep-union} with \rruleref{sep-sc} assembles atoms into the
entries. For \ruleref{tapp} these are the only pairs, and no source premise
is consumed, matching the rule, which has none.

\emph{Witness pairs.} A pair of the instantiated parameter entry
$\mu\,\embed{D}$ with a root entry over $q$ resolves upward
(\Cref{lem:tr:resolve}) to atom pairs $(\mu'\,d,\nu\,q)$ with
$d$ a variable of $\roots{\G}{D}$. If $d\neq q$, the pair is payable:
$d$ lies in $\rho(s)$ by tightness (\ruleref{app}) or because the witness
is charged (\ruleref{capp}, \ruleref{consume-app}), and $q$ as above. If
$d=q$, the source separation premise $\sep{\G}{D}{|A|}$, whose right side
contains $\mfree{C\cup C_T}$ and hence resolves onto $R_A$, makes both
atoms read-only by \Cref{lem:tr:sepchar}, and \rruleref{sep-ro} pays the
singleton pair as in \Cref{lem:tr:sep}. \rruleref{sep-union},
\rruleref{sep-symm}, and \rruleref{sep-sc} assemble the atoms and descend
to the entries.

\emph{Mode legs.} $\Phi_A$ is empty except at an $\RO$-bounded capture
binder, where the obligation after $[c:=\embed{D}]$ is
$\typs{}{\embed{D}}{\RO}$, the translated image through \Cref{lem:tr:kind}
of the kinding $\typs{\G}{D}{\RO}$ that \ruleref{sb-kind} established when
the source narrowed the bound to $D$.
\end{proof}

\begin{remark}[Ownership is paid, not decreed]\label{rem:tr:ownership}
The continuation case of \Cref{lem:tr:invariant} rests on the ownership
lock that \rruleref{unpack} manufactures; this remark records why the rule
carries it and why its plain sequencing premise suffices to pay for it.
Core \rruleref{sep-consumable} separates two \emph{distinct consumable}
roots, but the continuation also needs separations between an owned root
and \emph{older} bindings: minimally, with $T_a$, $\textsf{rp}$, and
$\textsf{alloc}$ as in \Cref{app:tr:example}, an access-only abstraction
root against a fresh result bound later in the same body,
$f=\lambda(op:T_a).\ \LET fresh=\textsf{alloc}\,()\IN \textsf{rp}\,op\,fresh$, whose source
derivation uses $\sep{\G}{\set{op}}{\set{fresh}}$. No fiat ownership rule
can supply these: an axiom separating a $\top$-bounded consume root from
everything older conflates the root's binding time with its footprint's
\emph{provenance}: an unpack witness is young as a \emph{binding}, but its
footprint is the consumed evidence's, inherited verbatim, and an older
binding may alias it. With $\CONSUME\,c_a<:\top,\ a:\REF[T]\capt\set{c_a}$
in scope and $c_a$ live,
\[
\lambda[c_i<:\top]\,\lambda(g:\REF[T]\capt\set{c_i})\
\LET \<c_j,x\> = \<\set{c_a},a\>\IN t
\]
satisfies both natural guards ($c_j$ is $\top$-bounded, $c_i$ is older), yet
\rruleref{capp} may instantiate $c_i:=\set{c_a}$, an access-only view of the
still-live source that the kill at the unpack does not remove; the two
``separate'' sides share $c_a$'s cell. The mechanization refuted exactly
this rule and its ownership-at-birth invariant; freshness-at-creation is a
fact about \rruleref{alloc}, not about unpacking.

Ownership is instead a \emph{paid certificate}, in the same way the
function locks of \Cref{app:tr:types} pay for their parameter separations:
\rruleref{unpack} manufactures the lock
$\Psi_w$ from what its own premises prove. The
payment is the footprint model's \emph{witness-confinement bound}: every
location reachable from a pack's witness evidence is either covered
\emph{at consume mode} by the producing computation's charge $C_1$, or
allocated during that computation. A fresh location cannot lie in the
continuation charge's footprint, live before the computation runs; and a
$C_1$-consume-covered location shared with $C_2$'s footprint is a
consume-then-use conflict, which is precisely what $\seqcomp{\G}{C_1}{C_2}$
denies: its access-only leg denotes a consume-free charge, forcing every
witness fresh, and its separating leg denies the overlap outright. In
particular the plain sequencing premise suffices; a strengthened separation
premise, which an earlier revision adopted after refuting a discharge of the
lock through the premise's separation semantics alone (a fact about that
discharge route, not about the rule), is unnecessary, and would forbid a
continuation from re-touching an access-shared producer, e.g.\ calling the
same existential-returning function twice. Two aspects of how the lock is
\emph{instantiated} matter to \Cref{lem:tr:invariant}. The witness entry is
stored at both access modes: the confinement bound is a fact about the
witnesses' locations, indifferent to the mode at which the continuation
later touches them, and the $\CONSUME$-marked reading is what pays the
continuation's own consumption of a witness against its remaining charge.
The continuation entry is stored in root-resolved form, chosen by the
translation when it instantiates the rule; a \emph{rule-level} root
enrichment would instead strictly enlarge the recorded footprint and demand
a genuinely stronger premise. All of this is mechanized: the refutations,
the lock-manufacturing \rruleref{unpack} as the sole unpack rule of the
development, and the headline theorems (the fundamental theorem, adequacy,
immutability, standardization, and confluence), which check with no axioms
beyond propositional extensionality, choice, and quotient soundness.
\end{remark}

\paragraph{L10: Kill transport.}
The core \rruleref{let} and \rruleref{unpack} type their continuation in the
killed context $\G\ominus C_1$, whereas the induction produces the
translated continuation in the unkilled $\embed{\G}^{\LOCK}$. The bridge is
that translated derivations consult consume authority only at charged
roots, which sequencing keeps live.

\begin{definition}[Authority robustness]\label{def:tr:robust}
A \corecapybara{} derivation is \emph{authority-robust for} a set $X$ of
capture variables when every leaf that consults a binding's consume
authority, an $\accessible{}{\cdot}$ premise, a $\consumable{}{\cdot}$
premise, or a \rruleref{sep-consumable} instance, mentions only capture
variables that are in $X$ or bound within the derivation itself.
(\rruleref{sc-cvar} also reads a binding, but applies only to access-only
bindings, which no kill targets, so it needs no tracking; \rruleref{k-imm}
and \rruleref{sep-lock} read locks, not authority.)
\end{definition}

\begin{lemma}[Kill transport]\label{lem:tr:kill}
Let $\typ{C}{\G}{t}{E}$ in \corecapybara{} by a derivation authority-robust
for $X$, and let no variable killed by $\G\ominus C_1$ belong to $X$. Then
$\typ{C}{\G\ominus C_1}{t}{E}$, by a derivation authority-robust for the
same $X$; and likewise for every auxiliary judgment.
\end{lemma}
\begin{proof}
By induction on the derivation. The kill rewrites consume authorities and
nothing else: every payload, every lock entry, and every type is unchanged,
so each rule that does not consult authority re-applies verbatim on the
inductive hypotheses. The authority-consulting leaves survive because their
variables are in $X$ or derivation-local, hence not killed: an
$\accessible{}{\cdot}$ premise finds no $\KILLED$ binding among them, a
$\consumable{}{\cdot}$ premise and a \rruleref{sep-consumable} instance
find their $\CONSUME$ bindings intact, and \rruleref{sc-cvar} reads
access-only bindings, which the kill never touches.
\end{proof}

The main induction produces derivations that are authority-robust for
$X_s=\embed{\rho(s)}$, the translated site root set: every elimination's
subject is charged by the \emph{source} (surface \ruleref{var} charges
occurrences eagerly), every consumable premise concerns a charged or
witness set, and every \rruleref{sep-consumable} instance introduced by
\Cref{lem:tr:sep,lem:tr:seq,lem:tr:disj,lem:tr:invariant} pairs atoms of
$\rho(s)$. The let case of \Cref{app:tr:mainproof} verifies that the killed
roots avoid exactly this set.

\subsection{Proof of the Translation Theorem}\label{app:tr:mainproof}

We prove \Cref{thm:tr:main} by induction on the canonical source derivation,
one case per rule. Throughout, $C'$ denotes the translated use set, verified
root-covered by $\embed{C}$; unless stated, the translated result type is
$\eemb{T}$ by construction. Three recurring facts are used without further
comment. \emph{Stamping}: whenever
$\roots{\embed{\G}^{\LOCK}}{C'}\sqsubseteq W$ for a root-level capture set
$W$, also $\subs{\embed{\G}^{\LOCK}}{C'}{W}$, by \Cref{lem:tr:resolve}
followed by \rruleref{sc-elem}, \rruleref{sc-mode}, and \rruleref{sc-union}
on the covering; this is what matches a translated body charge against the
stated annotation $W_A$ and a translated use set against a root-closure
slot. \emph{Robustness}: each case's derivation is authority-robust for
$\embed{\rho(s)}$ (\Cref{def:tr:robust}), by inspection of its leaves.
\emph{Consumed charges}: every $\CONSUME$-marked atom of a translated use
set either belongs to a translated \rruleref{pack},
\rruleref{consumer-app}, \rruleref{dealloc}, or \rruleref{unpack} charge,
each guarded by a consumable premise translated through \Cref{lem:tr:kind},
or resolves from an annotation stamped from such a charge, so its variable
is $\CONSUME$-bound.

\paragraph{\ruleref{var}.}
$\embed{x}=x$. The translated binding is $x:m\,\embed{S}\capt\embed{C}$, so
\rruleref{var} gives
$\typ{\set{}}{\embed{\G}^{\LOCK}}{x}{m\,\embed{S}\capt\set{x}}$, which is
$\eemb{m\,S\capt\set{x}}$ for either $m$: source and core \ruleref{var}
preserve the binding's qualifier on the same plain singleton. Uses:
$\set{}$, trivially below.

\paragraph{\ruleref{readonly}.}
$\embed{x}=\RO\,x$. By the convention of \Cref{app:tr:terms} the subject is
reference-shaped, $x:\REF[T]\capt C\in\G$, so \rruleref{reader} gives
$\typ{\set{}}{\embed{\G}^{\LOCK}}{\RO\,x}{\RO\,\REF[\embed{T}]\capt\set{\RO\,x}}
=\eemb{\RO\,\REF[T]\capt\set{\RO\,x}}$ primitively, with no coercion.
Uses: $\set{}$, trivially below the source charge $\set{\RO\,x}$.

\paragraph{\ruleref{fresh}.}
$\embed{x}=\langle\embed{D_1},\ldots,\embed{D_n},x\rangle$. $x$'s declared type is
\[
[c_1:=\embed{D_1},\ldots,c_n:=\embed{D_n}]\,\embed{\freshinst{T}{\set{c_1},\ldots,\set{c_n}}}
\]
(\Cref{lem:tr:subst}), so \rruleref{pack} applies: each $\embed{D_i}$ is
access-only and consumable (\Cref{lem:tr:kind} on the rule's premises), and
the witnesses are pairwise disjoint, $\disj{}{\embed{D_i}}{\embed{D_j}}$ for
$i\neq j$, by \Cref{lem:tr:disj} on the rule's disjointness premises, whose
sides the rule's other premises make consumable.
The pack has type
$\eemb{T}=\EXCAP{c_1,\ldots,c_n}\embed{\freshinst{T}{\set{c_1},\ldots,\set{c_n}}}$ and charges
$\bigcup_i\embed{D_i}\cup\CONSUME\bigcup_i\embed{D_i}
=\embed{D\cup\CONSUME\,D}$ (\Cref{lem:tr:mode}), exactly the source charge.

\paragraph{\ruleref{sub}.}
By \Cref{lem:tr:subtype} on the type, \Cref{lem:tr:sc} on the use set, and
\rruleref{sub}; the side condition follows by \Cref{lem:tr:rootmono}.

\paragraph{\ruleref{abs}.}
For $\alpha=\epsilon$,
\[
\embed{\lambda(x:T)t}=\lambda[\cstar<:\top]\lambda(x:\embed{\anyinst{T}{\set{\cstar}}})\LOCK[\Psi_A,\Phi_A]\embed{t}.
\]
The source premise types $t$ at $C\cup\set{x}$, so the IH types $\embed{t}$
at a charge $C'$ with
$\roots{}{C'}\sqsubseteq\embed{\roots{\G'}{C\cup\set{x}}}
= \embed{R_A}\cup\set{\cstar}$, the atoms of $W_A$, under the binder's frame
(\Cref{app:tr:ctx}). \rruleref{lock} stamps that charge onto the codomain
modal as $([\Psi_A,\Phi_A]\eemb{U})\capt C'$, and \rruleref{sub} with
stamping widens $C'$ to the stated annotation $W_A$.
\rruleref{abs} abstracts $x$ and \rruleref{cabs} abstracts $\cstar$ (premise:
$\top$ is access-only); the value's outer capture $\set{}$ is widened to
$\embed{C}$ by \rruleref{sub}, giving $\eemb{(\forall(x:T)U)\capt C}$. For
$\alpha=\CONSUME$ a \rruleref{consumer} frame replaces
\rruleref{cabs}+\rruleref{abs}; its body charge additionally carries
$\set{\cstar,\CONSUME\,\cstar}$, the image of the source
$D=\set{c,\CONSUME\,c}$, and the stamping target is
$W_A=\set{\cstar,\CONSUME\,\cstar}\cup\embed{R_A}$. Uses: $\set{}$.

\paragraph{\ruleref{tabs}, \ruleref{cabs}.}
As for \ruleref{abs} with one binder fewer:
$\embed{\lambda[X<:S]t}=\lambda[X<:\embed{S}]\LOCK[\Psi_A,\Phi_A]\embed{t}$
types by \rruleref{tabs} around \rruleref{lock}, and
$\embed{\lambda[c<:B]t}=\lambda[c<:\embed{B}]\LOCK[\Psi_A,\Phi_A]\embed{t}$
by \rruleref{cabs} (its premise $\accessonly{\G}{\embed{B}}$ by
\Cref{lem:tr:kind}, with $\top$ access-only by definition) around
\rruleref{lock}. The body charges resolve into $\embed{R_A}\cup\set{c}$
resp.\ $\embed{R_A}$, and stamping widens them to $W_A$. Both charge
$\set{}$ at the conclusion, as their source rules do.

\paragraph{\ruleref{app}.}
$\embed{x\,y}=\LET x_1=x[\embed{D}]\IN\LET x_2=x_1\,y\IN\UNLOCK\,x_2$, with
$D$ the canonical witness. By canonicity the callee premise concludes at the
declared type, translated
\[
m\big(\forall[\cstar<:\top]\forall(z:\embed{\anyinst{T}{\set{\cstar}}})\,([\Psi_A,\Phi_A]\eemb{U})\capt W_A\big)\capt\embed{C}.
\]
\rruleref{sub} with \rruleref{cfun} and \rruleref{sb-top} narrows the bound
$\top$ to $\embed{D}$, and \rruleref{capp} instantiates
$\cstar:=\embed{D}$; its premise $\accessonly{\G}{\embed{D}}$ is the image
of the source premise through \Cref{lem:tr:kind}, and it charges $\set{}$.
The domain of the instantiated arrow is
$[\cstar:=\embed{D}]\,\embed{\anyinst{T}{\set{\cstar}}}
=\embed{\anyinst{T}{D}}$ on the nose: by root-directedness the sole $\ANY$
sits in the domain's top-level capture set, which is translated elementwise
(\Cref{lem:tr:subst}). The IH on the argument premise types $y$ at exactly
this type, so \rruleref{app} applies, charging $\set{x_1}$ and producing
$[z:=y][\cstar:=\embed{D}]\big(([\Psi_A,\Phi_A]\eemb{U})\capt W_A\big)$,
which \Cref{lem:tr:cast} casts to the translation of the instantiated
source codomain. $\UNLOCK\,x_2$ (charge $\set{}$) opens the modal,
discharging its satisfaction by \Cref{cor:tr:sat}, and returns
$\eemb{[z:=y]U}$. The administrative lets sequence trivially:
\rruleref{tapp}-style inert charges make the first head $\set{}$
(\rruleref{seq-access-only}) and the unlock continuation $\set{}$
(\rruleref{seq-sep} via \rruleref{sep-empty} and \rruleref{sep-symm}).
Uses: $\set{x_1}$, whose roots are the atoms of the instantiated $W_A$,
covered by $\embed{\roots{\G}{\set{x,y}}}$: the $R_A$-part equals
$\roots{\G}{\set{x}}$ (declared annotation) and the witness part satisfies
$\roots{\G}{D}\sqsubseteq\roots{\G}{\set{y}}$ (tightness).

\paragraph{\ruleref{consume-app}.}
$\embed{x\,y}=\LET x_2=x\,\langle\embed{D},y\rangle\IN\UNLOCK\,x_2$. The callee
translates to a consumer type. \rruleref{pack} forms
$\langle\embed{D},y\rangle:\EXCAP{\cstar}\embed{\anyinst{T}{\set{\cstar}}}$
(charge $\embed{D}\cup\CONSUME\,\embed{D}$; access-onliness and
consumability from the source premises through \Cref{lem:tr:kind}).
\rruleref{consumer-app} applies $x$, requiring
$\seqcomp{\G}{\embed{D}\cup\CONSUME\,\embed{D}}{\set{x}}$:
\rruleref{seq-union} splits it into the access leg, free by
\rruleref{seq-access-only}, and the consume leg
$\seqcomp{}{\CONSUME\,\embed{D}}{\set{x}}$, given by \Cref{lem:tr:seq} on
the source premise $\seqcomp{\G}{\CONSUME\,D}{\set{x}}$ with payability from
\Cref{lem:tr:invariant} ($D$ is charged at this site). \rruleref{unlock}
opens the codomain by \Cref{cor:tr:sat}, after \Cref{lem:tr:cast} casts the
substituted codomain. The charge
$\embed{D}\cup\CONSUME\,\embed{D}\cup\set{x}$ lies below the source charge
$\embed{\set{x,y}\cup D\cup\CONSUME\,D}$ by \rruleref{sc-elem}.

\paragraph{\ruleref{tapp}, \ruleref{capp}.}
$\embed{x[S']}=\LET x_1=x[\embed{S'}]\IN\UNLOCK\,x_1$: \rruleref{sub} with
\rruleref{tfun} narrows the bound from $\embed{S}$ to $\embed{S'}$ (using
$\subs{\G}{S'}{S}$ through \Cref{lem:tr:subtype}), \rruleref{tapp}
instantiates, and \rruleref{unlock} opens the lock, whose entries the type
substitution does not touch (\Cref{lem:tr:subst}); \Cref{cor:tr:sat}
discharges the satisfaction from the ambient invariant alone, and the
result is $[X:=\embed{S'}]\eemb{U}=\eemb{[X:=S']U}$.
$\embed{x[D]}=\LET x_1=x[\embed{D}]\IN\UNLOCK\,x_1$: \rruleref{sub} narrows
the translated bound to $\embed{D}$ (\Cref{lem:tr:sc} for a capture-set
bound, \rruleref{sb-top} for a translated mutability bound), \rruleref{capp}
instantiates $c:=\embed{D}$ (access-only premise through
\Cref{lem:tr:kind}), \Cref{lem:tr:cast} casts the substituted codomain, and
\rruleref{unlock} opens the lock by \Cref{cor:tr:sat}, whose witness pairs
are payable because the revised \ruleref{capp} charges $D$; when $B=\RO$
the mode leg is the image of the \ruleref{sb-kind} kinding
(\Cref{lem:tr:kind}). In both cases every translated charge is $\set{x}$,
directly corresponding to the source charge,
and the administrative let sequences by \rruleref{seq-access-only}.

\paragraph{\ruleref{let}.}
Let $C_1,C_2$ be the source charges and $C_1',C_2'$ the translated ones from
the IHs. Both forms instantiate their core rule with the head slot $C_1'$
and the \emph{root-closure} continuation slot
$\widehat{C_2}=\embed{\roots{\G}{C_2}}$ (plus, for \rruleref{unpack}, the
witness uses the rule adds); the continuation's typing is widened from
$C_2'$ to $\widehat{C_2}$ by \rruleref{sub} and stamping, and the closure
adds no roots, so the side condition for the total
$C_1'\cup\widehat{C_2}$ still reads
$\sqsubseteq\embed{\roots{\G}{C_1\cup C_2}}$.
The sequencing premise $\seqcomp{}{C_1'}{\widehat{C_2}}$ is derived as in
\Cref{lem:tr:seq}: \rruleref{seq-sc} traces $C_1'$ to its roots, covered by
$\embed{\roots{\G}{C_1}}$; access-only atoms close by
\rruleref{seq-access-only}; and a $\CONSUME$-marked atom is, by
\Cref{lem:tr:sepchar} on the source premise $\seqcomp{\G}{C_1}{C_2}$, over
a variable outside $\roots{\G}{C_2}$, so its pairs are payable
(\Cref{lem:tr:invariant}) and \rruleref{seq-sep} closes it.

The continuation's IH derivation lives in the unkilled context, and
\Cref{lem:tr:kill} places it under the kill: it is authority-robust for
$X=\embed{\roots{\G}{C_2}}$ together with the continuation's witnesses and
locally bound roots, and the killed variables, the $\CONSUME$-marked roots
of $C_1'$, avoid $X$: they avoid $\roots{\G}{C_2}$ by
\Cref{lem:tr:sepchar} as above, they avoid the witnesses because a witness
is tight or charged, hence inside $\roots{\G}{C_2}$-coverage, and local
roots are bound after the kill.

When $T$ is $\FRESH$-free the translated term is
$\LET x=\embed{t}\IN\embed{u}$, typed by \rruleref{let}; the continuation's
result type $\eemb{U}$ and every payload pass through the kill unchanged,
since $\ominus$ rewrites authorities only. When $T$ carries $\FRESH$,
$\embed{\LET x=t\IN u}=\LET\langle c_1,\ldots,c_n,x\rangle=\embed{t}\IN\embed{u}$,
typed by \rruleref{unpack}. The match is exact: \rruleref{unpack} binds each
$\CONSUME\,c_i<:\top$ and adds
$\set{c_1,\ldots,c_n,\CONSUME\,c_1,\ldots,\CONSUME\,c_n}$ to the
continuation's use set, the images of the source \ruleref{let}'s
$\CONSUME\,c_i$ bindings and its body charge $D\cup\CONSUME\,D$; the
consumed roots of $C_1'$ are consumable by the consumed-charges fact. The
rule pushes the ownership lock
$\Psi_w=\set{c_1,\ldots,c_n,\CONSUME\,c_1,\ldots,\CONSUME\,c_n}\,,\
\widehat{C_2}$, whose root-level continuation entry is exactly what the
continuation case of \Cref{lem:tr:invariant} reads; the lock imposes no
obligation on the continuation (it is read, not proved, at the binder).

\paragraph{State, parallelism, and conditionals.}
\ruleref{alloc}, \ruleref{write}, \ruleref{dealloc}, and \ruleref{if} map to
their core namesakes with the premises translated by \Cref{lem:tr:kind};
\ruleref{read} inserts the reader as in \Cref{app:tr:terms}, its
administrative let sequencing by \rruleref{seq-access-only} (a reader value
charges nothing). \ruleref{par} maps to \rruleref{par}, whose separation
premise is \Cref{lem:tr:sep} on the source premise, with payability from
\Cref{lem:tr:invariant}: both branch charges are part of the site's charge.
Charges are the translated images of the source charges throughout.

This completes the induction, establishing \Cref{thm:tr:main}.\qed

\subsection{Example: \texttt{runParallel}}\label{app:tr:example}

We trace the \lstinline|runParallel| pattern of \Cref{sec:informal} through the
translation, exercising the lock manufacture of \Cref{app:tr:types} and the
satisfaction discharge of \Cref{cor:tr:sat}. Use a type variable $X$ (declared
$X<:\top$) for the resource shape, take $\TUNIT$ and $()$ as a base type and its
value as in the core calculus, and abbreviate $T_a=X\capt\set{\ANY}$, a
resource whose capture is a fresh $\ANY$. Take
\[
T_{\textsf{rp}}=\forall(op_1:T_a)\,\big((\forall(op_2:T_a)\TUNIT)\capt\set{op_1}\big),
\]
the curried type of \lstinline|runParallel|: applied to $op_1$ it returns a
closure that has captured $op_1$ and awaits $op_2$, and running the two together
requires $op_2$ separate from $op_1$. The context declares
$\textsf{rp}:T_{\textsf{rp}}$ and a primitive
$\textsf{alloc}:\forall(y:\TUNIT)(X\capt\set{\FRESH})$ producing
fresh resources. The source program is
\[
\LET f_1=\textsf{alloc}\,()\IN \LET f_2=\textsf{alloc}\,()\IN \textsf{rp}\,f_1\,f_2.
\]
Each $\LET$ binds a fresh consumable root; write $c_1,c_2$ for them, so
$f_1:X\capt\set{c_1}$, $f_2:X\capt\set{c_2}$ with $c_1\neq c_2$. The first
application $\textsf{rp}\,f_1$ carries the trivial separation
$\sep{\G}{\set{f_1}}{|T_{\textsf{rp}}|}=\sep{\G}{\set{f_1}}{\set{}}$; the second,
$(\textsf{rp}\,f_1)\,f_2$, carries the real one $\sep{\G}{\set{f_2}}{\set{f_1}}$, since the
residual callee $(\forall(op_2:T_a)\TUNIT)\capt\set{f_1}$ has footprint
$\set{f_1}$. \ruleref{sep-root} discharges it, $c_1\neq c_2$.

\paragraph{Translated type.}
By \Cref{app:tr:types}, $\embed{T_{\textsf{rp}}}$ manufactures one lock per arrow:
\[
\begin{aligned}
\embed{T_{\textsf{rp}}}=
\forall[c_{\star1}{<:}\top]\,\forall(op_1{:}X\capt\set{c_{\star1}})\,\Big(
&[\Psi_1,\emptyset]\,\big(\forall[c_{\star2}{<:}\top]\,\forall(op_2{:}X\capt\set{c_{\star2}})\\[-2pt]
&\quad([\Psi_2,\emptyset]\,\TUNIT)\capt\set{c_{\star2},c_{\star1}}\big)\capt\set{op_1}\Big)\capt\set{c_{\star1}},
\end{aligned}
\]
\[
\text{with}\quad
\Psi_1=\set{c_{\star1}}
\qquad\text{and}\qquad
\Psi_2=\set{c_{\star2}},\ \set{c_{\star1}}.
\]
The outer lock $\Psi_1$ has a single entry, the parameter root: the outer
arrow's annotation is empty, so $R_{A_1}=\set{}$, and a one-entry lock
demands no pairs; its satisfaction is vacuous. The inner arrow captures
$op_1$, so $R_{A_2}=\roots{\G}{\set{op_1}}=\set{c_{\star1}}$, and
$\Psi_2$ records that $op_2$'s root $c_{\star2}$ is separate from $op_1$'s
root $c_{\star1}$: the entries are individual roots, resolved once, at the
type former, and a body-internal fiat pair is read off them by
\rruleref{sep-lock} directly. This is the translated content of
``supplying $op_2$ requires it separate from $op_1$''. The codomain
annotations are the same root sets read as capture sets,
$W_{A_1}=\set{c_{\star1}}$ and $W_{A_2}=\set{c_{\star2},c_{\star1}}$; a
user-written \lstinline|runParallel|, whose inner body applies $op_2$ and
charges $\set{op_1,op_2}$, stamps that charge onto them by resolution
(\Cref{app:tr:mainproof}), so the manufactured arrow is inhabited by an
actual abstraction and not only by the context primitive $\textsf{rp}$.

\paragraph{Translated term.}
The nested application is administratively normalized, each
$\embed{x\,y}$ expanding as in \Cref{app:tr:terms}:
\[
\begin{aligned}
&\LET\langle c_1,f_1\rangle=\embed{\textsf{alloc}\,()}\IN
 \LET\langle c_2,f_2\rangle=\embed{\textsf{alloc}\,()}\IN{}\\
&\quad\LET g=\big(\LET a_1{=}\textsf{rp}[\set{f_1}]\IN\LET a_2{=}a_1\,f_1\IN\UNLOCK\,a_2\big)\IN{}\\
&\quad\LET b_1{=}g[\set{f_2}]\IN\LET b_2{=}b_1\,f_2\IN\UNLOCK\,b_2.
\end{aligned}
\]
The two $\LET\langle c_i,f_i\rangle$ are \rruleref{unpack}s of the fresh results,
binding $\CONSUME\,c_i<:\top$. The capture applications instantiate
$c_{\star1}:=\set{f_1}$ and $c_{\star2}:=\set{f_2}$, narrowing $\top$ by
\rruleref{sb-top} first; the witnesses are the canonical (tight) ones, the
arguments' declared capture sets.

\paragraph{Satisfaction discharge.}
Each $\UNLOCK$ discharges the satisfaction of the instantiated lock.
\begin{itemize}
\item $\UNLOCK\,a_2$: the instantiated $\Psi_1[c_{\star1}{:=}\set{f_1}]
  =(\set{f_1})$ has a single entry, so \rruleref{sat} demands no pair. This
  is the image of the trivial source separation against the empty spine.
\item $\UNLOCK\,b_2$: the instantiated
  $\Psi_2[c_{\star1}{:=}\set{f_1}][c_{\star2}{:=}\set{f_2}]
  =(\set{f_2},\set{f_1})$ requires $\sep{\G}{\set{f_2}}{\set{f_1}}$. Both
  entries resolve upward to the distinct $\top$-bounded consumable roots
  $c_2,c_1$ (\Cref{lem:tr:resolve}), \rruleref{sep-consumable} gives
  $\sep{\G}{\set{c_2}}{\set{c_1}}$, and \rruleref{sep-sc} descends to the
  entries. This is \Cref{cor:tr:sat} at the site: the witness pair is
  distinct by the source premise $\sep{\G}{\set{f_2}}{\set{f_1}}$
  (\Cref{lem:tr:sepchar}) and payable by the ambient invariant
  (\Cref{lem:tr:invariant}), here through the consume authority of the two
  unpacked roots.
\end{itemize}
Had the call instead passed two views of a single resource $f$, both
applications would instantiate the same $\set{f}$, and the inner obligation would
be $\sep{\G}{\set{f}}{\set{f}}$, underivable because the two roots coincide. The
manufactured lock thus rejects exactly the racy call, as \Cref{sec:informal}
requires.
 
\end{document}